\documentclass[cmp]{svjour}
\usepackage{amssymb, amsmath,latexsym}
\usepackage{bm}
\usepackage[dvipdfmx]{graphicx}
\usepackage{multirow}

\def\QED{\mbox{\rule[0pt]{1.5ex}{1.5ex}}}
\def\endproof{\hspace*{\fill}~\QED\par\endtrivlist\unskip}

 \newenvironment{proofof}[1]{\vspace*{5mm} \par \noindent
         \quad{\it Proof of #1:\hspace{2mm}}}{\endproof}

\def\Sp{\mathop{\rm Sp}\nolimits}

\def\argmax{\mathop{\rm argmax}}
\def\argmin{\mathop{\rm argmin}}
\def\mix{\mathop{\rm mix}}
\def\complex{\mathbb{C}}
\def\real{\mathbb{R}}
\def\Z{\mathbb{Z}}

\def\SU{\mathop{\rm SU}}
\def\SO{\mathop{\rm SO}}
\def\su{\mathop{\mathfrak{su}}}

\def\U{\mathop{\rm U}}
\def\cov{\mathop{\rm cov}}
\def\odd{\mathop{\rm odd}}
\def\even{\mathop{\rm even}}
\def\sgn{\mathop{\rm sgn}}

\def\ce{\mathop{\rm ce}}
\def\se{\mathop{\rm se}}

\newcommand{\bR}{\mathbb{R}}

\newcommand{\bZ}{\mathbb{Z}}

\def\rE{{\rm E}}

\def\Pr{{\rm Pr}}

\def\rank{{\rm rank}}

\newcommand{\Tr}{{\rm Tr}\,}

\makeatletter
\newcommand{\lleq}{\mathrel{\mathpalette\gl@align<}}
\newcommand{\ggeq}{\mathrel{\mathpalette\gl@align>}}
\newcommand{\gl@align}[2]{
\vbox{\baselineskip\z@skip\lineskip\z@
\ialign{$\m@th#1\hfil##\hfil$\crcr#2\crcr{}_{{}_{(=)}}\crcr}}}
\makeatother

\def\Label#1{\label{#1}\ [\ \text{#1}\ ]\ }
\def\Label{\label}

\begin{document}

\title{{Fourier Analytic Approach to Quantum Estimation of Group Action}}
\titlerunning{Fourier Analytic Approach to Quantum Estimation}

\author{Masahito Hayashi$^{1,2}$}
\institute{$^{1}$~Graduate School of Mathematics, Nagoya University, Japan. \\
$^{2}$~Centre for Quantum Technologies, National University of Singapore, Singapore. \\
\email{masahito@math.nagoya-u.ac.jp}}
\authorrunning{Masahito Hayashi}

\date{Received:}

\maketitle

\begin{abstract}
This article proposes a unified method to estimation of group action
by using the inverse Fourier transform of the input state. 
The method provides optimal estimation for 
commutative and non-commutative group with/without energy constraint.
The proposed method can be applied to projective representations of
non-compact groups as well as of compact groups.
This paper addresses the optimal estimation 
of $\bR$, $\U(1)$, $\SU(2)$, $\SO(3)$, and $\bR^2$ with Heisenberg representation
under a suitable energy constraint.
\end{abstract}


\section{Introduction}
In quantum theory,
the reversible dynamics of a system is often described by an element in a projective unitary representation of a group.
In this case, the unitary acting on the real quantum system 
reflects important physical parameters.
Therefore, we can estimate these physical parameters by
estimating the true unitary among a given projective unitary representation of a group.
Indeed, it is known that estimation of unitary has a square speed up
over the state estimation in quantum case.
However, only the limited case of estimation of unitaries has been solved\cite{Bu,Lu,Ba,Chi,Ha1,Chi2,Im,Ha2}.
Other case of estimation of unitaries has not been solved
while their Fisher information has been calculated\cite{Im2}.
Indeed, several researchers consider that
the Fisher information describes the attainable limit of 
the precision of the estimation of unitary\cite{Fuji,Im2,Giovannetti1,Giovannetti2,1,2,3,4}.
However, as was pointed in \cite{Ha3,Ha2},
it does not give the attainable bound of precision of the estimation of unitary.

The first studies \cite{Bu,Lu} treated the phase estimation, which is essentially the estimation of the representation of $\U(1)$.
Next, the estimation of $\SU(2)$ was studied \cite{Ba,Chi,Ha1}.
Chiribella et al \cite{Chi2} established a general theory of estimation of 
unitary representation of a compact group.
Chiribella \cite{Chi3} extended the result to the case of 
projective representations.
Kahn \cite{Kahn} applied this result to the case of $\SU(d)$.
These studies showed that the estimation error behaves as $\frac{C}{n^2}$
when $n$ is the number of tensor products of the representation.
We often call this phenomena the square speed up.
For a real implementation,
the energy of the input state might be a more important factor
than the available number of tensor products.
However, many existing studies do not address 
the optimal estimation with an energy constraint for the input state.
This paper deals with this kind of optimization problem.

On the other hand, Imai et al \cite{Im} 
treated phase estimation by using Fourier analysis.
In the estimation of action of finite group, 
the minimum error probability has been shown by \cite{Chi4,Chi5,AH}, 
and that with the projective representation case by \cite{TH}.
In the case of non-compact groups,
the estimation of group action has been formulated by Holevo \cite{Hol,Holevo}
when the input state is fixed.
However, the optimization of input state has been not resolved.
That is, there is no general theory 
of estimation of group action for non-compact groups.
In fact, the Fourier transform can be generalized to the case of
a non-compact group $G$,
whose generalized version is often called Plancherel transform.
In topological group theory, 
a locally compact Hausdorff topological group is called {\it unimodular}
when the left invariant measure is equal to the right invariant measure.
Further, when a unimodular group satisfies an additional condition,
it is called Type I.
In fact,
Fourier transform can be defined for a Type I group \cite{Fuhr,Folland}.
In this paper, we extend the concept of Type I group 
and the Fourier transform
to the case of projective representation with a fixed factor system.
In this case, we focus on the set $\hat{G}$ of irreducible  representations.
Under this method, the input state $\phi$
can be written as
totally square summable (integrable) matrices on irreducible representation spaces.
The inverse Fourier transform is given as
the unitary operator from the input state $\phi$
to the square integrable function on $G$, which can be regarded as an element of $L^2(G)$.
Hence, using the Fourier transform, 
we derive a general optimization result for estimation of a group.
In this formula, the minimum error can be written as
the minimum of the average error 
under the distribution 
given as the square integral of the inverse Fourier transform of the input pure state.
Then, we recover existing general results 
for finite groups and compact groups by \cite{Chi4,Chi5,AH,TH}
from our obtained general result.

Further, 
when the input system is infinite-dimensional,
it is natural to restrict the energy of the input state.
This constraint is also needed even in the finite-dimensional case, as is mentioned before.
However, the optimal estimation of group action with this type constraint has not been studied sufficiently with a general framework even in the compact case.
Using the Fourier transform, this paper gives a general result for this problem for a Type I group.
The merit of the obtain general result is to decrease the freedom of optimization.
That is, thanks to these results,
it is enough to treat the case when the measurement is a specific measurement
and the input is pure state.
These result reduce our optimization problem to the optimization 
with respect to input pure states.
Further, these results enable us to apply the known result of Fourier analysis
because these results clarify the relation with Fourier analysis.

In addition, 
we can consider the case when we can choose the input state probabilistically
as well as the case when
we choose only one input state.
When an arbitrary entangled state is available as the input state,
there is no difference between two schemes.
However, the relation between both is not so simple when there is 
restriction for available entanglement in the input state.
We treat this problem carefully, and show that there is no difference 
even when there is restriction for available entanglement in the input state.
However, we cannot show the same equivalence when 
there is an energy constraint.

Applying these general results,
we treat 
the estimation of actions of several concrete groups with and without energy constraint.
Firstly, we treat the case of commutative groups,
in which, the input state can be written as a function of 
the weight space.
We address the estimation of the action of real numbers $\real$ with energy constraint for the input state.
Then, we proceed to the same problem 
when the support of the input state belongs to the positive numbers.
We also discuss 
the estimation of the action of real numbers $\real$ 
when the support of the input state is limited to an interval.
These cases are treated by combination of the obtained general results 
and respective uncertainty relations. 
We also discuss the estimation of action of integers $\bZ$.

Further, we treat the estimation of action of the one-dimensional unitary group $\U(1)$
under the constraint of the support of the input state
as well as the constraint of the energy of the input state.
Under a suitable energy constraint and a suitable error function, 
the problem can be converted to
the eigenvalue problem of the specific periodic differential equation,
Mathieu equation.
Then, the optimal input state can be constructed from Mathieu function, which is
the solution of Mathieu equation with the minimum eigenvalue among periodic even functions.
Further, when the constraint energy is sufficiently large,
the optimal input state converges to a Gaussian state,
which is the wave function of the vacuum state.
As a byproduct, we derive an uncertainty relation for the wave function on the unit circle.

Next, we proceed to the non-commutative case.
First, we treat the estimation of
the action of 
the two-dimensional special unitary group $\SU(2)$ and 
the three-dimensional special orthogonal group $\SO(3)$
under the constraint of the available irreducible representation of the input state
as well as the constraint of the energy of the input state.
In fact, usually we consider the estimation of $\SU(2)$ with the gate fidelity
in the standard tensor product representation.
However, 
the standard tensor product representation of $\SU(2)$
can be regarded as at least a projective representation of $\SO(3)$.
Further, 
the gate fidelity can distinguish all elements in $\SO(3)$, but cannot in $\SU(2)$.
So, this paper treats them as a projective representation of $\SO(3)$.
Hence,
we discuss the estimation of $\SU(2)$ by using the trace of an element of $\SU(2)$
as error criterion when 
the input state is given as
the super position of a representation of $\SO(3)$ and
a projective representation of $\SO(3)$,
which cannot be regarded as a projective representation of $\SO(3)$.
Under the constraint of the available irreducible representation of the input state,
we can derive the optimal input state
in a similar way to the case of $\U(1)$. 
Under a suitable energy constraint and a suitable error function, 
similar to the case of $\U(1)$,
the problem can be converted to
the eigenvalue problem of Mathieu equation.
The case of $\SU(2)$ is different from the case of $\U(1)$
in that the solution can be derived from the minimum eigenvalue among periodic odd functions. 
Then, the optimal input state can be constructed 
as a superposition of maximally entangled state over irreducible representation.
The coefficients of superposition is given from the inverse Fourier transform of 
another type of Mathieu function, 
which is the solution of Mathieu equation with the minimum eigenvalue among periodic odd functions.
Further, when the constraint energy is sufficiently large,
The coefficients of superposition of the optimal input state converges to 
the wave function of the single photon state.
As a byproduct, we derive an uncertainty relation for the wave function on the 3-dimensional sphere, which is isomorphic to $\SU(2)$.

Next, we treat the case of $\SO(3)$.
Under the constraint of the available irreducible representation of the input state,
the asymptotically optimal estimation was derived by \cite{Ba,Chi,Ha1}.
However, they did not derive the exact form of the optimal estimation. 
In the case of truly projective representation of $\SO(3)$,
we can exactly derive the optimal input state
in a similar way to the cases of $\U(1)$ and $\SU(2)$. 
However, in the case of representation of $\SO(3)$,
we exactly derive the optimal input state in a way slightly different from the cases of $\U(1)$ and $\SU(2)$. 
Under a suitable energy constraint and a suitable error function, 
similar to the case of $\SU(2)$,
the problem can be converted to
the eigenvalue problem of Mathieu equation.
In the case of representation of $\SO(3)$,
the solution can be derived from the minimum eigenvalue among anti-periodic odd functions,
and
In the case of truly projective representation of $\SO(3)$,
the solution can be derived from the minimum eigenvalue among periodic odd functions.
Then, the optimal input state can be constructed 
as a superposition of maximally entangled state over irreducible representation
in a way similar to the case of $\SU(2)$.

We also show that
the asymptotically optimal performance under the energy constraint
can be physically realized 
by a repetition of the same input state and the individual measurement
in the cases of $\U(1)$, $\SU(2)$, and $\SO(3)$.
Since these methods require less entangled states (no entangled state in the case of $\U(1)$),
they give practical constructions.

When we can use so many tensor product systems, 
it is natural to restrict the average energy given by the total angular momentum of the input state.
However, 
the optimal estimation under this type of energy constraint
has not been studied. 
This paper treats the asymptotic behavior of this type optimization 
by using the above mentioned result with respect to $\U(1)$ with 
energy constraint and positivity constraint of the weight.

As by product, 
we can show the limiting distribution of
the outcome of the measurement corresponding 
to the irreducible decomposition in the $n$-fold tensor representation in the qubit system
when the true state is the completely mixed state.
In fact, 
when the true state is the $n$-fold tensor state of a non-completely-mixed state,
it is known that the limiting distribution is a Gaussian distribution \cite{qubit2,qubit3,qubit4,qubit5}.
In the case of the completely mixed state,
we derive the limiting distribution, which is different from the Gaussian distribution. 

Finally, we treat Heisenberg representation of $\real^2$
as a typical example of non-compact and non-commutative representation
by employing the uncertainty relation on $\bR^2$ and the general result based on the Fourier transform.
Under this representation, we 
give the minimum error of the estimation of action of group
when the average energy of the input state is fixed.
In this derivation, the Fourier analytic approach plays an important role.

The remaining parts are organized as follows.
Firstly, we summarize the typical obtained results with our typical energy constraint and
its application to uncertainty relations in Section \ref{s1-1},
which are a part of our obtained result.
In Section \ref{s2-1}, we introduce two schemes of the estimation of 
unknown group action.
In Section \ref{s2}, we give a formulation of the estimation of 
unknown group action.
In Section \ref{s4},
we derive a general formula
for minimum error as Theorems \ref{th3} and \ref{t6-24-1}
without and with an energy constraint as main theorems by using Plancherel theorem.
In Section \ref{s5}, we give their proofs.
Section \ref{s6} clarifies the relation between our theorem and the existing 
result for the case of finite groups by \cite{Chi4,Chi5,AH,TH}.
That is, this section explains how to recover the 
existing result for the case of finite group.
Section \ref{s7} treats the relation between our theorem and the existing 
result for the case of compact groups by \cite{Chi2,Chi3}.
The remaining sections discuss the concrete examples.
Sections \ref{s8}, \ref{s9}, \ref{s10}, \ref{s11}, \ref{s11b}, and \ref{s12}
treat the estimation of the action of
$\real$, $\bZ$, $\U(1)$, $\SU(2)$, $\SO(3)$, and the Heisenberg representation of $\real^2$,
respectively.

Appendix \ref{asB} summarizes the knowledges of Mathieu equation and  
Mathieu function, which play essential roles in the case of $\U(1)$, $\SU(2)$, and $\SO(3)$.
Appendices \ref{a1} and \ref{a5-18} are devoted for technical lemmas.

\section{Summary of obtained results with energy constraints}\Label{s1-1}
Here, we summarize the typical obtained results with our typical energy constraint 
and its application to uncertainty relations
as follows although our obtained results cover more general setups.

\vspace{0.5em}
\noindent{\it {Estimation of the location sift operation $\real$:}}
Firstly, let us consider 
the estimation of the location sift operation $x \in \real$.
In this case, any irreducible representation can be written 
as $x \mapsto e^{xp i}$ with the momentum $p \in \hat{\real}$ with 
$\hat{\real}=\real$.
Hence, any representation can be written 
as the unitary 
$U_x :=\int_{-\infty}^{\infty} e^{xp i} |p\rangle \langle p| d p$
on $L^2(\real)$.
In this case, 
the input state can be written as a square integrable function $\phi$ on 
the momentum space $\hat{\real}$.
When we apply the estimator $M(d \hat{x})$, which is a POVM, 
we obtain the output distribution $\langle \phi| U_x^\dagger  M(d ) U_x|\phi \rangle$.
 
Now, we consider the energy constraint on the momentum space $\hat{\real}$
as
$ \int_{-\infty}^{\infty} p^2 |\phi(p)|^2 \frac{d p}{\sqrt{2\pi}} \le E$,
which can be regarded as a constraint for the kinetic energy.
When we adopt the mean square error
${\cal D}(M,\phi):=
\int_{-\infty}^{\infty} 
(\hat{x}-x)^2 \langle \phi| U_x^\dagger  M(d \hat{x}) U_x|\phi \rangle$,
our problem can be formulated as the minimization problem:
\begin{align}
\min_{M,\phi} 
\{{\cal D}(M,\phi) |
 \int_{-\infty}^{\infty} p^2 |\phi(p)|^2 \frac{d p}{\sqrt{2\pi}} \le E \}
=\frac{8}{E},
\end{align}
which can be shown by employing the conventional minimum uncertainty relation
as Theorem \ref{T3-13-11}.
The optima input state is given by a Gaussian wave function.
Due to the central limit theorem,
the Gaussian wave function can be approximated by
the tensor product $\phi^{\otimes n}$
of an arbitrary pure state $\phi$.
In this case, the optimal coefficient of the first order can be attained 
by the maximum likelihood estimator with $n$ repeated applications 
of a proper covariant measurement to the system with the single copy input $\phi$.

\vspace{0.5em}
\noindent{\it {Estimation of the periodic location sift operation $\U(1)$:}}
Next, we consider the estimation of the location sift operation 
with the periodic condition.
In this case, the action can be described as the action $e^{\theta i} \in \U(1)$.
Then, any irreducible representation can be written 
as $\theta \mapsto e^{\theta k i}$ with the momentum $k \in \hat{U(1)}$ with 
$\hat{\U(1)}=\Z$.
Hence, any representation can be written 
as the unitary 
$U_\theta :=\oplus_{k=-\infty}^{\infty} e^{\theta k i} |k\rangle \langle k| $
on $L^2(\Z)$.
The input state can be written as a square integrable function $\phi$ on 
the momentum space $\hat{\U(1)}=\Z$.
Now, we consider the energy constraint on the momentum space $\hat{\U(1)}$
as $ \sum_{k=-\infty}^{\infty} k^2 |\phi(k)|^2 \le E$.
Similarly the output distribution is written as 
$\langle \phi| U_\theta^\dagger  M(d \hat{\theta}) U_\theta|\phi \rangle$
with the the estimator $M(d \hat{\theta})$.
When we adopt the error
${\cal D}(M,\phi):=
\int_{-\infty}^{\infty} 
(1-\cos (\hat{\theta}-\theta)) 
\langle \phi| U_\theta^\dagger  M(d \hat{\theta}) U_\theta|\phi \rangle$,
our problem can be formulated as the 
minimization problem:
\begin{align}
\min_{M,\phi} 
\Biggl\{{\cal D}(M,\phi) \Biggl|
 \sum_{k=-\infty}^{\infty} k^2 |\phi(k)|^2 \le E \Biggr\}
&=
\max_{s>0}\frac{sa_0(\frac{2}{s})}{4} +1- sE \nonumber\\
&\cong 
 \frac{1}{8E} - \frac{1}{128 E^2}
\hbox{ as } E \to \infty,
\end{align}
where $a_0$ is a function related to the Mathieu function,
and is defined in Appendix \ref{asB}. 
The above relations are shown as Theorem \ref{5-10-13} and \eqref{5-10-4}.

Further, the optimal coefficient of the first order can be attained 
by the following method.
The input state is the tensor product $\phi^{\otimes n}$
of an arbitrary pure state $\phi$.
We apply a proper covariant measurement to the system with the single copy input $\phi$.
Finally, we apply the maximum likelihood estimator for $n$ repeated applications 
of the above measurement.

\vspace{0.5em}
\noindent{\it {Estimation of the action $\SO(3)$ and $\SU(2)$:}}
Next, we consider the estimation of the rotating action $g \in \SO(3)$.
In this case, any irreducible representation can be written 
as $g \mapsto U_{\lambda,g}$ on the irreducible representation 
space ${\cal H}_\lambda$
with the maximum weight $\lambda \in \hat{\SO(3)}$.
Hence, any representation can be written 
as the unitary 
$U_g :=\oplus_{\lambda \in \hat{\SO(3)}} U_{\lambda,g} $
on $\oplus_{\lambda \in \hat{\SO(3)}}{\cal U}_\lambda \otimes {\cal U}_\lambda^*$,
where ${\cal U}_\lambda^*$ is the dual space of ${\cal U}_\lambda$.
In this case, 
the input state can be written as a square integrable function $\phi$ 
on $\oplus_{\lambda \in \hat{\SO(3)}}{\cal U}_\lambda \otimes {\cal U}_\lambda^*$.
When we apply the estimator $M(d \hat{g})$, 
we obtain the output distribution 
$\langle \phi| U_g^\dagger  M(d \hat{g}) U_g|\phi \rangle$.
 
Now, we consider the energy constraint 
as
$ 
\langle \phi| 
\oplus_{\lambda \in \hat{\SO(3)}} 
\lambda (\lambda +1) I_{\lambda} |\phi \rangle\le E$,
where $I_\lambda$ is the projection to the space ${\cal U}_\lambda \otimes {\cal U}_\lambda^*$,
by using the Casimir operator, which is natural in the relation with the angular momentum.
When we adopt the error
${\cal D}(M,\phi):=
\int_{-\infty}^{\infty} 
\frac{1}{4}(4- |\Tr g^{-1} \hat{g}|^2)
\langle \phi| U_x^\dagger  M(d \hat{x}) U_x|\phi \rangle$
with use of the gate fidelity $\frac{1}{4}|\Tr g^{-1} \hat{g}|^2$,
our problem can be formulated as the 
minimization problem:
\begin{align}
\min_{M,\phi} 
\Biggl\{{\cal D}(M,\phi) \Biggl|
\langle \phi| 
\bigoplus_{\lambda \in \hat{\SO(3)}} 
\lambda (\lambda +1) I_{\lambda} |\phi \rangle\le E \Biggr\} 
&=
\max_{s>0}\frac{sa_1(\frac{2}{s})}{4} +1- s(E+\frac{1}{4}) \nonumber \\
&\cong \frac{9}{8E} - \frac{81}{128 E^2}
\end{align}
as $E \to \infty$,
where $a_1$ is a function related to the Mathieu function,
and is defined in Appendix \ref{asB}. 
The above relations are shown as Theorem \ref{T3-13-3cx} and \eqref{5-10-4e2}.

Further, the optimal coefficient of the first order can be attained 
by the method given in the case of $\U(1)$.
A similar result can be shown when 
we consider the projective representation of $\SO(3)$.

For $\SU(2)$,
we adopt the error
${\cal D}(M,\phi):=
\int_{-\infty}^{\infty} 
(1-\frac{1}{2} \Tr g^{-1} \hat{g})
\langle \phi| U_x^\dagger  M(d \hat{x}) U_x|\phi \rangle$.
Then, our problem can be formulated as the 
minimization problem:
\begin{align}
\min_{M,\phi} 
\Biggl\{{\cal D}(M,\phi) \Biggl|
\langle \phi| 
\bigoplus_{\lambda \in \hat{\SU(2)}} 
\lambda (\lambda +1) I_{\lambda} |\phi \rangle\le E \Biggr\}
&=
\max_{s>0}\frac{sb_2(\frac{8}{s})}{16} +1- s(E+\frac{1}{4}) \nonumber \\
&\cong \frac{9}{32E} - \frac{7 \cdot 3^3}{2^11 E^2}
\end{align}
a $E \to \infty$,
where $b_2$ is a function related to the Mathieu function,
and is defined in Appendix \ref{asB}. 
The above relations are shown as Theorem \ref{T3-13-3b} and \eqref{5-10-4b}.

\vspace{0.5em}
\noindent{\it {Estimation of the action of the Heisenberg representation:}}
Finally, we consider the action of the Heisenberg representation $x=(x_1,x_2) \in \real^2$.
In this case, the irreducible representation 
is the equivalent with 
the Heisenberg representation $x \mapsto U_{x}$ on $L^2(\real)$
when we fix the commutation relation.
Then,
the input state can be written as a square integrable operator $\phi$ 
on $L^2(\real)$, which is a pure state on $L^2(\real)\otimes L^2(\real)$.
When we apply the estimator $M(d \hat{x})$, 
we obtain the output distribution 
$\langle \phi| U_x^\dagger  M(d \hat{x}) U_x|\phi \rangle$.
Now, we consider the energy constraint as
$ \langle \phi| (Q^2+P^2)\otimes I |\phi \rangle\le E$.
When we adopt the mean square error
${\cal D}(M,\phi):=
\int_{-\infty}^{\infty} 
( \hat{x}_1-x_1)^2+( \hat{x}_2-x_2)^2
\langle \phi| U_x^\dagger  M(d \hat{x}) U_x|\phi \rangle$,
our problem can be formulated as the 
minimization problem:
\begin{align}
\min_{M,\phi} 
\{{\cal D}(M,\phi) |
 \langle \phi| (Q^2+P^2)\otimes I |\phi \rangle\le E
\}
=\frac{1}{2E},
\end{align}
which can be shown by reducing the problem to the minimum uncertainty relation on 
the two-dimensional space
as Theorem \ref{Te4-23}

\vspace{0.5em}
\noindent{\it Uncertainty relations on 
$S^1$ and $S^3$:}
Using the relation
$S^1\cong \U(1)$ and $S^3\cong \SU(2)$, 
we derive uncertainty relations on $S^1$ and $S^3$.
Given $\varphi \in L^2(S^1)$, we focus on the relation between 
$\Delta_{\varphi}^2(\cos Q, \sin Q):=
\Delta_{\varphi}^2 \cos Q +\Delta_{\varphi}^2 \sin Q$
and $\Delta_{\varphi}^2 P$, where
$ \Delta_{\varphi}^2 X:= 
\langle \varphi | X^2 |\varphi \rangle -\langle \varphi | X |\varphi \rangle^2$.
Then, as is shown in Theorem \ref{T5-12-16}, 
we obtain 
\begin{align}
\min_{\varphi \in L^2_n(S^1)}
\{
\Delta_{\varphi}^2(\cos Q, \sin Q)|
\Delta_{\varphi}^2 P \le E \}
&=
\max_{s>0} 1-(sE-\frac{s a_0(\frac{2}{s})}{4})^2  \nonumber \\
& \cong 
\frac{1}{4E}-\frac{1}{32 E^2}
\hbox{ as } E \to \infty,
\end{align}
where $L^2_n(\Omega)$ is the set of normalized functions of $L^2(\Omega)$.
Given $\varphi \in L^2(S^3)$, we focus on the relation between 
$\Delta_{\varphi}^2 \vec{Q}:=\sum_{j=0}^3 \Delta_{\varphi}^2 Q_j $
and 
$\Delta_{\varphi}^2 \vec{P}:=\sum_{j=1}^3 \Delta_{\varphi}^2 P_j $,
where $P_j$ is the momentum operator for the $i$-th direction of $\sigma_j$ via the 
relation $S^3\cong \SU(2)$.
Then, as is shown in Theorem \ref{t5-29-1}, we obtain 
\begin{align}
\min_{\varphi \in L^2_n(S^3)}
\{
\Delta_{\varphi}^2\vec{Q}|
\Delta_{\varphi}^2 \vec{P} \le E \}
&=
1- (\min_{s>0} s(E+\frac{1}{4})-\frac{s b_2(\frac{8}{s})}{16})^2  \nonumber \\
&\cong 
\frac{9}{16E} -\frac{5 \cdot 3^3}{2^9 E^2}
\hbox{ as } E \to \infty.
\end{align}

\section{Estimation schemes of group action}\Label{s2-1}
We focus on a group $G$ 
acting on the Hilbert space ${\cal H}$ of our interest.
That is, we treat a projective unitary representation $f$ of $G$ over ${\cal H}$.
Our aim is estimating the unknown unitary $f(g)$ under the assumption that $g \in G$.
For this purpose, we can choose the input state $\rho$ and 
the output measurement, which is described by the POVM $M$ over the Hilbert space ${\cal H}$.
Since the aim of the measurement is the estimation of the element of $g \in G$,
the POVM $M$ takes values in the group $G$.
We describe the set of the above kinds of POVMs by ${\cal M}(G)$.
Hence, our estimator is given as a pair of an input state $\rho \in {\cal S}({\cal H})$ and a POVM $M$,
where ${\cal S}({\cal H})$ is the set of density operators on ${\cal H}$.
There are two kinds of extensions for this setting.
As the first extension, we allow to input a state entangled with the other system ${\cal H}_R$
and to apply an joint measurement between the output system and the other system
${\cal H}_R$
as Fig. \ref{f2}.
As the second extension given in  Fig. \ref{f1}, 
we choose the input state $\rho_i$ with the probability $p_i$ for $i=1, \ldots$
and 
choose the output POVM $M_i$ depending on the input state $\rho_i$.
Indeed, if we treat the representation space 
${\cal H}\otimes {\cal H}_R$, 
the first extension can be treated as the original setting.
The first extension (Fig. \ref{f2}) covers the second extension (Fig. \ref{f1})
when there is no restriction for the size of allowable entanglement in the initial state in the first extension (Fig. \ref{f2})
as follows.
Let ${\cal H}$ be the original input system
and ${\cal H}_R$ be the system spanned by $|i\rangle$.
Then, we choose the input state 
$\sum_i p_i \rho_i \otimes |i \rangle \langle i |$ on ${\cal H}\otimes {\cal H}_R$
and the POVM $M[\{M_i\}](\hat{g}):= \sum_i M_i(\hat{g}) \otimes |i \rangle \langle i |$ on ${\cal H}'\otimes {\cal H}_R$.
Hence, the second extension (Fig. \ref{f1}) is included in the first extension (Fig. \ref{f2}) with sufficient large entanglement.
However, the second extension has less choices than the first extension
and the second extension has larger choices than original setting.
Hence, we need to treat the second extension as a different setting.

\begin{figure}[htbp]
\begin{center}
\scalebox{0.6}{\includegraphics[scale=0.8]{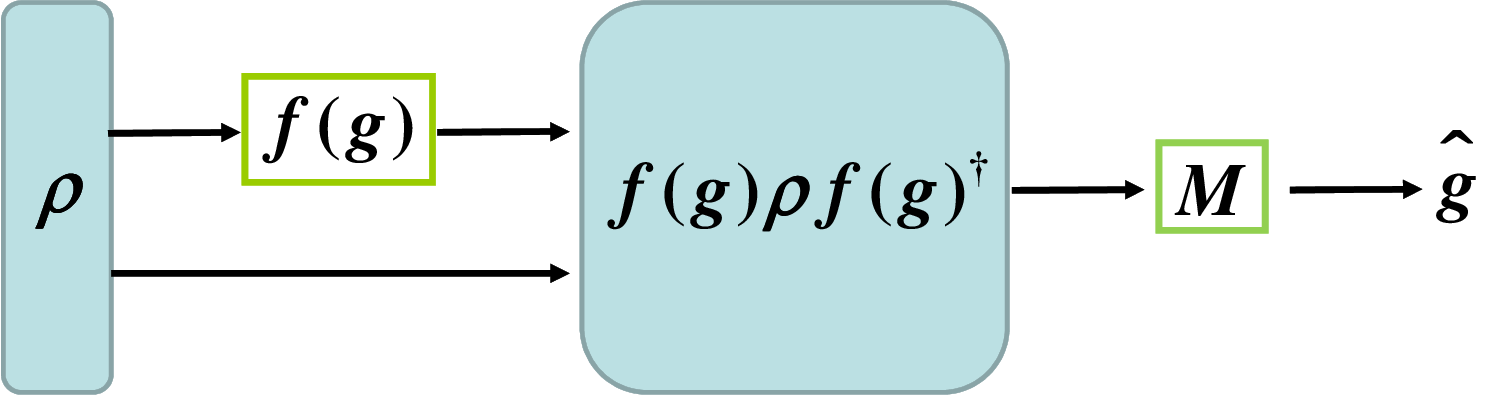}}
\end{center}
\caption{
Strategy for estimating the unknown group action $g$
with an entangled input}
\Label{f2}
\end{figure}%

\begin{figure}[htbp]
\begin{center}
\scalebox{0.6}{\includegraphics[scale=0.8]{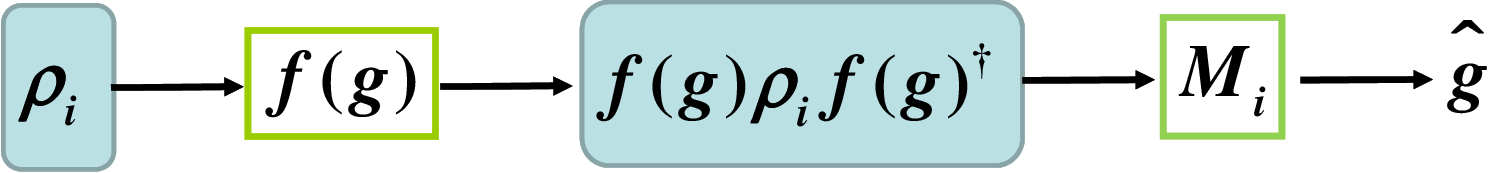}}
\end{center}
\caption{
Stochastic strategy for estimating the unknown group action $g$}
\Label{f1}
\end{figure}%

Thus, we have the following two schemes
for a given projective unitary representation $f$ of $G$ over ${\cal H}$
as follows.
\begin{itemize}
\item[\bf Scheme 1]
We input a state $\rho$ in the system ${\cal H}$.
We apply a measurement corresponding to a POVM $M$ on ${\cal H}$
after the unitary evolution $\rho \mapsto f(g) \rho f(g)^\dagger$.
In this scheme, we can use a state $\rho$ and a POVM $M$ on ${\cal H}$.

\item[\bf Scheme 2]
In the above scheme,
we choose the input state $\rho_i$ with the probability $p_i$ for $i=1, \ldots$
and 
choose the output POVM $M_i$ depending on the input state $\rho_i$.
The choices of the input and the measurement are
abbreviated to $(p_i,\rho_i)$ and $(M_i)$.
\end{itemize}

Then, we will discuss Schemes 1 and 2.
Subsection \ref{s2b} addresses Scheme 1 
with the optimization with respect to the POVM $M$
under the fixed choice of the input state $\rho$
in an arbitrary group $G$.
Subsection \ref{s2c} extends the analysis to Scheme 2.

\section{Formulations of estimation of group action}\Label{s2}
\subsection{Estimation with fixed input}\Label{s2b}
In order to treat the first scheme with the fixed input state $\rho$, 
we focus on the risk function $R$
depending on the true value $g$ and the estimate $\hat{g}$.
Then, when the true value is $g$, the average error is given as
\begin{align}
{\cal D}_{R,g}(\rho, M):=
\int_G R(g,\hat{g}) \Tr f(g) \rho f(g)^\dagger M(d\hat{g}).
\end{align}
Given a prior distribution $\nu$ for $g$ over $G$,
we can define the Bayesian error:
\begin{align}
{\cal D}_{R,\nu}(\rho,M):=
\int_G {\cal D}_{R,g}(\rho,M) \nu(dg).
\end{align}
Hence, our aim is finding a pair of 
the input state $\rho $ and POVM $M \in {\cal M}(G)$
minimizing ${\cal D}_{R,\nu}(\rho, M) $. 

As an alternative criterion, we optimize the worst case as
\begin{align}
{\cal D}_{R}(\rho,M):=
\max_G {\cal D}_{R,g}(\rho,M),
\end{align}
which is called the mini-max criterion.

Since the difference between $g$ and $\hat{g}$ is thought to be the same as that between
$g'g$ and $g'\hat{g}$,
we assume the left invariant condition in the following:
\begin{align}
R(g,\hat{g})=
R(g'g,g'\hat{g}),\quad \forall g,\hat{g},g' \in G.\Label{9-12-1}
\end{align}
According to Holevo\cite{Hol},
as an important class of POVMs,
we introduce a covariant POVM.
In the original formulation, he treats the estimation of a homogeneous space.
Since the group with the left action can be regarded as a homogeneous space,
we can apply his general method to our problem.
Hence, the right invariance in (\ref{9-12-1}) is not needed for its application.
A POVM $M$ taking values in $G$ is called {\it covariant} 
with respect to the projective representation $f$
when
\begin{align}
f(g) M(B) f(g)^\dagger = M(g B).
\end{align}
Holevo\cite{Hol} defined the concept ``covariant POVM''
 for a general homogeneous space.
The group $G$ can be regarded as a special case of homogeneous spaces.
We describe the set of covariant POVMs by ${\cal M}_{\cov}(G)$.
For any covariant POVM $M \in {\cal M}_{\cov}(G)$,
the average error ${\cal D}_{R,g}(\rho, M)$ does not depend on the true value $g$.
Hence, we obtain
\begin{align}
{\cal D}_{R,g}(\rho, M)= {\cal D}_{R,\nu}(\rho, M)={\cal D}_{R}(\rho, M).
\end{align}

In the following, we denote the left invariant measure of the group $G$
by $\mu_G$. 
When $G$ is compact, $\mu_G$ is chosen to be the probability measure.
Then, we obtain the following theorem,
which is called quantum Hunt-Stein theorem\cite{Hol}.
\begin{lemma}
When the risk function $R$ is invariant
and $G$ is compact and Hausdorff,
we obtain
\begin{align}
& \min_{M \in {\cal M}(G)} 
{\cal D}_{R,\mu_G}(\rho,M)
=
\min_{M \in {\cal M}(G)} 
{\cal D}_{R}(\rho,M) \\
= &
\min_{M \in {\cal M}_{\cov}(G)} 
{\cal D}_{R,\mu_G}(\rho,M) 
=
\min_{M \in {\cal M}_{\cov}(G)} 
{\cal D}_{R}(\rho,M).
\end{align}
\end{lemma}

However, when $G$ is not compact, it has no invariant probability measure.
In this case, the above theorem can be generalized to the following way\cite{Bog,Oz}.

\begin{lemma}\Label{th3-11}
When the risk function $R$ is left invariant
and
$G$ is locally compact and Hausdorff,
we obtain
\begin{align}
\min_{M \in {\cal M}(G)} 
{\cal D}_{R}(\rho,M)
=
\min_{M \in {\cal M}_{\cov}(G)} 
{\cal D}_{R}(\rho,M).
\end{align}
\end{lemma}

Hence, in the following, 
in order to treat our problem without the compactness condition,
we treat the minimization
\begin{align}
\min_{\rho \in {\cal S}({\cal H})}
\min_{M \in {\cal M}_{\cov}(G)} {\cal D}_{R}(\rho,M),
\end{align}
where ${\cal S}({\cal H})$ is the set of densities on ${\cal H}$.
That is, we can restrict our measurement into covariant measurements without loss of generality.
Given an input mixed state $\rho=\sum_i p_i |\phi_i \rangle \langle \phi_i|$,
any measurement $M$ satisfies
\begin{align}
{\cal D}_{R,g}(\rho,M)
=
\sum_i p_i {\cal D}_{R,g}(|\phi_i \rangle \langle \phi_i|,M).
\end{align}
Hence, any covariant measurement $M$ satisfies
\begin{align}
{\cal D}_{R}(\rho,M)
=
\sum_i p_i {\cal D}_{R}(|\phi_i \rangle \langle \phi_i|,M).
\end{align}

\subsection{Estimation with probabilistic input}\Label{s2c}
Next, we extend the above discussion to Scheme 2.
For this purpose, we apply the above discussion to the case with 
the Hilbert space ${\cal H}\otimes {\cal H}_R$
and the input state $\sum_{i} p_i \rho_i \otimes |i\rangle \langle i|$.
Then, for any POVM $M$ on ${\cal H}\otimes {\cal H}_R$,
we define the POVM $M'$ on ${\cal H}\otimes {\cal H}_R$
and the POVMs $M_i$ on ${\cal H}$ 
satisfying that
\begin{align}
M[\{M_i\}]
(\hat{g})
= \sum_{i=1}^l P_i M(\hat{g}) P_i,
\end{align}
where $P_i$ is the projection $I_{{\cal H}} \otimes |i \rangle \langle i| $.
Then, we obtain
\begin{align}
{\cal D}_{R,g}(\sum_{i} p_i \rho_i \otimes |i\rangle \langle i|, M)
=
{\cal D}_{R,g}(\sum_{i} p_i \rho_i \otimes |i\rangle \langle i|, M[\{M_i\}]
)
=
\sum_i p_i 
{\cal D}_{R,g}( \rho_i , M_i).
\end{align}
Combining Lemma \ref{th3-11}, we obtain the following lemma.
\begin{lemma}
When the risk function $R$ is left invariant
and
$G$ is locally compact and Hausdorff,
we obtain
\begin{align}
\min_{M \in {\cal M}(G)} 
\sum_i p_i 
{\cal D}_{R,g}( \rho_i , M_i)
=
\min_{M_i \in {\cal M}_{\cov}(G)} 
\sum_i p_i 
{\cal D}_{R,g}( \rho_i , M_i).
\end{align}
\end{lemma}
Hence,
in the following, 
in order to treat our problem without the compactness condition,
we treat the minimization
\begin{align}
\min_{\{p_i\}}
\min_{\rho \in {\cal S}({\cal H})}
\min_{M \in {\cal M}_{\cov}(G)} \sum_i p_i
{\cal D}_{R}(\rho_i,M).
\end{align}

Next, we characterize covariant POVMs.
It is known that any covariant measurement $M$ can be described by using a positive semi-definite
operator $T$ such that \cite{Hol}
\begin{align}
M(B) =\int_B f(g) T f(g)^\dagger \mu_G(dg) .\Label{h2}
\end{align}
Conversely, the above kind of operator $T$ satisfies
\begin{align}
I  =\int_G f(g) T f(g)^\dagger \mu_G(dg)  . \Label{h1}
\end{align}
When a positive semi-definite $T$ satisfies (\ref{h1}),
it gives a covariant measurement by (\ref{h2}), which is denoted by $M_T$.

\section{Analysis with irreducible decomposition}\Label{s4}
For a further analysis for general locally compact topological group, we employ the decomposition by irreducible representation spaces.
For this purpose, 
we prepare several notations and a condition for group.
In the following, we assume that the group $G$ is a {\it unimodular} group, i.e.,
$G$ is a locally compact Hausdorff topological group
and its left invariant measure $\mu_G$ is equal to its right invariant measure.

For a given projective representation $f$ of the group $G$,
we have the relation
\begin{align}
f(g)f(g')= e^{i\theta(g,g')}f(gg')
\end{align}
for $g,g' \in G$.
The set ${\cal L}:=\{e^{i\theta(g,g')}\}_{g,g' \in G}$ of complex numbers is called 
the factor system.
In particular, 
we call the factor system $\{1\}_{g,g' \in G}$ the trivial factor system
and denote it by ${\cal E}$.
The irreducible representation depends on the factor system ${\cal L}$.
In the following, we denote the set of symbols of irreducible projective representation of $G$ 
with the factor system ${\cal L}$
by $\hat{G}[{\cal L}]$.
For any label $\lambda \in \hat{G}[{\cal L}]$,
we denote the irreducible space corresponding to $\lambda$
by ${\cal U}_\lambda$,
and its irreducible representation by $f_{\lambda}$.
In the case of no factor system, i.e., the case of representation,
we denote the set $\hat{G}[{\cal L}]$ by $\hat{G}$.
When the group is simply connected,
any projective representation can be reduced to usual representation.
Then, since $\hat{G}[{\cal L}]$ does not depends on the factor system ${\cal L}$,
we denote $\hat{G}[{\cal L}]$ by $\hat{G}$.

When $G$ is compact, all of irreducible spaces ${\cal U}_{\lambda}$ are finite-dimensional.
When $G$ is not compact, there might be infinite-dimensional irreducible spaces ${\cal U}_{\lambda}$.
In this case, we define the generalized dimension as follows.
When the integral
\begin{align}
\int_G |\langle \phi | f_\lambda(g) | \phi \rangle|^2  \mu_G(d g)
\Label{3-10-1}
\end{align}
is finite for a normalized vector $\phi \in {\cal U}_{\lambda}$,
the integral (\ref{3-10-1}) does not depend on $\phi \in {\cal U}_{\lambda}$
because $G$ is a unimodular group.
Letting $d_{\lambda}$ be the inverse of the integral (\ref{3-10-1}),
we have
\begin{align}
I = d_{\lambda} \int_G f_\lambda(g) \rho f_\lambda(g)^\dagger \mu_G(d g)
\end{align}
for an arbitrary state $\rho$ on ${\cal U}_{\lambda}$.
In the compact case, 
since $\mu_G$ is a probability measure,
the generalized dimension $d_\lambda$ coincides with $\dim {\cal U}_\lambda$.

For any projective representation $f$ of a group $G$ on the Hilbert space ${\cal H}$, 
we define the commutant and the double commutant as follows. 
\begin{align}
f(G)'&:= \{A | f(g)A =Af(g),~\forall g \in G \} \\
f(G)''&:= \{A | BA =AB,~\forall B \in f(G)' \}.
\end{align}
Now, we introduce an important class of topological groups \cite[p.206]{Folland}.
\begin{definition}
A locally compact Hausdorff topological group $G$ is called type I
with the factor system ${\cal L}$
if $G$ is unimodular and satisfies the following condition.
For a unitary projective representation $f$ of 
on a Hilbert space ${\cal H}$, 
the set $f(G)' \cap f(G)''$ is the set of the constant operators on ${\cal H}$
if and only if $f$ is a direct sum of copies of a irreducible representation.
\end{definition}
The concept of `Type I' is closely related to Type I in von Neumann algebra \cite[p.206]{Folland}.
The original definition of `Type I' \cite[p.206]{Folland} 
is based on unitary representations, i.e., the case of the trivia factor system ${\cal E}$.
However, we employ the concept of `Type I' based on 
unitary projective representations with the factor system ${\cal L}$.
For example, any compact group is Type I with the trivial factor system ${\cal E}$ \cite[Example 1]{Folland}.
Since any projective representation of a compact group can be regarded
as a representation of its universal covering group,
any compact group is Type I with any factor system ${\cal L}$.
Any commutative group is also Type I with the trivial factor system ${\cal E}$ \cite[Example 2]{Folland}.
Also the group $\real^{2d}$ is Type I with the factor system given by the Heisenberg representation \cite[Example 3]{Folland}.
In the following, we assume that the group 
$G$ and the factor system ${\cal L}$ satisfies 
$G$ is type I with the factor system ${\cal L}$.

For a given projective representation $f$ of the group $G$ to the Hilbert space ${\cal H}$
with the factor system ${\cal L}$,
we can make the irreducible decomposition as follows \cite[Theorem 3.24]{Fuhr}.
\begin{align}
{\cal H}=\oplus_{\lambda \in \Lambda({\cal H})} {\cal U}_\lambda \otimes 
{\cal V}_{\lambda}
\Label{e4-21-1}
\end{align}
where 
${\cal V}_{\lambda}$
is the space describing the multiplicity of the irreducible space ${\cal U}_\lambda$, i.e.,
the group $G$ acts only on ${\cal U}_\lambda$ but not on ${\cal V}_{\lambda}$.
Here, $\Lambda({\cal H})$ is defied as a subset of $\hat{G}[{\cal L}]$ by
$\Lambda({\cal H}):= \{\lambda \in \hat{G}[{\cal L}] | \hbox{The space }{\cal H} \hbox{ contains }{\cal U}_\lambda .\}$.
For a pure state $|\phi\rangle \langle\phi|$,
the family of output states 
$\{f(g)|\phi\rangle \langle\phi|f(g)^\dagger \}_{g \in G}$
belongs to a subspace
$\oplus_{\lambda \in \Lambda} {\cal U}_\lambda \otimes 
{\cal V}_{\lambda}'$,
where the dimension ${\cal V}_{\lambda}'$ is $
\min\{\dim {\cal V}_{\lambda}, \dim {\cal U}_\lambda\}$.
Then, choosing an inclusion ${\cal V}_{\lambda}' \subset {\cal U}_\lambda^*$,
we have $|\phi \rangle \in {\cal H}_{\Lambda}$.
In fact, 
denoting the dual space of ${\cal V}_{\lambda}$ by ${\cal V}_{\lambda}^*$,
we can regarded 
a linear map $A$ from ${\cal V}_{\lambda}^*$ to ${\cal U}_{\lambda}$
as an element of the entangled space ${\cal U}_{\lambda}\otimes {\cal V}_{\lambda}$.
In this correspondence, we denote the entangled state by
$|A\rangle \rangle \in {\cal U}_{\lambda}\otimes {\cal V}_{\lambda}$.

Now, we fix a subset $\Lambda \subset \hat{G}[{\cal L}]$
such that the integral (\ref{3-10-1}) is finite for any $\lambda \in \Lambda$.
As a typical case, we focus on the following representation space:
\begin{align}
{\cal K}_\Lambda:= \oplus_{\lambda \in \Lambda}
 {\cal U}_\lambda \otimes {\cal U}_\lambda^*.
\end{align}

In order to employ the Fourier analysis, 
we identify 
the space ${\cal U}_\lambda \otimes {\cal U}_\lambda^*$
with the space of the Hilbert Schmidt operators on ${\cal U}_\lambda$.
By depending on the factor system ${\cal L}$,
the Fourier transform (Plancherel transform) with the factor system ${\cal L}$
is defined as a map ${\cal F}_{{\cal L}}$ from $L^2(G)$ to 
$\oplus_{\lambda \in \hat{G}[{\cal L}]} 
{\cal U}_\lambda \otimes {\cal U}_\lambda^*$
as follows.
Given $\varphi \in L^2(G)$,
we define
\begin{align}
({\cal F}_{{\cal L}} [\varphi])_\lambda:= 
\int_G f_\lambda (g) \varphi(g) \mu_G(dg).
\end{align}

Then, we have the following characterization\cite[Theorem 3.31]{Fuhr}\cite[Section 7.5]{Folland}.
\begin{proposition}[Plancherel Theorem]\Label{th4-24-1}
When $G$ is type I with the factor system ${\cal L}$,
there is a measure $\mu_{\hat{G}[{\cal L}]}$ on $\hat{G}[{\cal L}]$
such that
\begin{align}
\int_G |\varphi(g)|^2 \mu_G(d g)
=
\int_{\hat{G}[{\cal L}]} 
\|({\cal F}_{{\cal L}} [\varphi])_\lambda\|^2
 \mu_{\hat{G}[{\cal L}]}(d \lambda)
\end{align}
for $\varphi\in L^2(G)$,
where
\begin{align}
\|({\cal F}_{{\cal L}} [\varphi])_\lambda\|^2
:=
\Tr ({\cal F}_{{\cal L}} [\varphi])_\lambda
({\cal F}_{{\cal L}} [\varphi])_\lambda^\dagger.
\end{align}
\end{proposition}
In fact, Plancherel Theorem given in \cite[Theorem 3.31]{Fuhr}\cite[Section 7.5]{Folland} is based on representations.
It can be trivially extended to the case with a fixed factor system ${\cal L}$.

The measure $\mu_{\hat{G}[{\cal L}]}$ on $\hat{G}[{\cal L}]$
is called {\it Plancherel measure}.
In order to understand the meaning of Proposition \ref{th4-24-1},
we define the norm for 
$
|\phi\rangle=\oplus_{\lambda \in \hat{G}[{\cal L}]} |\phi_\lambda
\rangle\rangle\in
\oplus_{\lambda \in \hat{G}[{\cal L}]} 
{\cal U}_\lambda \otimes {\cal U}_\lambda^*$
as follows.
\begin{align}
\| \phi \|^2:=
\int_{\hat{G}[{\cal L}]} \|\phi_\lambda\|^2
 \mu_{\hat{G}[{\cal L}]}(d \lambda).
\end{align}
Here, we define
the space
$L^2(\hat{G}[{\cal L}])
:=\{ \phi \in
\oplus_{\lambda \in \hat{G}[{\cal L}]} 
{\cal U}_\lambda \otimes {\cal U}_\lambda^*
| \| \phi \|< \infty\}$,
and 
denote the subsets of normalized vectors in $L^2(G)$ and $L^2(\hat{G}[{\cal L}])$
by $L^2_n(G)$ and $L^2_n(\hat{G}[{\cal L}])$, respectively.
Then, we find that
the map ${\cal F}_{{\cal L}}$ is a unitary map 
from $L^2(G)$ to $L^2(\hat{G}[{\cal L}])$.
By using the Plancherel measure $\mu_{\hat{G}[{\cal L}]}$,
the inverse Fourier transform 
${\cal F}_{{\cal L}}^{-1}$
from $L^2(\hat{G}[{\cal L}])$ to $L^2(G)$ is given as
\begin{align}
{\cal F}_{{\cal L}}^{-1}[\phi](g) := 
\int_{\lambda \in \hat{G}[{\cal L}]} \Tr f_\lambda (g)^\dagger \phi_\lambda 
 \mu_{\hat{G}[{\cal L}]}(d \lambda)
\end{align}
for $|\phi\rangle= \oplus_{\lambda}|\phi_\lambda \rangle\rangle\in L^2(\hat{G}[{\cal L}])$.

When the group $G$ is compact,
the Plancherel measure $\mu_{\hat{G}[{\cal L}]}$ is discrete.
That is, we have
\begin{align}
\| \phi \|^2=
\sum_{\lambda \in \hat{G}[{\cal L}]} 
d_{\lambda}
\|\phi_\lambda\|^2.\Label{e-4-22-2}
\end{align}
Then, 
the inverse Fourier transform 
${\cal F}_{{\cal L}}^{-1}$
from $L^2(\hat{G}[{\cal L}])$ to $L^2(G)$ is given as
\begin{align}
{\cal F}_{{\cal L}}^{-1}[\phi](g) = 
\sum_{\lambda \in \hat{G}[{\cal L}]} {d_\lambda}
\Tr f_\lambda (g)^\dagger \phi_\lambda 
\Label{e-4-22-3}
\end{align}
for $|\phi\rangle= \oplus_{\lambda} |\phi_\lambda \rangle\rangle\in L^2(\hat{G}[{\cal L}])$. 
The relations \eqref{e-4-22-2} and \eqref{e-4-22-3} hold when the set $\hat{G}[{\cal L}]$ is discrete.

In order to give a typical covariant POVM,
we define a vector ${\cal I}_\lambda :=\sum_{i} |e_i\rangle \langle e_i | $
for CONS $\{e_i\}$ of ${\cal U}_\lambda$.
Then, we obtain an element $|{\cal I}_\lambda\rangle \rangle
\in {\cal U}_\lambda \otimes {\cal U}_\lambda^*$.
Then, we define
$|{\cal I}\rangle := \sum_{\lambda} |{\cal I}_\lambda\rangle\rangle $.
Here, the group $G$ acts only on the first space ${\cal U}_{\lambda}$
under the representation space ${\cal U}_\lambda \otimes {\cal U}_\lambda^*$.
Now, for a given representation space ${\cal H}$, 
we define a subset
${\cal K}_{{\cal H}} $
of ${\cal K}_{\Lambda({\cal H})}$ by
\begin{align}
{\cal K}_{{\cal H}}:= 
\{
\phi=\oplus_{\lambda \in \Lambda({\cal H})}
|\phi_\lambda \rangle \rangle
\in {\cal K}_{\Lambda({\cal H})}|
\rank \phi_\lambda \le \dim {\cal V}_\lambda
\}.\Label{6-1-1}
\end{align}
where $\rank \phi_\lambda$ is defined as a map from ${\cal U}_\lambda$
to ${\cal V}_\lambda^*$ by including the infinity.

Here, we can show the following lemma.
\begin{lemma}\Label{L4-21}
When the measure 
$\mu_{\hat{G}[{\cal L}]}(\Lambda({\cal H}))$ is zero,
there is no covariant POVM on ${\cal H}$.
\end{lemma}

When the group $G$ is not compact,
there is a possibility that
the set $\Lambda({\cal H})$ has zero measure under the 
Plancherel measure $\mu_{\hat{G}[{\cal L}]}$.
In this case, as is shown in Lemma \ref{L4-21},
we can perform no proper estimation.
Hence, for a proper estimation,
we have to prepare the Hilbert space ${\cal H}$ such that
$\mu_{\hat{G}[{\cal L}]}(\Lambda({\cal H}))>0$.
In the following, we assume that
any Hilbert space ${\cal H}$ satisfies 
the above condition and has the form (\ref{e4-21-1}).
Then, we employ the following inner product for $\Lambda=\Lambda({\cal H})$.
\begin{align}
\langle  \phi| \phi' \rangle
:=
\int_{\Lambda({\cal H})}
\langle \phi_{\lambda}| \phi_{\lambda}'\rangle 
\mu_{\hat{G}[{\cal L}]}(d \lambda).
\end{align}
Then, we can define the POVM $M_{|{\cal I}\rangle \langle {\cal I}|}$
on the quantum system ${\cal K}_\Lambda$.
Using the notation of the inverse Fourier transform, 
we have
\begin{align}
& {\cal D}_{R}(|\phi \rangle \langle \phi|,M_{|{\cal I}\rangle \langle {\cal I}|})
=\int_{G}
R(e,\hat{g}) |\langle {\cal I}| f(\hat{g})^\dagger|\phi \rangle
|^2 
\mu_G(d \hat{g}) \nonumber \\
=&\int_{G}
R(e,\hat{g}) |\langle {\cal I}| f(\hat{g})^\dagger|\phi \rangle
|^2 
\mu_G(d \hat{g}) \nonumber \\
=&\int_{G}
R(e,\hat{g}) 
| \int_{\Lambda} \Tr f(\hat{g})^\dagger \phi_{\lambda} \mu_{\hat{G}[{\cal L}]}
(d \lambda)|^2 
\mu_G(d \hat{g}) \nonumber \\
=&\int_{G}
R(e,\hat{g}) 
| {\cal F}_{{\cal L}}^{-1}[\phi] (\hat{g})|^2 
\mu_G(d \hat{g}), 
\Label{th7}
\end{align}
which is simplified to ${\cal D}_{R}(|\phi \rangle)$.
Hence, the output distribution can be written by using the inverse 
Fourier transform.


Then, we obtain the following theorem.
\begin{theorem}\Label{th3}
Let $f$ be 
a projective representation of a unimodular group $G$ 
to a Hilbert space ${\cal H}$.
Then, we obtain
\begin{align*}
&
\min_{\rho \in {\cal S}({\cal H})} 
\min_{M \in {\cal M}_{\cov}(G)} {\cal D}_{R}(\rho,M) 
=
\min_{ \{p_i\} }
\min_{\rho_i \in {\cal S}({\cal H})} 
\min_{M_i \in {\cal M}_{\cov}(G)} \sum_i p_i {\cal D}_{R}(\rho_i,M_i) \\
=&
\min_{|\phi \rangle \in {\cal K}_{{\cal H},n}} 
{\cal D}_{R}(|\phi \rangle),
\end{align*}
where 
${\cal K}_{{\cal H},n}$ is the set of normalized vectors in ${\cal K}_{{\cal H}}$.
\end{theorem}

When $G$ is not compact,
the representation space ${\cal H}=\oplus_{\lambda\in S}{\cal U}_{\lambda}\otimes {\cal V}_{\lambda}$ 
might be infinite-dimensional.
In this case, it is difficult to prepare an arbitrary state on the Hilbert space ${\cal H}$ as the initial state.
Hence, it is natural to restrict the average energy for the input state.
That is,
we consider a positive semi-definite self-adjoint operator $H_\lambda$ on the respective space ${\cal U}_{\lambda}$
and a given constant $E$,
and
we assume the condition for the initial state $\rho$
\begin{align}
\Tr H \rho \le E , \Label{6-20-3}
\end{align}
by using the Hamiltonian $H$ with the form
\begin{align}
H:=\bigoplus_{\lambda \in \Lambda({\cal H})} H_\lambda\otimes I .
\end{align}
This condition is meaningful even in the compact case
when there is a restriction for energy
while all of irreducible representation space in $\hat{G}[{\cal L}]$ can be prepared.
When the initial state 
is given by the pure state
$|\phi\rangle=\oplus_\lambda | \phi_\lambda \rangle\rangle \in L^2(\hat{G}[{\cal L}])$,
the above condition can be simplified to
\begin{align}
\int_{S} 
\Tr H_\lambda \phi_\lambda \phi_\lambda^\dagger 
\mu_{\hat{G}[{\cal L}]}(d \lambda)
\le E \Label{6-20-4}.
\end{align}

In Scheme 2, the constraint for the input choice $(p_i,\rho_i)$ is given as
\begin{align}
\sum_i p_i \Tr H \rho_i
\le E . \Label{6-20-3b}
\end{align}
When the states $\rho_i$ 
is given by the pure state
$|\phi_i\rangle=\oplus_\lambda | \phi_{\lambda,i} \rangle \rangle\in L^2(\hat{G}[{\cal L}])$,
the above condition can be simplified to
\begin{align}
\sum_i p_i
\int_{S} 
\Tr H_\lambda \phi_{\lambda,i} \phi_{\lambda,i}^\dagger 
\mu_{\hat{G}[{\cal L}]}(d \lambda)
\le E \Label{6-20-4b}.
\end{align}

Then, using 
the function
\begin{align}
\kappa(E):=
\min_{\phi \in {\cal K}_{{\cal H},n}} 
\{ {\cal D}_R(|\phi\rangle)| \langle \phi | H |\phi\rangle \le E\},
\end{align}
we can show the following theorem.

\begin{theorem}\Label{t6-24-1}
The relations
\begin{align}
& 
\kappa(E)
\ge 
\min_{\rho \in {\cal S}({\cal H})} 
\min_{M \in {\cal M}_{\cov}(G)} 
\{{\cal D}_R(\rho,M)| 
\Tr H \rho \le E
\}
\nonumber \\
\ge & \min_{ \{p_i\} }
\min_{\rho_i \in {\cal S}({\cal H})} 
\min_{M_i \in {\cal M}_{\cov}(G)} 
\{\sum_i p_i {\cal D}_{R}(\rho_i,M_i) | 
\sum_i p_i \Tr H \rho_i \le E
\} \nonumber\\
= &
\min_{ \{p_i,E_i\} }
\{\sum_i p_i \kappa(E_i)| \sum_i p_i E_i=E \}
 \Label{6-20-5-c}
\end{align}
hold.
In particular, 
when
the function $\kappa(E)$ is convex,
the relations
\begin{align}
&
\min_{\rho \in {\cal S}({\cal H})} 
\min_{M \in {\cal M}_{\cov}(G)} 
\{{\cal D}_R(\rho,M)| \Tr H \rho \le E \}
\nonumber \\
=&
\min_{ \{p_i\} }
\min_{\rho_i \in {\cal S}({\cal H})} 
\min_{M_i \in {\cal M}_{\cov}(G)} 
\{\sum_i p_i {\cal D}_{R}(\rho_i,M_i) | 
\sum_i p_i \Tr H \rho_i \le E
\} 
=
\kappa(E)
\Label{6-20-5-a}
\end{align}
hold.
\end{theorem}

Further, we have the following lemma.
\begin{lemma}\Label{L4-25-2}
When 
the relation $\dim {\cal V}_\lambda \ge \dim {\cal U}_\lambda$
holds for any $\lambda \in \Lambda({\cal H})$,
the function $\kappa(E)$ is convex.
\end{lemma}

Therefore, when the above condition holds,
it is sufficient to minimize ${\cal D}_R (|\phi\rangle)$
among pure input states $|\phi \rangle$
under the condition $\langle \phi | H |\phi\rangle \le E$.


It is often that $\kappa(E)$ is not easy to calculate.
In this case, we consider its Legendre transform
$\gamma(s):= \min_{E} \kappa(E)+ sE $,
which is concave and easier to calculate by the following way.
Since the map $\phi \mapsto {\cal D}_R(|\phi\rangle)$
is affine, there exists an operator $Y$ such that
$\langle \phi |Y| \phi \rangle=  {\cal D}_R(|\phi\rangle)$.
Then, $\gamma(s)$ is written as 
\begin{align}
\gamma(s)= \min_{\phi \in {\cal K}_{{\cal H},n}} 
\langle \phi | Y+ s H |\phi\rangle .
\end{align}
This value can be calculated by seeking the minimum eigenvalue of
the operator $Y+ s H$
when ${\cal K}_{{\cal H}}$ is a vector space, i.e.,
the relation $\dim {\cal V}_\lambda \ge \dim {\cal U}_\lambda$
holds for any $\lambda \in \Lambda({\cal H})$.
Using the function $\gamma(s)$, we can calculate $\kappa(E)$ as follows.
\begin{lemma}\Label{L5-10-1}
Assume that the function $\kappa(E)$ is convex.
For any $E$, there uniquely exists $s_E$ such that 
$\gamma'(s_E)=E$, 
where $\gamma'$ is the derivative of $\gamma$ .
Then, we have
\begin{align}
\kappa(E)= \gamma(s_E)- s_E E
= \max_{s>0} \gamma(s)- s E
.\Label{5-10-2}.
\end{align}
Further, $s_E$ is 
positive and monotone decreasing with respect to $E$.
\end{lemma}

\section{Proofs of Theorems \ref{th3} and \ref{t6-24-1}
and Lemmas \ref{L4-21}, \ref{L4-25-2}, and \ref{L5-10-1}}\Label{s5}
\begin{proofof}{Theorem \ref{th3}}
For any $|\phi\rangle=\oplus_{\lambda \in \Lambda({\cal H})}
|\phi_\lambda\rangle\rangle
\in {\cal K}_{{\cal H}}$,
the rank of $\phi_{\lambda}$ is not greater than
$\dim {\cal V}_{\lambda}$.
We choose a subspace ${\cal V}_{\lambda}'$ of ${\cal V}_{\lambda}$
so that $\dim{\cal V}_{\lambda}'= \min \{\dim {\cal V}_{\lambda}, \dim {\cal U}_{\lambda}\}$.
Then, we choose an inclusion map $P_\lambda:
{{\cal V}_{\lambda}'}^* 
{\cal U}_\lambda$.
We have the map 
$\phi_{\lambda} P_\lambda$ from 
${{\cal V}_{\lambda}'}^*$ to ${\cal U}_\lambda$ and
the dual map $P_\lambda^*:{\cal U}_\lambda^* \to {{\cal V}_{\lambda}'}$
so that
$|\phi_{\lambda} P_\lambda \rangle \rangle
\in {\cal U}_\lambda \otimes {\cal V}_{\lambda}'$
and
$| P_\lambda^*\rangle \rangle \in {\cal U}_\lambda \otimes {{\cal V}_{\lambda}'}$.
Choosing $|\tilde{\cal I}\rangle:=
\oplus_\lambda |P_\lambda^*\rangle \rangle$
and 
$|\tilde{\phi}\rangle:=
\oplus_{\lambda}
|\phi_{\lambda} P_\lambda\rangle\rangle$,
we have
\begin{align}
{\cal D}_{R}( |\phi\rangle\langle \phi|, M_{|{\cal I}\rangle \langle {\cal I}|})
=
{\cal D}_{R}( |\tilde{\phi}\rangle\langle \tilde{\phi}|, 
M_{|\tilde{\cal I}\rangle \langle \tilde{\cal I}|}).
\Label{4-26-1}
\end{align}
Hence, we have
\begin{align*}
\min_{|\phi \rangle \in {\cal K}_{{\cal H},n}} 
{\cal D}_{R}(|\phi \rangle) 
\ge 
\min_{\rho \in {\cal S}({\cal H})} 
\min_{M \in {\cal M}_{\cov}(G)} {\cal D}_{R}(\rho,M) .
\end{align*}

Further, the relations
\begin{align}
&\min_{\rho \in {\cal S}({\cal H})} 
\min_{M \in {\cal M}_{\cov}(G)} {\cal D}_{R}(\rho,M) 
\ge 
\min_{ \{p_i\} }
\min_{\rho_i \in {\cal S}({\cal H})} 
\min_{M_i \in {\cal M}_{\cov}(G)} 
\sum_i p_i {\cal D}_{R}(\rho_i,M_i) 
\nonumber \\
=& 
\min_{ \{p_i\} }
\min_{|\phi_i\rangle \in {\cal H}_n} 
\min_{M_i \in {\cal M}_{\cov}(G)} 
\sum_i p_i {\cal D}_{R}(|\phi_i\rangle \langle \phi_i| ,M_i) 
\Label{4-25-1}
\end{align}
are trivial, where ${\cal H}_n$ is the set of normalized vectors in ${\cal H}$.
For any probabilistic strategy $\{(p_i ,\rho_i,M_i)\}_i$,
we can choose $i$ such that
$\sum_i p_i {\cal D}_{R}(\rho_i,M_i) 
\ge  {\cal D}_{R}(\rho_i,M_i) $,
which implies the equality of the above second inequality.
Hence, we show the opposite inequality of the above first inequality.

Now, we make a decomposition of the operator $T$ 
as $T =\sum_{k}|\eta_k\rangle \langle \eta_k|$.
In the following, we use the notations 
$|\phi\rangle=\oplus_{\lambda \in \Lambda}| \phi_\lambda \rangle\rangle$,
$|\eta_k\rangle=\oplus_{\lambda\in \Lambda}|\eta_{k,\lambda} \rangle\rangle $,
$|x_k\rangle:=\oplus_{\lambda\in \Lambda} |\phi_\lambda
\eta_{k,\lambda}^\dagger\rangle\rangle$.
The output $\hat{g}$ satisfies the following distribution.
\begin{align*}
& 
\langle \phi |f(\hat{g}) T f(\hat{g}^{-1}) | \phi \rangle
\mu_G(d\hat{g}) 
=\sum_{k} 
|\langle \eta_k| f(\hat{g}^{-1}) | \phi \rangle|^2
\mu_G(d\hat{g}) \\
=& \sum_{k} 
|
\int_{\Lambda }
\Tr \eta_{k,\lambda}^{\dagger} f_\lambda(\hat{g}^{-1})  \phi_\lambda
\mu_{\hat{G}[{\cal L}]}(d \lambda)
|^2
\mu_G(d\hat{g})\\
=&
\sum_{k} 
|
\int_{\Lambda }
\Tr f_\lambda(\hat{g})^\dagger  \phi_\lambda
\eta_{k,\lambda}^{\dagger} 
\mu_{\hat{G}[{\cal L}]}(d \lambda)
|^2
\mu_G(d\hat{g})\\
=&
 \sum_{k} 
|\langle {\cal I}| f(\hat{g})^\dagger | x_k \rangle |^2
\mu_G(d\hat{g}) 
=
\sum_{k} 
|
{\cal F}_{{\cal L}}^{-1}[x_k](\hat{g})|^2
\mu_G(d\hat{g}).
\end{align*}
Hence, the relation
\begin{align*}
1=&
\int_{G}\langle \phi |f(\hat{g}) T f(\hat{g}^{-1}) | \phi \rangle
\mu_G(d\hat{g}) 
=
\sum_{k} 
\int_{G}
|
{\cal F}_{{\cal L}}^{-1}[x_k](\hat{g})|^2
\mu_G(d\hat{g})\\
=&
\int_{\Lambda }
\Tr  
\phi_\lambda^{\dagger}
\phi_\lambda 
\sum_{k} 
\eta_{k,\lambda}^{\dagger}\eta_{k,\lambda}
\mu_{\hat{G}[{\cal L}]}(d \lambda) 
\end{align*}
holds for any $\phi$ with $\|\phi\|=1$.
We obtain
$\sum_{k} 
\eta_{k,\lambda}^{\dagger}\eta_{k,\lambda}
=I_\lambda$, which is the identity operator on 
${\cal U}_{\lambda}$.
Thus,
\begin{align*}
\Tr \sum_{k} |x_k\rangle \langle x_k|
=
\int_{\Lambda }
\Tr  
\phi_\lambda^{\dagger}
\phi_\lambda 
\mu_{\hat{G}[{\cal L}]}(d \lambda) 
=1,
\end{align*}
i.e., 
$\sum_{k} |x_k\rangle \langle x_k|$ is a density operator.
Hence, we obtain
\begin{align*}
 {\cal D}_{R}(|\phi \rangle \langle \phi|, M_T)
=
{\cal D}_{R}(\sum_k |x_k\rangle\langle x_k|,
M_{|{\cal I}\rangle \langle {\cal I}|}) 
=
\sum_k 
\|x_k \|^2
{\cal D}_{R}( \frac{1}{\|x_k \|^2}|x_k\rangle\langle x_k|,
M_{|{\cal I}\rangle \langle {\cal I}|}) .
\end{align*}
Similarly, for the pure state $|\phi_i \rangle\langle \phi_i |$
and the covariant POVM $M_{T_i}$, 
we choose $\eta_{k,\lambda,i}$ and $|x_{k,i}\rangle$ 
as $T_i= \sum_k |\eta_{k,i} \rangle \langle \eta_{k,i}|$,
$|\eta_{k,i}\rangle = \oplus_{\lambda\in \Lambda}| \eta_{k,\lambda,i}\rangle \rangle$,
$|\phi_i \rangle = \oplus_{\lambda\in \Lambda}| \phi_{\lambda,i}\rangle\rangle $
and
$|x_{k,i}\rangle =\oplus_{\lambda\in \Lambda}|\phi_{\lambda,i} \eta_{k,\lambda,i}^\dagger\rangle\rangle $,
which implies that
\begin{align*}
&\sum_i p_i {\cal D}_{R}(|\phi_i \rangle \langle \phi_i|, M_{T_i})
=
\sum_i p_i {\cal D}_{R}(\sum_k |x_{k,i}\rangle\langle x_{k,i}|,
M_{|{\cal I}\rangle \langle {\cal I}|}) \\
=&
\sum_{k,i} 
p_i \|x_{k,i} \|^2
{\cal D}_{R}( \frac{1}{\|x_{k,i} \|^2}|x_{k,i}\rangle\langle x_{k,i}|,
M_{|{\cal I}\rangle \langle {\cal I}|}) .
\end{align*}
Hence, there exist $k$ and $i$ such that
\begin{align*}
\sum_i p_i {\cal D}_{R}(|\phi_i \rangle \langle \phi_i|, M_{T_i})
\ge 
{\cal D}_{R}( \frac{1}{\|x_{k,i} \|^2}|x_{k,i}\rangle\langle x_{k,i} |,
M_{|{\cal I}\rangle \langle {\cal I}|}) .
\end{align*}
Since the rank of 
$\phi_{\lambda,i} \eta_{k,\lambda,i}^\dagger$
is not greater than
$\dim {\cal V}_{\lambda}$,
$ |x_{k,i}\rangle$ belongs to ${\cal K}_{{\cal H}}$.
Hence, we have
the inequality 
\begin{align*}
\min_{|\phi \rangle \in {\cal K}_{{\cal H},n}} 
{\cal D}_{R}(|\phi \rangle) 
\le 
\min_{ \{p_i\} }
\min_{|\phi_i\rangle \in {\cal H}} 
\min_{M_i \in {\cal M}_{\cov}(G)} 
\sum_i p_i {\cal D}_{R}(|\phi_i\rangle \langle \phi_i| ,M_i) ,
\end{align*}
which is opposite to (\ref{4-25-1}).
\end{proofof}

\begin{proofof}{Theorem \ref{t6-24-1}}
We have the relations
\begin{align}
& 
\kappa(E)
\ge 
\min_{\rho \in {\cal S}({\cal H})} 
\min_{M \in {\cal M}_{\cov}(G)} 
\{{\cal D}_R(\rho,M)| 
\Tr H \rho \le E
\}
\nonumber \\
\ge & \min_{ \{p_i\} }
\min_{\rho_i \in {\cal S}({\cal H})} 
\min_{M_i \in {\cal M}_{\cov}(G)} 
\{\sum_i p_i {\cal D}_{R}(\rho_i,M_i) | 
\sum_i p_i \Tr H \rho_i \le E
\} \nonumber \\
= & \min_{ \{p_i\} }
\min_{\phi_i \in {\cal H}} 
\min_{M_i \in {\cal M}_{\cov}(G)} 
\{\sum_i p_i {\cal D}_{R}(|\phi_i \rangle \langle \phi_i|,M_i) | 
\sum_i p_i \langle \phi_i| H |\phi_i \rangle  \le E
\} ,\Label{5-8-1}
\end{align}
where the first inequality can be shown by (\ref{4-26-1}).
Other relations in \eqref{5-8-1} are trivial.
In fact, the relation (\ref{4-26-1}) yields that
\begin{align}
& \min_{ \{p_i\} }
\min_{\phi_i \in {\cal H}} 
\min_{M_i \in {\cal M}_{\cov}(G)} 
\{
\sum_i p_i {\cal D}_{R}(|\phi_i \rangle \langle \phi_i|,M_i) | 
\sum_i p_i \langle \phi_i| H |\phi_i \rangle  \le E
\} \nonumber \\
\le &
\min_{ \{p_i,E_i\} }
\{\sum_i p_i \kappa(E_i)| \sum_i p_i E_i=E \}.
\Label{4-25-2b}
\end{align}
Hence, it is enough to show the inequality opposite to (\ref{4-25-2b}).
Similar to Proof of Theorem \ref{th3}, 
for pure states $|\phi_i \rangle\langle \phi_i |$
and covariant POVM $M_{T_i}$, 
we choose $\eta_{k,\lambda,i}$ and 
$|x_{k,i}\rangle$. 
Then, 
\begin{align}
&\sum_i p_i 
\langle \phi_i | H |\phi_i \rangle
=
\sum_i p_i 
\int
\langle \phi_{\lambda,i} | H_\lambda |\phi_{\lambda,i} \rangle
\mu_{\hat{G}[{\cal L}]}(d \lambda) \nonumber \\
=&
\sum_i p_i 
\int
\Tr \phi_{\lambda,i}^\dagger H_\lambda \phi_{\lambda,i}
\sum_k \eta_{k,\lambda,i}^\dagger  \eta_{k,\lambda,i}
\mu_{\hat{G}[{\cal L}]}(d \lambda) \nonumber \\
=&
\sum_k \sum_i p_i 
\int
\Tr \eta_{k,\lambda,i}
\phi_{\lambda,i}^\dagger H_\lambda \phi_{\lambda,i}
 \eta_{k,\lambda,i}^\dagger  
\mu_{\hat{G}[{\cal L}]}(d \lambda) 
=
\sum_k \sum_i p_i 
\langle x_{k,i} | H |x_{k,i} \rangle.
\end{align}
Since
\begin{align*}
\sum_i p_i {\cal D}_{R}(|\phi_i \rangle \langle \phi_i|,M_i) 
=
\sum_{i,k} \|x_{k,i}\|^2 p_i 
{\cal D}_{R}(\frac{1}{\|x_{k,i}\|^2}|x_{k,i} \rangle \langle x_{k,i}|,
M_{|{\cal I}\rangle \langle {\cal I}|}) ,
\end{align*}
we obtain the inequality opposite to (\ref{4-25-2b}).

Further, when $\kappa(E)$ is convex
$\min_{ \{p_i,E_i\} }
\{\sum_i p_i \kappa(E_i)| \sum_i p_i E_i=E \}
=\kappa(E)$, which implies (\ref{6-20-5-a}).
\end{proofof}

\begin{proofof}{Lemma \ref{L4-25-2}}
It is enough to show that
\begin{align}
p{\cal D}_R (|\phi_1\rangle)
+(1-p){\cal D}_R (|\phi_2\rangle)
\ge 
\min_{|\phi\rangle \in {\cal K}_{{\cal H},n}} 
\{ {\cal D}_R(|\phi\rangle)| \langle \phi | H |\phi\rangle \le E\}
\Label{4-25-10}
\end{align}
when 
$p\in [0,1]$ and
$p \langle \phi_1 |H | \phi_1 \rangle 
+(1-p) \langle \phi_2 |H | \phi_2 \rangle =E$.
The map $\rho\mapsto {\cal D}_R (\rho, M_{|{\cal I}\rangle \langle {\cal I}|})$
is affine.
Thus, there exists a self-adjoint map $Y$ such that
$\Tr \rho Y={\cal D}_R (\rho, M_{|{\cal I}\rangle \langle {\cal I}|})$.
We apply Lemma \ref{L4-25} to the two-dimensional subspace spanned by 
$|\phi_1\rangle$ and $|\phi_2\rangle$.
Then, there exists a vector $|\phi\rangle$ given as a superposition of
$|\phi_1\rangle$ and $|\phi_2\rangle$ such that
\begin{align*}
p{\cal D}_R (|\phi_1\rangle)
+(1-p){\cal D}_R (|\phi_2\rangle)
\ge {\cal D}_R (|\phi\rangle) 
\hbox{\rm~and~}
\langle \phi | H |\phi\rangle = E.
\end{align*}
Thanks to the condition of Lemma \ref{L4-25-2},
$|\phi\rangle$ belongs to ${\cal K}_{{\cal H},n}$.
Hence, we obtain (\ref{4-25-10}).
\end{proofof}

\begin{remark}
Here, we should remark that the condition of Lemma \ref{L4-25-2} is crucial for the above proof.
If the condition does not hold,
we cannot say that
the superposition $|\phi\rangle$ of $|\phi_1\rangle$ and $|\phi_2\rangle$
belongs to ${\cal K}_{{\cal H}}$
because 
${\cal K}_{{\cal H}}$ is not a linear space.
\end{remark}

\begin{proofof}{Lemma \ref{L5-10-1}}
Due to the concavity of $\gamma(s)$, there uniquely exists $s_E$ 
$\gamma'(s_E)=E$.

Since $\kappa(E)$ is convex,
\begin{align}
\gamma(-\kappa'(E) )= \kappa(E)- E\kappa'(E).
\Label{5-10-1}
\end{align}
Taking the derivative with respect to $E$, we have
$-\kappa''(E)\gamma'(-\kappa'(E))
= \kappa'(E)- \kappa'(E) - E \kappa''(E)
=- E \kappa''(E)$.
That is, $\gamma'(-\kappa'(E))=E$, which implies $s_E=-\kappa'(E)$.
Hence, substituting $-\kappa'(E)=s_E$ into \eqref{5-10-1}, 
we have the first equation in \eqref{5-10-2}.
Since $\gamma$ is concave, we have 
the second equation in \eqref{5-10-2}.

Since $\kappa(E)$ is monotone decreasing and convex,
$s_E$ is 
positive and monotone decreasing with respect to $E$.
\end{proofof}

\begin{proofof}{Lemma \ref{L4-21}}
For any $|\phi\rangle = \oplus_{\lambda \in \Lambda({\cal H})} 
|\phi_\lambda \rangle \rangle
\in {\cal H}$ with $\|\phi\|=1$, 
we have
\begin{align}
& \| \int_G f(g)|\phi \rangle \langle \phi| f(g)^\dagger \mu_G(d g) \|
\ge
\langle \phi |\int_G f(g)|\phi \rangle \langle \phi| f(g)^\dagger \mu_G(d g) |
\phi \rangle \nonumber \\
=&
 \int_G |\langle \phi| f(g)|\phi \rangle |^2 \mu_G(d g) 
=
\int_{\Lambda({\cal H})}
|\langle\langle \phi_\lambda| \phi_\lambda \rangle\rangle |^2
\mu_{\hat{G}[{\cal L}]}(d \lambda) 
\nonumber \\
\ge &
\frac{(\int_{\Lambda({\cal H})}
\langle \langle \phi_\lambda| \phi_\lambda \rangle \rangle
\mu_{\hat{G}[{\cal L}]}(d \lambda))^2}
{
\int_{\Lambda({\cal H})}
1\mu_{\hat{G}[{\cal L}]}(d \lambda)} 
=
\frac{\|\phi\|^4}
{\mu_{\hat{G}[{\cal L}]}(\Lambda({\cal H}) )}
=
\frac{1}{\mu_{\hat{G}[{\cal L}]}(\Lambda({\cal H}) )},
\Label{4-26-4}
\end{align}
where the inequality in (\ref{4-26-4}) follows from Schwarz inequality.
This inequality implies that
the norm $\| \int_G f(g)|\phi \rangle \langle \phi| f(g)^\dagger \mu_G(d g) \|$ is infinity 
when the measure $\mu_{\hat{G}[{\cal L}]}(\Lambda({\cal H}) )$
is zero.
This fact implies that no covariant measure exists
when the measure $\mu_{\hat{G}[{\cal L}]}(\Lambda({\cal H}) )$
is zero.
\end{proofof}


\section{Finite group}\Label{s6}
As a typical case, we treat finite groups.
This section explains how we recover the 
minimum error formula for finite groups
by \cite{Chi4,Chi5,AH,TH}
from our general result, Theorem \ref{th3}.

It is natural to treat the case
\begin{align*}
R(g,\hat{g})=
\left \{
\begin{array}{ll}
0 & \hbox{ if } g= \hat{g} \\
1 & \hbox{ if } g\neq \hat{g} .
\end{array}
\right.
\end{align*}
In this case,
since the invariant probability measure is $\frac{1}{|G|}$,
\begin{align*}
{\cal D}_{R}(|\phi \rangle \langle \phi|,M_{|{\cal I}\rangle \langle {\cal I}|}) 
=
1- \frac{|{\cal F}_{{\cal L}}^{-1}[\phi](e)|^2}{|G|}.
\end{align*}
Hence,
it is sufficient to calculate 
$|{\cal F}_{{\cal L}}^{-1}[\phi](e)|^2$.
Since ${\cal V}_{\lambda}$ is a subspace of
${\cal U}_{\lambda}^*$,
${\cal V}_{\lambda}^*$ can be regarded as subspace 
of ${\cal U}_{\lambda}$.
Now, we denote the projection to the subspace 
by $P({\cal V}_{\lambda}^* )$.
Hence, for any input state $|\phi \rangle = 
\oplus_{\lambda \in S }|\phi_\lambda\rangle$, we have
\begin{align*}
&|{\cal F}_{{\cal L}}^{-1}[\phi](e)|^2 
=|\sum_{\lambda \in S} {d_\lambda}\Tr \phi_\lambda |^2 
=|\sum_{\lambda \in S} {d_\lambda}\Tr \phi_\lambda
P({\cal V}_{\lambda}^* ) |^2 \\
\le &
(\sum_{\lambda \in S} {d_\lambda}\Tr 
P({\cal V}_{\lambda}^* )^2 )
(\sum_{\lambda \in S} d_\lambda \Tr \phi_\lambda^\dagger \phi_\lambda) 
= 
\sum_{\lambda \in S} d_\lambda
\dim {\cal V}_{\lambda} .
\end{align*}
Further, 
the equality holds
when $
\phi_\lambda=
\frac{1}{\sqrt{\sum_{\lambda \in S} d_\lambda \dim {\cal V}_{\lambda}}}
P({\cal V}_{\lambda}^* )$.
Thus,
we can recover the existing result\cite{Chi4,Chi5,AH,TH}
\begin{align*}
&\min_{\phi}
{\cal D}_{R}(|\phi \rangle \langle \phi|,M_{|{\cal I}\rangle \langle {\cal I}|}) 
=
1-
\frac{\sum_{\lambda \in S} d_\lambda \dim {\cal V}_{\lambda} }{|G|}.
\end{align*}

\section{Characterization by irreducible characters}\Label{s7}
When $G$ is a compact group,
Chiribella et al \cite{Chi2} showed that
a general formula for the minimum error by using irreducible characters.
Chiribella \cite{Chi3} extended the result to the case with projective representation.
This section explains how we recover the result 
from our general result, Theorem \ref{th3}.

When the error function $R$ satisfies that
the right invariance 
$R(g,\hat{g})=R(g g',\hat{g}g')$ for $g,g\hat{g},g' \in G$ as well as the left invariance,
we have
$R(g,\hat{g})=R(e,g^{-1}\hat{g})=R(g'e {g'}^{-1},g' g^{-1}\hat{g}{g'}^{-1} )=
R(e,g' g^{-1}\hat{g}{g'}^{-1} )$.
Hence, the function $R(g,\hat{g})$ can be written as
$R(g,\hat{g})=\sum_{\lambda \in \hat{G}}
\tilde{a}_\lambda \chi_{\lambda}(g^{-1}\hat{g})$
with constants $\tilde{a}_\lambda$, where
$\chi_\lambda$ is the irreducible character of the representation $f_\lambda$.
For a factor system ${\cal L}$,
we denote the factor system composing of the complex conjugate of ${\cal L}$
by $-{\cal L}$.
Then, for a projective representation $f_{\lambda}$ with $\lambda \in \hat{G}[{\cal L}]$,
the complex conjugate projective representation is denoted as $f_{\lambda^*}$
and its factor system is $-{\cal L}$.
That is, 
$\lambda^* \in \hat{G}[-{\cal L}]$.
In particular, 
when $\lambda \in \hat{G}$, $\lambda \in \hat{G}$.

Now, we additionally assume 
\begin{align}
R(g,\hat{g})=
R(e, \hat{g} g^{-1})=
a_0-\sum_{\lambda \in \hat{G} \setminus \{0\}}
a_{\lambda} 
(\chi_{\lambda}(\hat{g} g^{-1} )+ \chi_{\lambda^*}(\hat{g} g^{-1} )),
\quad a_{\lambda}=a_{\lambda^*} \ge 0 \Label{6-26-14},
\end{align}
where we denote the trivial representation by $f_{\lambda^*}$ and $f_0$.

This problem is equivalent with the maximization of the merit function
\begin{align}
\tilde{R}(g,\hat{g})=
\sum_{\lambda \in \hat{G} \setminus \{0\}}
a_{\lambda} 
(\chi_{\lambda}(\hat{g} g^{-1} )+ \chi_{\lambda^*}(\hat{g} g^{-1} )),
\quad a_{\lambda}=a_{\lambda^*} \ge 0 \Label{6-26-14-a}.
\end{align}

For example,
in the case of $G=\SU(d)$,
as a merit function,
we often adopt the gate fidelity $\frac{1}{d^2}|\Tr \hat{g} g^{-1}|^2$.
In the case of $d=2$,
the set $\hat{G}$
can be identified with the set $\{\frac{n}{2}| n \in \bZ, n \ge 0\}$
by identifying the irreducible representation space by the maximal weight of the representation of 
$\left(
\begin{array}{cc}
1/2 & 0 \\
0 & -1/2
\end{array}
\right) \in \su(2)$.
Then, 
the gate fidelity $\frac{1}{4}|\Tr \hat{g} g^{-1}|^2$ is calculated to $1+\chi_{1}(\hat{g} g^{-1})$.
When $\hat{g}=g$, the gate fidelity is 1.
Hence, 
we often use the quantity $4-|\Tr \hat{g} g^{-1}|^2
=3-\chi_{1}(\hat{g} g^{-1})$
as an error function.

In order to deal with the compact case, we define 
the coefficient $C_{{\lambda''},\lambda}^{{\lambda'}}$ 
for $\lambda\in \hat{G}[{\cal L}]$, $\lambda''\in \hat{G}[-{\cal L}]$, and $\lambda' \in \hat{G}$ 
as follows.
\begin{align}
{\cal U}_{\lambda}\otimes 
{\cal U}_{\lambda''}
=
\bigoplus_{\lambda' \in \hat{G}}
{\cal U}_{\lambda'}
\otimes
\complex^{C_{{\lambda''},\lambda}^{{\lambda'}}} .
\end{align}
That is,
the integer $C_{{\lambda''},\lambda}^{{\lambda'}}$ is the multiplicity 
of the irreducible representation space ${\cal U}_{\lambda'}$
in the tensor product space ${\cal U}_{\lambda}\otimes {\cal U}_{\lambda''}$.

In this case, any input pure state $|\phi\rangle\in L^2_n(\hat{G}[{\cal L}])$
has a decomposition
\begin{align}
|\phi\rangle= \bigoplus_{\lambda\in \hat{G}[{\cal L}]} c_\lambda |\Phi_\lambda\rangle\rangle
\Label{e4-22}
\end{align}
with the conditions
$c_\lambda\ge 0 $
and
$\Tr \Phi_\lambda^\dagger \Phi_\lambda=1$.
The normalized condition is given as
$\sum_{\lambda\in \hat{G}[{\cal L}]}
d_\lambda c_\lambda^2=1$.
Then, we obtain the following lemma.

\begin{lemma}\Label{11-16-2}
For any input pure state $|\phi\rangle\in L^2_n(\hat{G}[{\cal L}])$,
choosing an error function $R$ satisfying (\ref{6-26-14}),
under the decomposition (\ref{e4-22}),
we obtain
\begin{align}
{\cal D}_R(|\phi\rangle)
\ge &
-\sum_{\lambda,\lambda'\in \hat{G}[{\cal L}]}
\sqrt{d_\lambda d_{\lambda'}}
c_\lambda c_{\lambda'}
\sum_{\lambda''\in \hat{G}}
a_{\lambda''}
(C_{\lambda,{\lambda'}^*}^{{\lambda''}^*}
+C_{\lambda,{\lambda'}^*}^{\lambda''}) 
 \Label{11-16-1}.
\end{align}
The equality holds when
$|\phi\rangle=\oplus_{\lambda\in \hat{G}[{\cal L}]}  c_\lambda |\Psi_\lambda\rangle\rangle $,
where
$\Psi_{\lambda}:=
\frac{1}{\sqrt{d_{\lambda}}}
I_{\lambda} $.
Further, 
In this case,
the relations
\begin{align}
\min_{M\in {\cal M}_{\cov}(G)}{\cal D}_R (|\phi\rangle, M)
&={\cal D}_R (|\phi\rangle, M_{|{\cal I}\rangle \langle {\cal I}|})
\Label{5-4-1}
\\
{\cal F}^{-1}_{{\cal L}}[\phi]
&=
\sum_{\lambda\in \hat{G}[{\cal L}]} \sqrt{d_\lambda} c_{\lambda}\chi_\lambda
\Label{5-4-2}
\end{align}
hold.
\end{lemma}
Therefore, our optimization problem can be reduced to 
the optimization concerning the choice of $c=(c_\lambda)_{\lambda \in \Lambda}$
when our representation space is 
${\cal H}_\Lambda$.

That is, combining Theorem \ref{th3},
we can recover the following known result\cite{Chi2,Chi3}.
Under the same assumption as Lemma \ref{11-16-2},
we have
\begin{align*}
& \min_{\rho \in {\cal S}({\cal H}_\Lambda),M \in {\cal M}(G)} 
{\cal D}_{R,\mu_G}(\rho,M)
=
\min_{\rho\in {\cal S}({\cal H}_\Lambda),M \in {\cal M}(G)} 
{\cal D}_{R}(\rho,M) \nonumber \\
= &
\min_{ c\in V_\Lambda }
-\sum_{\lambda,\lambda'\in \Lambda}
c_\lambda c_{\lambda'}\sqrt{d_\lambda d_{\lambda'}}
\sum_{\lambda''\in \hat{G}}
a_{\lambda''}
(C_{\lambda,{\lambda'}^*}^{{\lambda''}^*}
+C_{\lambda,{\lambda'}^*}^{\lambda''}) ,
\end{align*}
where $V_\Lambda$ is the set of real vectors $c=(c_\lambda)_{\lambda \in \Lambda}$
satisfying that
$c_\lambda \ge 0$ and $\sum_{\lambda \in \Lambda}
d_\lambda c_\lambda^2=1$.


\begin{proof}
The relations
\begin{align}
& {\cal D}_R(\phi)
=\int_G R(e,\hat{g})
|{\cal F}^{-1}[\phi](\hat{g})|^2
\mu_G(d\hat{g})\nonumber \\
=& \int_G R(e,\hat{g})
|
\sum_{\lambda \in \hat{G}[{\cal L}]}
\Tr f_\lambda (\hat{g}) c_\lambda {d_\lambda} \Phi_\lambda 
|^2
\mu_G(d\hat{g})\nonumber \\
=& \int_G R(e,\hat{g})
\sum_{\lambda,\lambda' \in \hat{G}[{\cal L}]}
c_\lambda c_{\lambda'}
\Tr f_{\lambda} (\hat{g}) \otimes f_{{\lambda'}^*}(\hat{g})
{d_\lambda} \Phi_\lambda \otimes 
{d_{\lambda'}} \Phi_{\lambda'}^\dagger  
\mu_G(d\hat{g})\nonumber\\
=& 
\sum_{\lambda,\lambda' \in \hat{G}[{\cal L}]}
c_\lambda c_{\lambda'}\sqrt{d_\lambda d_{\lambda'}}
\Tr
[\int_G R(e,\hat{g})
f_{\lambda} (\hat{g}) \otimes f_{{\lambda'}^*}(\hat{g})
\mu_G(d\hat{g})]
\sqrt{d_\lambda} \Phi_\lambda \otimes \sqrt{d_{\lambda'}} \Phi_{\lambda'}^\dagger
\nonumber 
\\
=& 
\sum_{\lambda,\lambda' \in \hat{G}[{\cal L}]}
c_\lambda c_{\lambda'}\sqrt{d_\lambda d_{\lambda'}}
\Tr
\Xi_{\lambda,{\lambda'}^*}
\sqrt{d_\lambda} \Phi_\lambda \otimes \sqrt{d_{\lambda'}} \Phi_{\lambda'}^\dagger
  \Label{6-10-6}
\end{align}
hold, where
we use the matrix
$\Xi_{\lambda,{\lambda'}^*}
:=\int_{G}
-R(e,\hat{g})
f_{\lambda} (\hat{g}) \otimes 
f_{{\lambda'}^*} (\hat{g})
\mu_{G}(d \hat{g})$.
Using the formula
$\int_G \chi_\lambda(g) f_{\lambda'} (g) \mu_G(dg)
=
\int_G \overline{\chi_{\lambda^*}(g)} f_{\lambda'} (g^{-1}) \mu_G(dg)=$
$
\frac{\delta_{{\lambda}^*,\lambda'}}{d_{\lambda^*}}I_{{\lambda}^*}$,
we obtain
\begin{align*}
& \Xi_{\lambda,{\lambda'}^*}
=
\int_{G}
\sum_{\lambda'' \in \hat{G}}
a_{\lambda''} 
(\chi_{\lambda''}(\hat{g}  )+ \chi_{{\lambda''}^*}(\hat{g} ))
f_{\lambda} (\hat{g}) 
\otimes f_{{\lambda'}^*} (\hat{g})
\mu_{G}(d \hat{g})\\
=&
\sum_{\lambda'' \in \hat{G}}
a_{\lambda''} 
\int_{G}
(\chi_{\lambda''}(\hat{g})+ \chi_{{\lambda''}^*}(\hat{g}))
f_{\lambda} (\hat{g}) \otimes f_{{\lambda'}^*} (\hat{g})
\mu_{G}(d \hat{g})\\
=&
\sum_{\lambda'' \in \hat{G}}
\frac{a_{\lambda''} }{d_{\lambda''}}
(C_{\lambda,{\lambda'}^*}^{{\lambda''}^*}I_{{\lambda''}^*}
+C_{\lambda,{\lambda'}^*}^{\lambda''}I_{\lambda''}
)
\ge 0.
\end{align*}
Hence, applying Schwarz inequality with respect to the inner product 
$\langle A,B \rangle:=
\Tr 
\Xi_{\lambda,{\lambda'}^*}
A^{\dagger} B$,
to the case
$A:= I_\lambda \otimes \sqrt{d_{\lambda'}}  \Phi_{\lambda'} $,
$B:= \sqrt{d_\lambda} \Phi_\lambda \otimes I_{\lambda'}$,
we obtain
\begin{align*}
& 
\Tr \Xi_{\lambda,{\lambda'}^*} 
\sqrt{d_\lambda} 
\Phi_\lambda 
\otimes 
\sqrt{d_{{\lambda'}^*}} 
\Phi_{\lambda'}^\dagger  
\\
\le &
\sqrt{
\Tr 
\Xi_{\lambda,{\lambda'}^*} 
d_\lambda 
\Phi_\lambda^\dagger \Phi_\lambda 
\otimes 
I_{{\lambda'}^*}
} 
 \sqrt{
\Tr 
\Xi_{\lambda,{\lambda'}^*} 
I_{\lambda}
\otimes 
d_{{\lambda'}^*} 
\Phi_{\lambda'}^\dagger \Phi_{\lambda'} 
} .
\end{align*}
Since $\Xi_{\lambda^*,{\lambda'}} $ is invariant with respect to the action of $G$,
\begin{align*}
& \Tr 
\Xi_{\lambda,{\lambda'}^*} 
d_\lambda 
\Phi_\lambda^\dagger \Phi_\lambda 
\otimes 
I_{\lambda'} \\
= &
\Tr 
\Xi_{\lambda,{\lambda'}^*} 
\int_G (f_{\lambda}(g)\otimes
f_{{\lambda'}^*}(g))
d_\lambda 
\Phi_\lambda^\dagger \Phi_\lambda 
\otimes 
I_{\lambda'}
(f_{\lambda}(g)\otimes
f_{{\lambda'}^*}(g))^{\dagger}
\mu_G(d g) \\
=&
\Tr 
\Xi_{\lambda,{\lambda'}^*} 
I_{\lambda}
\otimes 
I_{{\lambda'}^*}
=
\sum_{\lambda'' \in \hat{G}}
a_{\lambda''} 
(C_{\lambda,{\lambda'}^*}^{{\lambda''}^*}
+C_{\lambda,{\lambda'}^*}^{\lambda''}),
\end{align*}
where we used the condition $\Tr \Phi_\lambda^\dagger \Phi_\lambda=1$.
Similarly, we have
\begin{align*}
\Tr 
\Xi_{\lambda,{\lambda'}^*} 
I_{\lambda}
\otimes 
d_{{\lambda'}^*} 
\Phi_{\lambda'}^\dagger \Phi_{\lambda'} 
=&
\sum_{\lambda'' \in \hat{G}}
a_{\lambda''} 
(C_{\lambda,{\lambda'}^*}^{{\lambda''}^*}
+C_{\lambda,{\lambda'}^*}^{\lambda''}),
\end{align*}
which implies
\begin{align*}
\Tr 
\Xi_{\lambda,{\lambda'}^*} 
d_\lambda 
\Phi_\lambda^\dagger \Phi_\lambda 
\otimes 
I_{\lambda'} 
\le
\sum_{\lambda'' \in \hat{G}}
a_{\lambda''} 
(C_{\lambda,{\lambda'}^*}^{{\lambda''}^*}
+C_{\lambda,{\lambda'}^*}^{\lambda''}).
\end{align*}
Hence, combining the above relation with
(\ref{6-10-6}),
we obtain (\ref{11-16-1}).
Due to the equality condition for Schwarz inequality,
the equality in (\ref{11-16-1}) holds when
$|\phi\rangle = \oplus_{\lambda \in \hat{G}[{\cal L}]}c_\lambda |\Psi_\lambda\rangle\rangle $.

Next, we show (\ref{5-4-1}) when $|\phi\rangle=\oplus_{\lambda\in \hat{G}[{\cal L}]}  c_\lambda |\Psi_\lambda\rangle\rangle $.
Any covariant measurement $M$ can be written as $M_T$ by using an operator $T$.
We make a decomposition of the operator $T$ 
as $T =\sum_{k}|\eta_k\rangle \langle \eta_k|$.
Then, we use the notations 
$|\phi\rangle=\oplus_{\lambda \in \in \hat{G}[{\cal L}]}| \phi_\lambda \rangle\rangle$,
$|\eta_k\rangle=\oplus_{\lambda\in \in \hat{G}[{\cal L}]}|\eta_{k,\lambda} \rangle\rangle $,
$|x_k\rangle:=\oplus_{\lambda\in \in \hat{G}[{\cal L}]} |\phi_\lambda
\eta_{k,\lambda}^\dagger\rangle\rangle$.
Using (\ref{11-16-1}) and its equality condition, we have
\begin{align*}
& {\cal D}_R(|\phi\rangle\langle \phi|, M_T)
=
\sum_{k}
\|x_k\|^2
{\cal D}_R(\frac{|x_k\rangle\langle x_k|}{\|x_k\|^2} , 
M_{|{\cal I}\rangle \langle {\cal I}|})
=
\sum_{k}
\|x_k\|^2
{\cal D}_R(\frac{|x_k\rangle}{\|x_k\|} ) \\
\ge &
{\cal D}_R(|\phi\rangle),
\end{align*}
which implies (\ref{5-4-1}).
In this case,
\begin{align*}
{\cal F}^{-1}[\phi](g)
=\sum_{\lambda \in \hat{G}[{\cal L}]}
\Tr f_\lambda (g) d_{\lambda} c_{\lambda} \Psi_\lambda 
=\sum_{\lambda \in \hat{G}[{\cal L}]}
\Tr f_\lambda (g) \sqrt{d_{\lambda}} c_{\lambda} 
=\sum_{\lambda \in \hat{G}[{\cal L}]}
\sqrt{d_{\lambda}} c_{\lambda} \chi_{\lambda}(g),
\end{align*}
which implies (\ref{5-4-2}).
\end{proof}

\section{Real numbers $\real$}\Label{s8}
\subsection{Energy constraint}\Label{s8-1}
In this section, 
as a typical example of commutative group, we treat the real group $G=\real$.
In this case, $\hat{G}$ is also $\real$ and is called the weight space.
That is, since $L^2(\hat{G})=L^2(\real)$, the input state is given as a wave function $\phi(\lambda)$ on the space $L^2(\real)$.
Here, we choose the invariant measure $\mu_{\real}(d\hat{g})=\frac{1}{\sqrt{2\pi}} d\hat{g}$ on $\real$.
Then, 
$\mu_{\hat{\real}}(d\lambda)=\frac{1}{\sqrt{2\pi}} d \lambda$.
The input state $\phi$ satisfies
$\int_{\real} |\phi(\lambda)|^2 \frac{1}{\sqrt{2\pi}} d \lambda =1$.
The covariant POVM $M_{|{\cal I}\rangle \langle {\cal I}|}$ is 
the spectral decomposition of the position operator $Q$ on $L^2(\real)$.
When the true parameter is $0$, 
the estimate $\hat{g}\in \real$ obeys the distribution
$|{\cal F}^{-1}[\phi](\hat{g})|^2 \frac{1}{\sqrt{2\pi}} d\hat{g}$,
which is given by the inverse Fourier transform ${\cal F}^{-1}[\phi]$ of 
$\phi$.

We minimize the average of the square error $(\hat{g}-g)^2$,
which is calculated to 
$\int \hat{g}^2 |{\cal F}^{-1}[\phi](\hat{g})|^2 \frac{d \hat{g}}{\sqrt{2\pi}}$. 
In this setting, it is natural to restrict the average energy for input state $\phi$, i.e.,
we impose the constraint
\begin{align}
\int \lambda^2 |\phi(\lambda)|^2 \frac{1}{\sqrt{2\pi}} d \lambda \le E
\Label{3-13-1}
\end{align}
for a given constant $E$.

When the uncertainty of operator $X$ is defined
as $\Delta_\phi^2 X = 
\langle \phi | (X-\langle \phi | X | \phi \rangle )^2 | \phi \rangle$
and the momentum operator $P$ is defined as $P:= -i \frac{d}{d\lambda}$,
we have the uncertainty relation between $\Delta_\phi^2 Q$ and 
$\Delta_\phi^2 P$ as
\begin{align}
& \min_{|\phi \rangle \in L^2_{n}(\real) }
\left\{\left.
\langle \phi | P^2 | \phi \rangle 
\right|
\langle \phi | Q^2 | \phi \rangle \le E
\right\} 
=
 \min_{|\phi \rangle \in L^2_{n}(\real) }
\left\{\left.
\Delta_\phi^2 P
\right|
\Delta_\phi^2 Q
 \le E \right\} =\frac{1}{4E} \Label{5-12-1}.
\end{align}
Thus, the combination of \eqref{th7}, \eqref{5-12-1}, and Theorem \ref{t6-24-1} 
yields the following theorem.
\begin{theorem}\Label{T5-12-1}
The relations
\begin{align}
&\min_{\rho \in {\cal S}(L^2(\real))} 
\min_{M \in {\cal M}_{\cov}(\real)} 
\{{\cal D}_R(\rho,M)| \Tr \rho Q^2 \le E \}
\nonumber \\
=&
\min_{ \{p_i\} }
\min_{\rho_i \in {\cal S}(L^2(\real))} 
\min_{M_i \in {\cal M}_{\cov}(\real)} 
\{\sum_i p_i{\cal D}_R(\rho_i,M_i)| \sum_i p_i \Tr \rho_i Q^2 \le E \}
\nonumber \\
=&
\min_{|\phi \rangle \in L^2_n(\real) }
\left\{\left.
\int_{-\infty}^{\infty}
\hat{g}^2 |{\cal F}^{-1}[\phi](\hat{g})|^2 \frac{d \hat{g}}{\sqrt{2\pi}}
\right|
\int_{-\infty}^{\infty} \lambda^2 |\phi(\lambda)|^2 \frac{1}{\sqrt{2\pi}} d \lambda \le E
\right\}\nonumber \\
=&
\min_{|\phi \rangle \in L^2_n(\real) }
\left\{\left.
\langle \phi | P^2 | \phi \rangle 
\right|
\langle \phi | Q^2 | \phi \rangle \le E
\right\} 
=\frac{1}{4E} \Label{5-13-14}
\end{align}
hold.
The minimum $\frac{1}{4E}$ is attained if and only if the input state $\phi(\lambda)$
is $\frac{1}{E^{1/4}}e^{-\frac{\lambda^2}{4 E}}$,
whose inverse Fourier transform is $(2 E)^{1/4} e^{-E \lambda^2}$.
\end{theorem}
That is, in Schemes 1 and 2, the minimum average square error with the energy constraint (\ref{3-13-1})
is $\frac{1}{4E}$.

Now, we consider two systems ${\cal H}_i$ ($i=1,2$) equivalent with $L^2(\real)$ with the Hamiltonian $Q^2$.
We focus on the composite system ${\cal H}_1 \otimes {\cal H}_2$
with the Hamiltonian
$(Q \otimes I+ I \otimes Q)^2 =Q^2 \otimes I+ 2 Q \otimes Q +I \otimes Q^2$,
which has a strong interaction term $2 Q \otimes Q$.
In this case, 
the optimal estimation 
in the composite system with the energy $E_1+E_2$
can be realized by the following way.
Let the input state $|\phi_i\rangle$ be the optimal input state with the energy $E_i$.
In this case, 
the input state $
|\phi_1\otimes \phi_2\rangle=
|\phi_1\rangle \otimes |\phi_2\rangle$ has the energy $E_1+E_2$
because
\begin{align*}
&\langle \phi_1\otimes \phi_2|(Q \otimes I+ I \otimes Q)^2 |\phi_1\otimes \phi_2\rangle \\
=&
\langle \phi_1\otimes \phi_2|Q^2 \otimes I |\phi_1\otimes \phi_2\rangle
+2 \langle \phi_1\otimes \phi_2|Q \otimes Q |\phi_1\otimes \phi_2\rangle
+\langle \phi_1\otimes \phi_2|I \otimes Q^2 |\phi_1\otimes \phi_2\rangle \\
=&
\langle \phi_1|Q^2 |\phi_1\rangle
+2 \langle \phi_1|Q |\phi_1\rangle\langle \phi_2|Q |\phi_2\rangle
+\langle \phi_2|Q^2 |\phi_2\rangle \\
=& E_1+E_2.
\end{align*}
Note that $\langle \phi_i|Q |\phi_i\rangle=0$ for $i=1,2$
because
$\phi_i$ has the Gaussian form with the average $0$.
The state $|\phi_1\otimes \phi_2\rangle$
realizes the optimal estimator in the composite system ${\cal H}_1 \otimes {\cal H}_2$ with the energy $E_1+E_2$
by employing the following measurement.
First, we measure the position operator $Q_i$ in the respective system and denote the outcome by $X_i$.
When the unknown parameter to be estimated is $\theta$,
$X_i$ obeys the Gaussian distribution with the variance $\frac{1}{4E_i}$ and the average $\theta$.
The value $X=\frac{E_1 X_1+E_2 X_2}{E_1+E_2}$ obeys 
the Gaussian distribution with the variance $\frac{\frac{E_1}{4E_1}+ \frac{E_2}{4E_2}}{E_1+E_2}=
\frac{1}{4(E_1+E_2)}$ and the average $\theta$.
Hence, the value $X$ realizes the optimal estimate with the energy $E_1+E_2$. 
That is, 
we can realize the optimal estimator by the combination of the optimal estimators of the individual systems.

\subsection{Energy constraint and positivity constraint}
In the above setting, we assume that all irreducible representations are available while the energy constraint is imposed.
Next, we assume that only the irreducible representation with positive $\lambda >0$ is available.
That is, the real number $\lambda$ is restricted to $\real_+:=\{x \in \real| x \ge 0\}$.
This problem with the energy constraint can be solved in the following way.
\begin{theorem}\Label{T3-13-11}
The relations
\begin{align}
&\min_{\rho \in {\cal S}(L^2(\real_+))} 
\min_{M \in {\cal M}_{\cov}(\real)} 
\{{\cal D}_R(\rho,M)| \Tr \rho Q^2 \le E \}\nonumber \\
=&
\min_{ \{p_i\} }
\min_{\rho_i \in {\cal S}(L^2(\real_+))} 
\min_{M_i \in {\cal M}_{\cov}(\real)} 
\{\sum_i p_i{\cal D}_R(\rho_i,M_i)| \sum_i p_i \Tr \rho_i Q^2 \le E \}
\nonumber \\
=&
\min_{|\phi \rangle \in L^2_n(\real_+) }
\left\{\left.
\int_{-\infty}^{\infty}
\hat{g}^2 |{\cal F}^{-1}[\phi](\hat{g})|^2 \frac{d \hat{g}}{\sqrt{2\pi}}
\right|
\int_{0}^{\infty} \lambda^2 |\phi(\lambda)|^2 \frac{d \lambda }{\sqrt{2\pi}} 
\le E
\right\} \nonumber \\
=&
\min_{|\phi \rangle \in L^2_n(\real_+) }
\left\{\left.
\langle \phi | P^2 | \phi \rangle 
\right|
\langle \phi | Q^2 | \phi \rangle \le E
\right\} \nonumber \\
=& \min_{|\tilde{\phi} \rangle \in L^2_{\odd,n}(\real) }
\left\{\left.
\langle \tilde{\phi} | P^2 | \tilde{\phi} \rangle 
\right|
\langle \tilde{\phi} | Q^2 | \tilde{\phi} \rangle \le E
\right\} 
=\frac{9}{4E}
\label{5-13-16}
\end{align}
hold, where
a function $\phi \in L^2(\real_+)$ is regarded as an element on $L^2(\real)$ in the following sense.
\begin{align}
\phi(x):= 
\left\{
\begin{array}{ll}
\phi(x) & \hbox{ if } x \ge 0 \\
0 & \hbox{ if } x < 0 .
\end{array}
\right. \Label{3-13-10}
\end{align}
The minimum $\frac{9}{4E^2}$ is attained if and only if the input state $\phi(\lambda)$
is $\sqrt{2}(\frac{3 }{E})^{3/4} \lambda e^{-\frac{3 \lambda^2}{4E}}
\in L^2_n(\real_+)$,
whose inverse Fourier transform is 
$i (\frac{4 E}{3})^{3/4} \lambda e^{-\frac{\hat{g}^2 E}{3}}$.
\end{theorem}
That is, in Schemes 1 and 2, the minimum average square error is $\frac{9}{4E}$
when we consider the energy constraint (\ref{3-13-1}) and the positivity constraint for $\lambda$.

\begin{proof}
The first and the second equation follow from Theorem \ref{t6-24-1} and \eqref{th7}.
The third equation can be shown by the correspondence (\ref{3-13-10}).
The fourth equation can be shown by Lemma \ref{L3-13-1} and considering 
the following odd function $\tilde{\phi}$ for any function $\phi \in L^2(\real_+)$.
\begin{align}
\tilde{\phi}(x):= 
\left\{
\begin{array}{ll}
\frac{1}{\sqrt{2}}\phi(x) & \hbox{ if } x >0 \\
-\frac{1}{\sqrt{2}} \phi(-x) & \hbox{ if } x < 0 \\
0 & \hbox{ if } x = 0 .
\end{array}
\right.
\end{align}
Using the above correspondence and Lemma \ref{L3-13-1},
we can show that the minimum can be attained by 
$\phi(\lambda):=\sqrt{2}(\frac{3 }{E})^{3/4} \lambda e^{-\frac{3 \lambda^2}{4E}}
\in L^2_n(\real_+)$.
\end{proof}

\begin{lemma}\Label{L3-13-1}
The relation 
\begin{align}
& \min_{|\phi \rangle \in L^2_{\odd,n}(\real) }
\left\{\left.
\langle \phi | P^2 | \phi \rangle 
\right|
\langle \phi | Q^2 | \phi \rangle \le E
\right\} 
=\frac{9}{4E}\Label{5-9-3}
\end{align}
holds, where
the minimum value is $\frac{9}{4E}$ 
is attained only by the wave function
$\phi(\lambda):=
(\frac{3 }{E})^{3/4} \lambda e^{-\frac{3 \lambda^2}{4E}}
(\in L^2_{\odd}(\real))$
whose inverse Fourier transform is $i (\frac{4 E}{3})^{3/4} \lambda e^{-\frac{\hat{g}^2 E}{3}}$.
\end{lemma}

Since
\begin{align}
& \min_{|\phi \rangle \in L^2_{\odd,n}(\real) }
\left\{\left.
\langle \phi | P^2 | \phi \rangle 
\right|
\langle \phi | Q^2 | \phi \rangle \le E
\right\} 
= \min_{|\phi \rangle \in L^2_{\odd,n}(\real) }
\left\{\left.
\Delta_\phi^2 P
\right|
\Delta_\phi^2 Q
 \le E \right\} 
\end{align}
Lemma \ref{L3-13-1} can be regarded as the uncertainty relation between 
$\Delta_\phi^2 Q$ and $\Delta_\phi^2 P$ among odd functions on $\bR$.

\begin{proof}
Since the condition of Lemma \ref{L4-25-2} hold, 
\begin{align*}
\kappa_1(E):=&
\min_{|\phi \rangle \in L^2_n(\real_+) }
\left\{\left.
\langle \phi | P^2 | \phi \rangle 
\right|
\langle \phi | Q^2 | \phi \rangle \le E
\right\} \nonumber \\
=& \min_{|\tilde{\phi} \rangle \in L^2_{\odd,n}(\real) }
\left\{\left.
\langle \tilde{\phi} | P^2 | \tilde{\phi} \rangle 
\right|
\langle \tilde{\phi} | Q^2 | \tilde{\phi} \rangle \le E
\right\} 
\end{align*}
is convex.
Hence, we employ Lemma \ref{L5-10-1} to calculate $\kappa_1(E)$.
For this purpose, we consider a squeezed number operator
$\frac{1}{2}(P^2/t+ tQ^2)-\frac{1}{2}$.
The minimum eigenvalue in $L^2_{\odd}(\real)$ is $1$
and the corresponding eigenvector is the squeezed one-photon state.
Then, we obtain
\begin{align}
&\gamma_1(s):=
\min_{\phi \in L^2_{\odd,n}(\real)}
 \langle \phi |P^2 | \phi \rangle 
+s \langle \phi |Q^2 | \phi \rangle \nonumber \\
=&
\sqrt{s} \min_{\phi \in L^2_{\odd,n}(\real)}
\frac{1}{\sqrt{s}} \langle \phi |P^2 | \phi \rangle 
+\sqrt{s} \langle \phi |Q^2 | \phi \rangle 
=3 \sqrt{s}
\Label{5-9-1}
\end{align}
Solving the equation $\gamma_1'(s_E)=E$, we have
$s_E= \frac{9}{4 E^2}$.
Hence, 
\begin{align*}
\kappa_1(E)= \gamma_1(s_E)-s_E E=
\frac{9}{2E} -\frac{9}{4E}=\frac{9}{4E} ,
\end{align*}
which implies \eqref{5-9-3}.
Since the minimum \eqref{5-9-1} with $s=\frac{9}{4 E^2}$
is attained only by 
$\phi(\lambda)=(\frac{3 }{E})^{3/4} \lambda e^{-\frac{3 \lambda^2}{4E}}$,
the minimum in \eqref{5-9-3} is attained only by 
$\phi(\lambda)=
(\frac{3 }{E})^{3/4} \lambda e^{-\frac{3 \lambda^2}{4E}}$.
\end{proof}

\subsection{Interval constraint}\Label{s8-3}
As another restriction, we assume that
the support of the input state $\phi(\lambda)$ is 
included in the interval $[-L,L]$.
In this case, we have the following theorem.
\begin{theorem}\Label{T5-12-1b}
The relations
\begin{align}
&\min_{\rho \in {\cal S}(L^2([-L,L]))} 
\min_{M \in {\cal M}_{\cov}(\real)} 
{\cal D}_R(\rho,M)
\nonumber \\
=&
\min_{ \{p_i\} }
\min_{\rho_i \in {\cal S}(L^2([-L,L]))} 
\min_{M_i \in {\cal M}_{\cov}(\real)} 
\sum_i p_i{\cal D}_R(\rho_i,M_i)
\nonumber \\
=&
\min_{|\phi \rangle \in L^2_n([-L,L]) }
\int_{-\infty}^{\infty}
\hat{g}^2 |{\cal F}^{-1}[\phi](\hat{g})|^2 \frac{d \hat{g}}{\sqrt{2\pi}}
=
\min_{|\phi \rangle \in L^2_n([-L,L]) }
\langle \phi | P^2 | \phi \rangle 
=
\frac{\pi^2}{4L^2}
\Label{5-12-10}
\end{align}
hold.
The minimum value $\frac{\pi^2}{4L^2}$
is attained only by the input state
$\phi(\lambda)=
\frac{(2 \pi)^{1/4}}{\sqrt{L}}\sin \frac{\pi (1+\frac{\lambda}{L})}{2}$
whose inverse Fourier transform is 
$-\frac{\pi^{3/4}\sqrt{L}}{2^{1/4}} \frac{\cos L \hat{g}}{L^2 \hat{g}^2-\pi^2/4}$.
\end{theorem}

Now, we define the maximum uncertainty of an operator $X$ 
by $\Delta_{\phi,\max}X:= \inf \{a\in \bR| 
\langle \phi | \{ |X- \langle \phi | X | \phi \rangle | > a\} |\phi \rangle =0 \}$,
where 
the projection $\{|X| > a\}$ is defined by using the spectral measure $E_X$ of $X$ 
as $\{|X| > a\} := \int_{x >a} E_X(dx)$. 
Since
\begin{align}
\min_{|\phi \rangle \in L^2_n([-L,L]) }
\langle \phi | P^2 | \phi \rangle 
=
\min_{|\phi \rangle \in L^2_{n}(\real) }
\left\{\left.
\Delta_\phi^2 P
\right|
\Delta_{\phi,\max} Q \le L \right\} 
\end{align}
Theorem \ref{T5-12-1b} can be regarded as 
the uncertainty relation between 
$\Delta_\phi^2 P$ and $\Delta_{\phi,\max} Q$.

\begin{proof}
Combining Theorem \ref{th3} and \eqref{th7}, 
we obtain the first and the second equations.
The third equation follows from knowledge for Fourier analysis.
As in \cite{Im}, the fourth equation can be shown in the following way.
The restriction of the operator $P^2$ on $L^2([-L, L])$
has the minimum eigenvalue 
$\frac{\pi^2}{4L^2}$
with the eigenvector 
$\frac{(2 \pi)^{1/4}}{\sqrt{L}}\sin \frac{\pi (1+\frac{\lambda}{L})}{2}$
\cite{Coddington}.
Then, we obtain the desired argument.
\end{proof}

Now, we consider the energy of the optimal input state 
$\phi(\lambda)=
\frac{(2 \pi)^{1/4}}{\sqrt{L}}\sin \frac{\pi (1+\frac{\lambda}{L})}{2}$
in Theorem \ref{T5-12-1b}.
\begin{align*}
&\langle \phi | Q^2 | \phi \rangle 
= 
\int_L^L \frac{\lambda^2}{L} \sin^2 \frac{\pi(1+\frac{\lambda}{L})}{2}
d \lambda
= 
\int_L^L \frac{(\lambda+L)^2}{L} \sin^2 \frac{\pi(1+\frac{\lambda}{L})}{2}
d \lambda - L^2 \\
=&
\frac{8 L^2}{\pi^3}
\int_0^\pi x^2 \sin^2 x dx-L
= (\frac{1}{3}-\frac{2}{\pi^2})L^2.
\end{align*}
Hence, the energy 
$\langle \phi | Q^2 | \phi \rangle $
increases with the order $L^2$ when $L$ is large.

\section{Integers $\bZ$}\Label{s9}
As another typical example of commutative group, we treat the real group $G=\bZ$.
The one-dimensional unitary representation is characterized by $e^{i \lambda n}$ with a real number $\lambda
\in \real$.
When the difference between two real numbers
$\lambda$ and $\lambda'$ is an integer times of $2 \pi$,
we obtain $e^{i \lambda n}=e^{i \lambda' n} $.
Hence, $\hat{G}$ is $\U(1)=(-\pi,\pi]$, i.e., $L^2(\hat{G})=L^2(\U(1))$.
That is, the input state is given as a wave function $\phi(\lambda)$ on the space $L^2(\U(1))$.
In this case, the measure on the dual space $\U(1)$ is
$\mu_{\U(1)}(\lambda)=\frac{1}{{2\pi}} d\lambda$.
Hence,
The input state $\phi$ satisfies
$\int_{\U(1)} |\phi(\lambda)|^2 \frac{1}{{2\pi}} d \lambda =1 $.

Now, we define the CONS $\{ |f_n\rangle \}$ of $L^2(\U(1))$
by $f_n(\lambda):= e^{i n\lambda}$.
Then,
the covariant POVM $M_{|{\cal I}\rangle \langle {\cal I}|}$ is 
the PVM $\{ |f_n\rangle \langle f_n|\}$.
When the true parameter is $0$, 
the estimate $\hat{g}\in \U(1)$ obeys the distribution
$|{\cal F}^{-1}[\phi](\hat{g})|^2 \frac{1}{{2\pi}} d\hat{g}$,
which is given by the inverse Fourier transform ${\cal F}^{-1}[\phi]$ of $\phi$.

In this case, when $\phi(\lambda)= 1$,
we have
\begin{align}
{\cal F}^{-1}[\phi](n)
=
\left\{
\begin{array}{ll}
0 & \hbox{ if } n \neq 0 \\ 
1 & \hbox{ if } n =0.
\end{array}
\right.
\end{align}
Hence, if the input $\phi(\lambda)= 1$
is available, 
the perfect discrimination is possible.

\section{One-dimensional unitary group $\U(1)$}\Label{s10}
\subsection{General structure}\Label{s10-1}
Next, we treat the estimation of action of the group $\U(1)$.
The square integrable space $L^2(\U(1))$
can be identified with the space of the periodic square integrable functions
with the period $2L$, which is denoted by $L^2_p((-L,L])$.
Under this correspondence,
the one-dimensional unitary representation 
is characterized by $e^{\lambda i\theta}$ with $\lambda= \frac{\pi}{L}k$ and
$k \in \bZ$.
That is, $\hat{U(1)}$ can be regarded as $\frac{\pi}{L} \bZ$,
and $L^2(\hat{U(1)})$ can be identified with the square summable space 
$\ell^2(\frac{\pi}{L} \bZ)
:=\{ \{a_x\}_{ x \in  \frac{\pi}{L} \bZ} | \sum_{x \in  \frac{\pi}{L} \bZ} |a_x|^2<
\infty\}$.
In this case, we denote the Fourier transform by ${\cal F}_{L}$.
In the following discussion of this section, we consider the case of $L=\pi$.
Then, the input state is given as a wave function $\phi(\lambda)$ on the space $\ell^2(\bZ)$.
Here, we choose the invariant measure $\mu_G(d\hat{\theta})=\frac{1}{{2\pi}} d\hat{\theta}$ on $G$.

The covariant POVM $M_{|{\cal I}\rangle \langle {\cal I}|}$ is 
the spectral decomposition of the position operator $Q$ on $L^2_p((-\pi,\pi])$.
When the true parameter is $0$, 
the estimate $\hat{\theta}\in \U(1)$ obeys the distribution
$|{\cal F}^{-1}[\phi](\hat{\theta})|^2 \frac{1}{2\pi} d\hat{\theta}$,
which is given by the inverse Fourier transform ${\cal F}^{-1}[\phi]$ of 
$\phi$.
Since the amount of error should be invariant under the change of the sign,
it is natural to assume that 
the risk function $R(\theta,\hat{\theta})$ is written as 
\begin{align}
R(\theta,\hat{\theta})=R(0,\hat{\theta}-\theta)=w(\hat{\theta}-\theta)
\Label{5-19-11}
\end{align}
by using an even function $w$ on $(-\pi,\pi]$.
Further, we assume that the function $w$ satisfies the condition (\ref{6-26-14}),
i.e., is written as 
\begin{align}
w(\hat{\theta})= w(0) -\sum_{n=1}^{\infty} c_n \cos n \hat{\theta} 
\Label{5-19-11r}
\end{align}
with $c_n \ge 0$.
Hence, defining the unitary $U_{\sgn}$ on $L^2_p((-\pi,\pi])$ by 
$U_{\sgn}[\varphi](\theta)=\varphi(-\theta)$ for $\varphi \in L^2_p((-\pi,\pi])$,
we have $U_{\sgn} w(Q) U_{\sgn}^{-1}=w(Q)$.
In this case, we can decompose the space $L^2_p((-\pi,\pi])$
as $L^2_{p,\even}((-\pi,\pi])\oplus L^2_{p,\odd}((-\pi,\pi])$,
where $L^2_{p,\even}((-\pi,\pi])$ ($L^2_{p,\odd}((-\pi,\pi])$) is 
the space of even (odd) functions in $L^2_{p}((-\pi,\pi])$.
For example, we often employ the risk function
$R_{1,\U(1)}(0,\hat{\theta}):=1- \cos \hat{\theta}$.
Further, it is also natural to assume that
the Hamiltonian $H$ is written as
\begin{align}
H=h(0)I_0+ \sum_{k=1}^{\infty} h(k^2) (I_k +I_{-k})
\Label{5-19-10}
\end{align}
by using a function $h$.
Note that the support restriction case 
can be realized by the energy constraint with a proper Hamiltonian.

\begin{theorem}\Label{T5-20}
Under the assumptions \eqref{5-19-11}, \eqref{5-19-11r}, and \eqref{5-19-10},
the relations
\begin{align}
&\min_{\rho \in {\cal S}({\cal K}_{\hat{\U(1)}})} 
\min_{M \in {\cal M}_{\cov}(\U(1))} 
\{{\cal D}_R(\rho,M)| \Tr \rho H \le E \}
\nonumber \\
=&
\min_{ \{p_i\} }
\min_{\rho_i \in {\cal S}({\cal K}_{\hat{\U(1)}})}
\min_{M_i \in {\cal M}_{\cov}(\U(1))} 
\{\sum_i p_i{\cal D}_R(\rho_i,M_i)| \sum_i p_i \Tr \rho_i H \le E \}
\nonumber \\
=&\min_{|\phi\rangle \in \ell^2_n(\bZ)} 
\{ {\cal D}_R(|\phi\rangle)|
\Tr \langle \phi |H |\phi\rangle \le E 
\}\nonumber \\
=&
\min_{\varphi \in L^2_{p,\even,n}((-\pi,\pi])} 
\{\langle \varphi| w(Q) | \varphi\rangle|
\langle \varphi|h(P^2)| \varphi\rangle \le E \}
\Label{5-18-2d}
\end{align}
hold.
Further, an input state 
$|{\phi} \rangle:=
\oplus_{k=-\infty}^{\infty} {\beta}_{k}|k\rangle$
satisfies the relation
\begin{align}
\min_{M\in {\cal M}_{\cov}(\U(1))} 
{\cal D}_R (|{\phi}\rangle \langle {\phi}|,M)
= {\cal D}_R (|{\phi}\rangle)
= \eqref{5-18-2d}
\end{align}
if and only if
${\cal F}^{-1}[\phi]$ 
is an even function and realizes the minimum \eqref{5-18-2d}.

Additionally,
when $H
=\sum_{k=-\infty}^{\infty} k^2 I_{k}$,
i.e., $h(x)=x$,
we have
\begin{align}
\eqref{5-18-2d}
=\min_{\varphi \in L^2_{p,\even,n}((-\pi,\pi])} 
\{\langle \varphi| w(Q) | \varphi\rangle|
\langle \varphi|P^2| \varphi\rangle \le E \}  
\Label{5-18-3d}
\end{align}
\end{theorem}

\begin{proof}
Any odd function $\varphi_o \in L^2_{p,\odd,n}((-\pi,\pi])$
satisfies that
$\langle \varphi_o | w(Q)|\varphi_o\rangle=w(0)$.
Due to \eqref{5-19-11r},
any even function $\varphi_e \in L^2_{p,\even,n}((-\pi,\pi])$
satisfies that 
\begin{align}
\langle \varphi_o | w(Q)|\varphi_o\rangle=w(0)
\ge \langle \varphi_e | w(Q)|\varphi_e\rangle
\Label{5-18}
\end{align}
because 
$ \langle \varphi_e | \cos n Q|\varphi_e\rangle \ge 0$.
So, combining (\ref{5-19-10}) and this fact, we have
\begin{align*}
& \min_{\varphi \in L^2_{p,\even,n}((-\pi,\pi])} 
\{\langle \varphi| w(Q) | \varphi\rangle|
\langle \varphi|{\cal F}^{-1} H {\cal F} | \varphi\rangle \le E \} \\
\le &
\min_{\varphi \in L^2_{p,\odd,n}((-\pi,\pi])} 
\{\langle \varphi| w(Q) | \varphi\rangle|
\langle \varphi|{\cal F}^{-1} H {\cal F} | \varphi\rangle \le E \}).
\end{align*}
Since
\begin{align}
{\cal D}_R(|\phi\rangle)
= \int_{-\pi}^{\pi}w(\hat{\theta})|{\cal F}^{-1}[\phi](\hat{\theta})|^2 
\frac{d\hat{\theta}}{2\pi}
=\langle \varphi| w(Q) | \varphi\rangle,
\end{align}
Theorem \ref{t6-24-1} and \eqref{th7} implies that
\begin{align}
&\min_{|\phi\rangle \in \ell^2 (\bZ)} 
\{ {\cal D}_R(|\phi\rangle)|
\Tr \langle \phi |H |\phi\rangle \le E 
\}\nonumber \\
=&
\min_{\varphi \in L^2_{p,n}((-\pi,\pi])} 
\{\langle \varphi| w(Q) | \varphi\rangle|
\langle \varphi|{\cal F}^{-1} H {\cal F} | \varphi\rangle \le E \} \nonumber \\
=&
\min \Bigl(
\min_{\varphi \in L^2_{p,\even,n}((-\pi,\pi])} 
\{\langle \varphi| w(Q) | \varphi\rangle|
\langle \varphi|{\cal F}^{-1} H {\cal F} | \varphi\rangle \le E \},\nonumber \\
& \quad 
\min_{\varphi \in L^2_{p,\odd,n}((-\pi,\pi])} 
\{\langle \varphi| w(Q) | \varphi\rangle|
\langle \varphi|{\cal F}^{-1} H {\cal F} | \varphi\rangle \le E \}\Bigr) \nonumber \\
=&
\min_{\varphi \in L^2_{p,\even,n}((-\pi,\pi])} 
\{\langle \varphi| w(Q) | \varphi\rangle|
\langle \varphi|{\cal F}^{-1} H {\cal F} | \varphi\rangle \le E \}.
\Label{5-18-1}
\end{align}
Since \eqref{5-19-10} implies 
$\langle \varphi|{\cal F}^{-1} H {\cal F} | \varphi\rangle 
=\langle \varphi|h(P^2) | \varphi\rangle $, 
combining Theorem \ref{t6-24-1} and Lemma \ref{L4-25-2},
we obtain \eqref{5-18-2d}.
\end{proof}

\subsection{Constraint for available irreducible representations}\Label{s10-1-2}
Since the error can be reduced infinitesimally with the infinite support of the input state,
it is natural to restrict the support of $\phi$ to $\Lambda_n:= \{k\in \bZ| |k|\le n\}$.
Now, we treat the error function
$1-\cos (\hat\theta-\theta)=
1-(e^{i(\theta-\hat\theta)}+e^{-i(\theta-\hat\theta)})/2
= 2\sin^2(\frac{\theta-\hat\theta}{2})$,
which satisfies the condition (\ref{6-26-14}).
When the input state is $( \phi(\lambda) )_{\lambda=-n}^n$,
the average error is calculated to
\begin{align}
1-
\sum_{\lambda=-n}^{n-1}
\frac{1}{2}
\phi(\lambda) \phi(\lambda+1)
-
\sum_{\lambda=-n+1}^{n}
\frac{1}{2}
\phi(\lambda-1) \phi(\lambda)
=
1-\sum_{\lambda=-n}^{n-1}
\phi(\lambda) \phi(\lambda+1).
\Label{3-13-19}
\end{align}
Then, we obtain the following theorem.
\begin{theorem}\Label{T5-12-4}
The relations
\begin{align}
&
\min_{\rho \in {\cal S}({\cal K}_{\Lambda_n})} 
\min_{M \in {\cal M}_{\cov}(\U(1))} 
{\cal D}_R(\rho,M)
\nonumber \\
=&
\min_{ \{p_i\} }
\min_{\rho_i \in {\cal S}({\cal K}_{\Lambda_n})} 
\min_{M_i \in {\cal M}_{\cov}(\U(1))} 
\sum_i p_i {\cal D}_{R}(\rho_i,M_i) 
\nonumber \\
=&
\min_{|\phi \rangle \in {\cal K}_{\Lambda_n},n} 
{\cal D}_{R}(|\phi \rangle)
=
1-\cos \frac{\pi}{2n+2}
\Label{3-12-3}
\end{align}
hold.
The minimum $1-\cos \frac{\pi}{2n+2}$ is attained 
by the input state $|\phi\rangle$
with the measurement ${\cal M}_{|I\rangle \langle I|}$
if and only if the input state $\phi(\lambda)$ is $C \sin\frac{\pi(\lambda+n+1)}{2n+2}$
with the normalizing constant $C$.
\end{theorem}
\begin{proof}
Thanks to Theorem \ref{th3} and \eqref{th7} and (\ref{3-13-19}),
the above theorem can be shown from Lemma \ref{t6-10-1} in Appendix \ref{a5-18}
with $m=2n+1$.
\end{proof}

Hence, the above minimum error is 
$1-\cos \frac{\pi}{2n+2}$, which is attained by $\phi(\lambda)= C \sin\frac{\pi(\lambda+n+1)}{2n+2}$ with the normalizing constant $C$.
Since
$1-\cos \frac{\pi}{2n+2}\cong 1-(1-\frac{1}{2}(\frac{\pi}{2n+2})^2)
\cong \frac{\pi^2}{8 n^2}$,
we have the following asymptotic characterization.
\begin{align*}
&
\lim_{n \to \infty}n^2
\min_{\rho \in {\cal S}({\cal K}_{\Lambda_n})} 
\min_{M \in {\cal M}_{\cov}(\U(1))} 
{\cal D}_R(\rho,M)
\nonumber \\
=&
\lim_{n \to \infty}n^2
\min_{ \{p_i\} }
\min_{\rho_i \in {\cal S}({\cal K}_{\Lambda_n})} 
\min_{M_i \in {\cal M}_{\cov}(\U(1))} 
\sum_i p_i {\cal D}_{R}(\rho_i,M_i) 
=
\frac{\pi^2}{8 }.
\end{align*}

For the asymptotic optimality condition with respect to input states,
we obtain the following lemma.
\begin{lemma}\Label{T3-13-3c}
For a sequence $\{\phi_n\}$ in ${\cal K}_{\Lambda_n}$
satisfying $\phi_n(k)\ge 0$, 
the relation
$\min_{M\in {\cal M}_{\cov}(\U(1))} 
{\cal D}_R (|\phi_{n}\rangle \langle \phi_{n}|,M)
=
{\cal D}_R (|\phi_{n}\rangle)
\cong \frac{1}{8 n^2} $ holds as $n \to \infty$,
if and only if the sequence of functions
$\tilde{\phi}_n(\lambda):= (2\pi)^{1/4}\sqrt{n}\phi_{n}(\lfloor n \lambda  +\frac{1}{2}\rfloor) 
\in L^2([-1,1])$
satisfies that
$\tilde{\phi}_n(\lambda)$ goes to $(2\pi)^{1/4} \sin \frac{\pi(1+\lambda) }{2}$
as $n \to \infty$.
\end{lemma}

\begin{proof}
Since $\phi_n(k)\ge 0$ and the condition (\ref{6-26-14}) holds, 
Lemma \ref{11-16-2} guarantees that
$\min_{M\in {\cal M}_{\cov}(\U(1))} 
{\cal D}_R (|\phi_{n}\rangle \langle \phi_{n}|,M)
=
{\cal D}_R (|\phi_{n}\rangle)$.
Now, we choose $\lambda:= \frac{k}{n}$
and $\hat{g}:= n \hat{\theta}$.
Then, 
using the limiting function $\tilde{\phi}(\lambda):= \lim_{n \to \infty} \tilde{\phi}_n(\lambda)$,
we have
\begin{align*}
&\frac{1}{n^2}
\sum_{k=-n}^{n} k^2 |\phi_n(k)|^2
=
\sum_{k=-n}^{n}  (\frac{k}{n})^2 \frac{1}{\sqrt{2\pi}n} 
|\tilde{\phi}_n(\frac{k}{n})|^2 
\to 
\int_{-\infty}^{\infty} 
\lambda^2 |\tilde{\phi}(\lambda )|^2 \frac{d \lambda}{2\pi} 
\end{align*}
as $n \to \infty$.
Similarly,
since
\begin{align*}
& \frac{{\cal F}_{\pi}^{-1}(\phi_n)(\frac{\hat{g}}{n})}{(2\pi)^{\frac{1}{4}}n^{\frac{1}{2}} } 
=
\sum_{k=-n}^{n}
e^{-i k \frac{\hat{g}}{n}} \phi_E(k) 
\frac{1}{(2\pi)^{\frac{1}{4}} n^{\frac{1}{2}}} \\
=&
\sum_{k=-n}^{n}
e^{-i \frac{k}{n} \hat{g}} \tilde{\phi}_n(\frac{k}{n}) 
\frac{1}{\sqrt{2\pi}n} 
\to 
\int_{-\infty}^{\infty}
e^{-i \lambda \hat{g}} \tilde{\phi}(\lambda) 
\frac{d \lambda}{\sqrt{2\pi}} 
=
{\cal F}_{\pi}^{-1}[\tilde{\phi}](\hat{g}) ,
\end{align*}
we have
\begin{align*}
& n^2
\int_{-\pi}^{\pi}
(1-\cos (\hat{\theta}))
|{\cal F}_{\pi}^{-1}[{\phi}_n](\hat{\theta})|^2 \frac{d\hat{\theta}}{{2\pi}}  
\cong
n^2
\int_{-\pi}^{\pi}
\frac{\hat{\theta}^2}{2}
|{\cal F}_{\pi}^{-1}[{\phi}_n](\hat{\theta})|^2 \frac{d\hat{\theta}}{{2\pi}}  \\
= &
\int_{-\pi n}^{\pi n}
\frac{\hat{g}^2}{2}
|{\cal F}_{\pi}^{-1}[{\phi}_n](\frac{\hat{g}}{n})|^2 
\frac{d\hat{g}}{{2\pi} n} 
\to
\int_{-\infty}^{\infty}
\frac{\hat{g}^2}{2}
|{\cal F}_{\pi}^{-1}[\tilde{\phi}](\hat{g})|^2 
\frac{d\hat{g}}{\sqrt{2\pi} } .
\end{align*}
In Theorem \ref{T5-12-1b}, the minimum \eqref{5-12-10} with $L=1$
is attained by $\tilde{\phi}(\lambda)= (2\pi)^{1/4} \sin \frac{\pi(1+\lambda) }{2}$.
Hence, 
${\cal D}_R (|\phi_{n}\rangle) \cong \frac{1}{8 n^2} $ as $n \to \infty$
if and only if 
$\tilde{\phi}_n(\lambda)$ goes to $(2\pi)^{1/4} \sin \frac{\pi(1+\lambda) }{2}$
as $n \to \infty$.
\end{proof}

\subsection{Typical energy constraint}\Label{s10-2}
Next, we consider the risk function $R_{\U(1)}(\theta,\hat{\theta})= 1- \cos (\hat{\theta}-\theta)$
and the Hamiltonian $H= \sum_{k=-\infty}^{\infty} k^2 |k \rangle \langle k|$.
Then, thanks to Theorem \ref{T5-20}, 
the minimum error can be characterized by the following value.
\begin{align}
\kappa_{\U(1)}(E)
:=
\min_{\varphi \in L^2_{p,\even,n}((-\pi,\pi])} 
\{\langle \varphi| I-\cos (Q)| \varphi\rangle|
\langle \varphi| P^2 | \varphi\rangle \le E \} .
\Label{3-13-11}
\end{align}
For example, we can show that
\begin{align}
\kappa_{\U(1)}(0)=1 \Label{5-10-10}.
\end{align}
This fact can be also checked by the following way.
In fact, the condition $\langle \phi |H|\phi\rangle=0$
can be realized only when $\phi(0)=1$ and $\phi(n)=0$ with $n \neq 0$,
i.e., ${\cal F}_{\pi}^{-1}[\phi]=1$.
In this case, we have $\int_{-\pi}^{\pi}
(1-\cos (\hat{\theta}))
|{\cal F}_{\pi}^{-1}[\phi](-\hat{\theta})|^2 \frac{1}{{2\pi}} d\hat{\theta}
=1$.
Hence, we see (\ref{5-10-10}).

Now, we consider the case with non-zero $E$.
Since the condition of Lemma \ref{L4-25-2} hold, $\kappa_{\U(1)}(E)$ is convex.
Hence, we employ Lemma \ref{L5-10-1} to calculate $\kappa_{\U(1)}(E)$,
and consider 
the minimum 
\begin{align*}
\gamma_{\U(1)}(s):=&
\min_{\varphi\in L^2_n((-\pi,\pi ])}
\langle \varphi| (I-\cos (Q)) +  s P^2| \varphi \rangle \\
=& \min_{\varphi\in L^2_n((-\pi/2,\pi/2 ])}
\langle \varphi| (I-\cos (2Q))+ \frac{s P^2}{4}| \varphi \rangle .
\end{align*}
So, $\gamma_{\U(1)}(s)$ can be characterized as the minimum $\gamma_{\U(1)}$
having the solution in $L^2_n((-\pi/2,\pi/2 ])$ 
of the following differential equation.
\begin{align}
\frac{s}{4}\frac{d^2}{d\theta^2} \varphi(\theta) + 
( \gamma_{\U(1)}- 1 +\cos (2\theta))
\varphi(\theta) =0,
\end{align}
which is equivalent to
\begin{align}
\frac{d^2}{d\theta^2} \varphi(\theta)+ 
( \frac{4(\gamma_{\U(1)}- 1)}{s} + \frac{4}{s}\cos (2\theta))
\varphi(\theta) =0,
\end{align}
In order to find the minimum $\gamma_{\U(1)}$, we employ Mathieu equation (\ref{5-10-7}),
whose detail is summarized in Subsection \ref{asB}.
Hence, using the function $a_0$ given in Subsection \ref{asB},
we have 
$\gamma_{\U(1)}(s)
=\frac{s a_0(-\frac{2}{s} )}{4} +1
=\frac{s a_0(\frac{2}{s})}{4} +1$.
So, applying (\ref{5-10-2}) to $\kappa_{\U(1)}(E)$, and
combining 
the facts given in Subsection \ref{asB},
we obtain the following theorem.

\begin{theorem}\Label{T3-13-3}
\begin{align}
\kappa_{\U(1)}(E)
=\max_{s>0} \frac{s a_0(\frac{2}{s})}{4} +1 -sE
\Label{5-10-13}.
\end{align}
The minimum \eqref{3-13-11} is attained by the input state $|\phi\rangle$
with the measurement ${\cal M}_{|I\rangle \langle I|}$
if and only if 
${\cal F}_{\pi}^{-1}[\phi](\theta)=\ce_0(\frac{\theta}{2}, -\frac{2}{s_E})$,
where $s_E$ is 
$\argmax_{s>0} \frac{s a_0(\frac{2}{s})}{4} +1 -sE$
and the function $\ce_0$ is given in Subsection \ref{asB}.
\end{theorem}

Using the formula \eqref{5-10-13},
we can calculate $\kappa_{\U(1)}(E)$ as Fig. \ref{g1}.

\begin{figure}[htbp]
\begin{center}
\scalebox{0.6}{\includegraphics[scale=1.2]{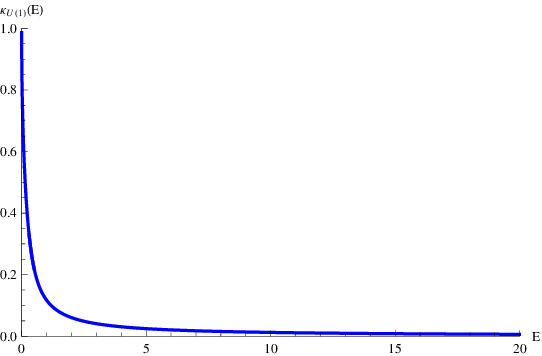}}
\end{center}
\caption{Graph of $\kappa_{\U(1)}(E)$.}
\Label{g1}
\end{figure}%

By using the expansion \eqref{exp1} for $a_0$, as $s \to 0$, $\gamma_{\U(1)}(s)$ can be expanded to
\begin{align*}
\gamma_{\U(1)}(s)
\cong
\frac{s(-2\frac{2}{s}+2 \sqrt{\frac{2}{s}}-\frac{1}{4})}{4} +1
=\sqrt{\frac{s}{2}}- \frac{s}{16}.
\end{align*}
As is shown in Lemma \ref{L5-10-1},
$s_E$ is decreasing as a function of $E$.
Hence, when $E$ is large, 
solving the equation $\gamma_{\U(1)}'(s_E)=E$, we approximately obtain
$s_E\cong \frac{1}{8(E+1/16)^2}$.
Hence,
\begin{align}
&\kappa_{\U(1)}(E)= \gamma_{\U(1)}(s_E)-s_E E
\cong \sqrt{\frac{s_E}{2}}- \frac{s_E}{16} -s_E E
= \sqrt{\frac{s_E}{2}}- s_E(E+\frac{1}{16}) 
\nonumber \\
\cong & \frac{1}{8 (E+1/16)}
\cong \frac{1}{8E}- \frac{1}{128 E^2}.
\Label{5-10-4}
\end{align}
As is shown in Fig. \ref{g11},
while 
the first order approximation $\kappa_{1,\U(1),\infty}(E):=\frac{1}{8E}$ gives a good approximation 
for $\kappa_{\U(1)}(E)$ with a large $E$,
the second order approximation 
$\kappa_{2,\U(1),\infty}(E):=\frac{1}{8E}- \frac{1}{128 E^2}$
much improves the approximation 
for $\kappa_{\U(1)}(E)$ with a large $E$.
Hence, we have the following asymptotic characterization.
\begin{align}
&
\lim_{E \to \infty}
E
\min_{\rho \in {\cal S}(L^2(\bZ))} 
\min_{M \in {\cal M}_{\cov}(U(1))} 
\{{\cal D}_R(\rho,M)| \Tr \rho H \le E \}
\nonumber \\
=&
\lim_{E \to \infty}
E
\min_{ \{p_i\} }
\min_{\rho_i \in {\cal S}(L^2(\bZ))} 
\min_{M_i \in {\cal M}_{\cov}(U(1))} 
\{\sum_i p_i  {\cal D}_R(\rho_i,M_i)| \sum_i p_i  \Tr \rho_i H \le E \}
\nonumber \\
=&
\lim_{E \to \infty}
E
\min_{|{\phi} \rangle \in L^2_n(\bZ) }
\{{\cal D}_{R}(|\phi \rangle)|
\langle \phi | H |\phi \rangle \le E \} 
=
\frac{1}{8} 
\Label{3-13-7}.
\end{align}

\begin{figure}[htbp]
\begin{center}
\scalebox{0.6}{\includegraphics[scale=1.2]{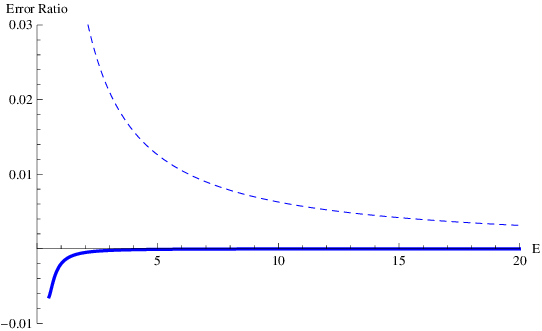}}
\end{center}
\caption{Comparison of two approximations 
$\kappa_{1,\U(1),\infty}$ and
$\kappa_{2,\U(1),\infty}$
of $\kappa_{\U(1)}$ with a large $E$.
Thick line expresses 
the error ratio
$\frac{\kappa_{2,\U(1),\infty}(E)-\kappa_{\U(1)}(E)}{\kappa_{\U(1)}(E)}$,
and
dashed line expresses 
the error ratio
$\frac{\kappa_{1,\U(1),\infty}(E)-\kappa_{\U(1)}(E)}{\kappa_{\U(1)}(E)}$.}
\Label{g11}
\end{figure}%

Next, we consider the case when $E$ is small.
Hence, when $s$ is large,
by using the expansion \eqref{exp2} for $a_0$,
$\gamma_{\U(1)}(s)$ can be expanded to
\begin{align*}
\gamma_{\U(1)}(s)
\cong
\frac{s(-\frac{1}{2}(\frac{2}{s})^2 +\frac{7}{128}(\frac{2}{s})^4)}{4} +1
= 1 -\frac{1}{2s}+\frac{7}{32 s^3}.
\end{align*}
When $E$ is small, 
solving the equation $\gamma_{\U(1)}'(s_E)=E$, we approximately obtain
$s_E\cong \sqrt{\frac{1}{2E}+\frac{21}{128}}\cong \frac{1}{\sqrt{2E}}(1+\frac{21}{32}E)$.
Hence,
\begin{align}
&\kappa_{\U(1)}(E)= \gamma_{\U(1)}(s_E)-s_E E
\cong 1 -\frac{1}{2s_E} -s_E E +\frac{7}{32 s_E^3} \nonumber \\
\cong &1 
- \frac{\sqrt{E}}{\sqrt{2}}(1-\frac{21}{32}E)
- \frac{\sqrt{E}}{\sqrt{2}}(1+\frac{21}{32}E)
+\frac{7}{32} \sqrt{2E}^3 
\cong 
1- \sqrt{2E} +\frac{7 \sqrt{2}}{16}E^{\frac{3}{2}}.
\Label{5-10-4c}
\end{align}
This expansion with $E=0$ coincides with \eqref{5-10-10}. 
As is shown in Fig. \ref{g12},
while 
the first order approximation $\kappa_{1,\U(1),+0}(E):=1- \sqrt{2E}$ gives a good approximation 
for $\kappa_{\U(1)}(E)$ with a small $E$,
the second order approximation 
$\kappa_{2,\U(1),+0}(E):=1- \sqrt{2E} +\frac{7 \sqrt{2}}{16}E^{\frac{3}{2}}$
much improves the approximation 
for $\kappa_{\U(1)}(E)$ with a small $E$.

\begin{figure}[htbp]
\begin{center}
\scalebox{0.6}{\includegraphics[scale=1.2]{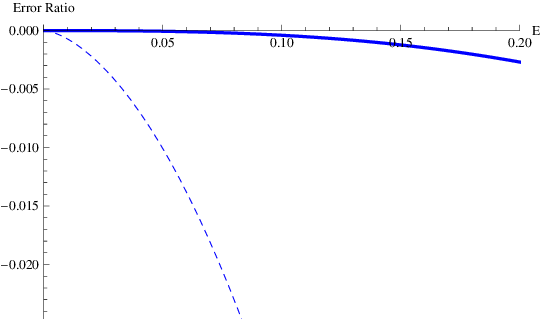}}
\end{center}
\caption{Comparison of two approximations 
$\kappa_{1,\U(1),+0}$ and $\kappa_{2,\U(1),+0}$
of $\kappa_{\U(1)}$ with a small $E$.
Thick line expresses 
the error ratio
$\frac{\kappa_{2,\U(1),+0}(E)-\kappa_{\U(1)}(E)}{\kappa_{\U(1)}(E)}$,
and
dashed line expresses 
the error ratio
$\frac{\kappa_{1,\U(1),+0}(E)-\kappa_{\U(1)}(E)}{\kappa_{\U(1)}(E)}$.}
\Label{g12}
\end{figure}%

For the asymptotic optimality condition with respect to input states,
we obtain the following lemma.
\begin{lemma}\Label{T3-13-3x}
For a sequence $\{E_l\}$ satisfying $E_l \to \infty$ as $l \to \infty$,
we focus on a sequence of input states 
$\{\phi_{E_l}\}$ satisfying that $\phi_{E_l}(n)\ge 0$
and $\langle \phi_{E_l}| H |\phi_{E_l}\rangle \le E_l$. 
Then,
$\min_{M\in {\cal M}_{\cov}(\U(1))} 
{\cal D}_R (|\phi_{E_l}\rangle \langle \phi_{E_l}|,M)
=
{\cal D}_R (|\phi_{E_l}\rangle)
\cong \frac{1}{8 E_l} $ as $l \to \infty$
if and only if the sequence of functions
$\tilde{\phi}_l(\lambda):= (2\pi E_l)^{1/4}\phi_{E_l}(\lfloor \sqrt{E_l} \lambda  +\frac{1}{2}\rfloor) 
\in L^2(\real)$
satisfies that
$\tilde{\phi}_l(\lambda)$ goes to $e^{-\frac{\lambda^2}{4}}$
as $l \to \infty$.
For example,
when we $E_l=\frac{l}{2}$, 
the following input state $|\phi_{b,l}\rangle$ 
asymptotically attains the minimum \eqref{3-13-7}.
\begin{align}
\phi_{b,l}(k):= 
\left\{
\begin{array}{ll}
\frac{1}{2^{l}} \sqrt{{2l \choose n+l}}
& \hbox{if }|k| \le l \\
0
& \hbox{if }|k| >  l.
\end{array}
\right.
\end{align}
\end{lemma}

\begin{proof}
The relation 
$\min_{M\in {\cal M}_{\cov}(G)} 
{\cal D}_R (|\phi_{E_l}\rangle \langle \phi_{E_l}|,M)
=
{\cal D}_R (|\phi_{E_l}\rangle)$ holds by the same reason as Lemma \ref{T3-13-3c}.
Now, we choose $\lambda:= \frac{k}{\sqrt{E} }$
and $\hat{g}:= \sqrt{E} \hat{\theta}$.
Then, 
we have
\begin{align*}
&\frac{1}{E_l}
\sum_{k=-\infty}^{\infty}k^2 |\phi_{E_l}(k)|^2
=
\sum_{k=-\infty}^{\infty}  (\frac{k}{\sqrt{E_l}})^2 \frac{1}{\sqrt{2\pi E_l}} 
|\tilde{\phi}_{E_l}(\frac{k}{\sqrt{E_l}})|^2 
\to 
\int_{-\infty}^{\infty} 
\lambda^2 |\tilde{\phi}(\lambda )|^2 \frac{d \lambda}{2\pi} 
\end{align*}
as $E_l \to \infty$.
Similarly,
since
\begin{align*}
& \frac{{\cal F}_{\pi}^{-1}(\phi_{E_l})(\frac{\hat{g}}{\sqrt{E}})}{(2\pi E)^{\frac{1}{4}}} 
=
\sum_{k=-\infty}^{\infty}
e^{-i k \frac{\hat{g}}{\sqrt{E}}} \phi_{E_l}(k) 
\frac{1}{(2\pi E_l)^{\frac{1}{4}}} \\
=&
\sum_{k=-\infty}^{\infty} 
e^{-i \frac{k}{\sqrt{E_l}} \hat{g}} \tilde{\phi}_{E_l}(\frac{k}{\sqrt{E}}) 
\frac{1}{\sqrt{2\pi E}} 
\to 
\int_{-\infty}^{\infty}
e^{-i \lambda \hat{g}} \tilde{\phi}(\lambda) 
\frac{d \lambda}{\sqrt{2\pi}} 
=
{\cal F}^{-1}[\tilde{\phi}](\hat{g}) ,
\end{align*}
we have
\begin{align*}
& E_l
\int_{-\pi}^{\pi}
(1-\cos (\hat{\theta}))
|{\cal F}_{\pi}^{-1}[{\phi}_{E_l}](\hat{\theta})|^2 \frac{d\hat{\theta}}{{2\pi}}  
\cong
 E_l
\int_{-\pi}^{\pi}
\frac{\hat{\theta}^2}{2}
|{\cal F}_{\pi}^{-1}[{\phi}_{E_l}](\hat{\theta})|^2 \frac{d\hat{\theta}}{{2\pi}}  \\
= &
\int_{-\pi \sqrt{E_l}}^{\pi \sqrt{E_l}}
\frac{\hat{g}^2}{2}
|{\cal F}_{\pi}^{-1}[{\phi}_{E_l}](\frac{\hat{g}}{\sqrt{E_l}})|^2 
\frac{d\hat{g}}{{2\pi} \sqrt{E_l}} 
\to
\int_{-\infty}^{\infty}
\frac{\hat{g}^2}{2}
|{\cal F}^{-1}[\tilde{\phi}](\hat{g})|^2 
\frac{d\hat{g}}{\sqrt{2\pi} } .
\end{align*}
In Theorem \ref{T5-12-1}, the minimum \eqref{5-13-14} with $E=1$
is attained only by 
$\tilde{\phi}(\lambda)=e^{-\frac{\lambda^2}{4}}$
Hence, 
${\cal D}_R (|\phi_{E_l}\rangle) \cong \frac{1}{8 E_l} $ as $l \to \infty$
if and only if 
$\tilde{\phi}_l(\lambda)$ goes to $e^{-\frac{\lambda^2}{4}}$
as $l \to \infty$.

Since
\begin{align}
\langle \phi_{b,l}| H |\phi_{b,l}\rangle
=
\sum_{k=-l}^l
k^2 \frac{1}{2^{2l}}{2l \choose k+l}
= \frac{2l}{4}=\frac{l}{2},
\end{align}
the state $|\phi_{b,l}\rangle$ has the energy $\frac{l}{2}$.
Thanks to the central limit theorem,
$(\pi l)^{\frac{1}{2}} |\phi_{b,l}(\sqrt{\frac{l}{2}}\lambda)|^2
$ goes to 
$e^{-\frac{\lambda^2}{2}}$, i.e.,
$(\pi l)^{\frac{1}{4}} \phi_{b,l}(\sqrt{\frac{l}{2}}\lambda)$ goes to 
$e^{-\frac{\lambda^2}{4}}$.
Hence, the input state $|\phi_{b,l}\rangle$
also asymptotically attains the minimum \eqref{3-13-7}.
\end{proof}

In this problem, the global phase factor does not effect the representation, but 
changes the energy slightly.
By using a $\lambda_0 \in \bR$,
the representation is changed to
\begin{align}
f_{(\lambda_0)}(\theta):=
\sum_{k=-\infty}^{\infty}
e^{i (k+\lambda_0) \theta}
|k \rangle \langle k|.
\end{align}
Then, the Hamiltonian is given as
\begin{align}
H_{(\lambda_0)}:= N_{(\lambda_0)}^2 ,\quad
N_{(\lambda_0)}:=\sum_{k=-\infty}^{\infty}
(k+\lambda_0 )|k \rangle \langle k|.
\end{align}
Even in this modification,
the result in Theorem \ref{T3-13-3} is not changed
because this modification does not effect the asymptotic behavior of the energy.

\subsection{Practical construction of asymptotically optimal estimator with energy constraint}\Label{s10-3}
While \eqref{3-13-7} provides an asymptotically optimal estimator with energy
constraint, its construction is not so practical. However, the optimal performance
with energy constraint can be realized with easier construction by the
following ways. 
Now, we fix a state $|\phi \rangle = \sum_{k=-\infty}^{\infty}
\sqrt{p_k}|k\rangle 
\in \ell^2(\bZ)$, and choose the real number 
$\lambda:= -\sum_{k=-\infty}^{\infty}k p_k$.
The energy of $|\phi \rangle$ is 
$E_{\phi}:=
\langle \phi | H_{(\lambda)}| \phi \rangle=
\sum_{k=-\infty}^{\infty} (k+\lambda)^2 p_k$ under the Hamiltonian $H_{(\lambda)}$.

We also consider the $m$-tensor product system ${\cal H}_m:= 
\ell^2( \bZ)^{\otimes m}$, 
the Hamiltonian $H_{m}:= (\sum_{i=1}^m  I^{\otimes (i-1)} \otimes N_{(\lambda)} 
\otimes I^{\otimes (m-i)})^2$, and
the tensor product representation $f_{(\lambda)}(\theta)^{\otimes m}$.
Then, 
since
$\langle \phi^{\otimes n} | 
I^{\otimes (i-1)} \otimes N_{(\lambda)} \otimes I^{\otimes (m-i)}
| \phi^{\otimes n} \rangle =0$,
the energy of $|\phi^{\otimes m}\rangle$ is
\begin{align*}
& \langle \phi^{\otimes n} | H_{m}| \phi^{\otimes n} \rangle 
=
\langle \phi^{\otimes n} | 
(\sum_{i=1}^m  I^{\otimes (i-1)} \otimes N_{(\lambda)} 
\otimes I^{\otimes (m-i)})^2
| \phi^{\otimes n} \rangle \\
=&
\langle \phi^{\otimes n} | 
\sum_{i=1}^m  I^{\otimes (i-1)} \otimes N_{(\lambda)}^2 
\otimes I^{\otimes (m-i)}
| \phi^{\otimes n} \rangle 
= m E_{\phi}.
\end{align*}

Now, we give the following estimation protocol (Protocol 1).
\begin{description}
\item[(1.1)]
We set the initial state $|\phi^{\otimes m}\rangle$.
\item[(1.2)]
We apply the covariant measurement $M_{|{\cal I}\rangle \langle {\cal I}|}$
on each system $\ell^2(\bZ)$.
Then, we obtain $n$ outcomes ${\theta}_1, \ldots, {\theta}_n$.
Each outcome $\theta_i$ obeys the distribution
$p_{\theta}(\theta_i)d \theta_i :=
|\sum_{k=-\infty}^{\infty} \sqrt{p_k}e^{-ik (\theta_i-\theta)} |^2
\frac{d \theta_i}{2\pi} $
when the true parameter is $\theta$.
\item[(1.3)]
We apply the maximum likelihood estimator to the obtained outcomes
${\theta}_1, \ldots, {\theta}_m$.
Then, we obtain the final estimate $\hat{\theta}_m$.
That is, we decide $\hat{\theta}_m$ as
\begin{align}
\hat{\theta}_m:= \argmax_{\theta \in (-\pi,\pi]} \sum_{i=1}^m 
\log p_{\theta}(\theta_i).
\end{align}
\end{description}

We denote the above measurement with the output $\hat{\theta}_m$ by $\tilde{M}_m$.
Then, due to the following theorem,
the above protocol asymptotically realizes 
the minimum error under the energy constraint. 
The optimal performance with energy constraint 
can be attained without use of quantum correlation in the measurement process.
\begin{theorem}
Assume that $\phi$ satisfies one of two conditions.
\begin{description}
\item[(a)]
$\lambda=0$ and 
${\cal F}[\phi]$ is an even function.
\item[(b)]
$\lambda=-\frac{1}{2}$ and 
$\sum_{k=-\infty}^{\infty}
\sqrt{p_k} e^{i(k+\lambda)\theta }$
is an even function. 
\end{description}

The relation
\begin{align}
\lim_{m \to \infty} 
m {\cal D}_R(|\phi^{\otimes m} \rangle, \tilde{M}_m)
&= \frac{1}{8E_\phi} \Label{5-1-6b}.
\end{align}
holds. That is,
\begin{align}
\lim_{m \to \infty} 
\langle\phi^{\otimes m}| H_m |\phi^{\otimes m}\rangle
{\cal D}_R(|\phi^{\otimes m} \rangle, \tilde{M}_n)
= \frac{1}{8}.
\end{align}
\end{theorem}

\begin{proof}
We show (\ref{5-1-5}).
For this purpose, we calculate the Fisher information
of the distribution family $\{ p_{\theta}(\theta') \}$.
Due to the assumption, we have
\begin{align}
\sqrt{p_{\theta}(\theta')}=
\sum_{k=-\infty}^{\infty} \sqrt{p_k}e^{-i(k+\lambda) (\theta_i-\theta)} .
\end{align}
and $\sqrt{p_k}=\sqrt{p_{k'}}$ when $k+\lambda = -(k'+\lambda)$.
The logarithmic derivative is given as
\begin{align*}
& l_{\theta}(\theta_i):=
\frac{d}{d \theta} \log p_{\theta}(\theta_i)
=
2\frac{d}{d \theta} \log \sqrt{p_{\theta}(\theta_i)} 
= 
2 \frac{\sum_{k=-\infty}^{\infty} 
i(k+\lambda) \sqrt{p_k}e^{-i(k+\lambda) (\theta_i-\theta)} }
{\sqrt{p_{\theta}(\theta_i)}}.
\end{align*}
Since 
\begin{align}
\int_{-\pi}^{\pi}
e^{-i(k+\lambda) (\theta_i-\theta)} 
e^{-i(k'+\lambda) (\theta_i-\theta)} 
\frac{d  \theta_i}{2\pi}=
\left\{
\begin{array}{cl}
0  & \hbox{if }k+\lambda \neq -(k'+\lambda) \\
1   & \hbox{if }k+\lambda = -(k'+\lambda) ,
\end{array}
\right.
\end{align}
the Fisher information is calculated to
\begin{align}
J_{\theta}
:=& \int_{-\pi}^{\pi}
l_{\theta}(\theta_i)^2 
p_{\theta}(\theta_i) 
\frac{d \theta_i}{2\pi} 
=\int_{-\pi}^{\pi}
4
(
\sum_{k=-\infty}^{\infty} 
i(k+\lambda) \sqrt{p_k}e^{-i(k+\lambda) (\theta_i-\theta)} )^2
\frac{d \theta_i}{2\pi} \nonumber \\
=&
4
\sum_{k=-\infty}^{\infty} 
(k+\lambda)^2 p_k
=4 E_{\phi}.
\Label{5-2-1bx}
\end{align}

Remember that
the asymptotic mean square error of the maximum likelihood estimator
can be characterized by 
the inverse of Fisher information $J_{\theta}$ \cite{van}.
That is, we obtain  
$n \rE_\theta (\hat{\theta}_n-\theta)^2 \to  J_{\theta}^{-1}$,
where $\rE_\theta$ expresses the expectation under the distribution $p_{\theta} $.
Hence, we have
\begin{align*}
{\cal D}_R(|\phi^{\otimes m} \rangle, \tilde{M}_n)
= 
\rE_\theta[ 1 -\cos (\hat{\theta}_{n}-\theta) ]
\cong
\frac{1}{2} \rE_\theta[ (\hat{\theta}_{n}-\theta)^2]
\cong 
\frac{1}{2} \frac{1}{4 E_{\phi} n}=
\frac{1}{8 E_{\phi} n},
\end{align*}
where $\rE_\theta$ expresses the expectation
under the distribution $p_{\theta}(\theta')$.
\end{proof}

\subsection{Uncertainty relation}\Label{s10-4}
We consider the relation between the above results and the 
uncertainty relation on the space $L^2_p((-\pi,\pi])$.
In this space, we can consider the pair of operators
$(\cos Q, \sin Q)$.
Then, we focus on the uncertainty 
\begin{align*}
&\Delta_{\varphi}^2(\cos Q, \sin Q):=
\Delta_{\varphi}^2 \cos Q +\Delta_{\varphi}^2 \sin Q \\
=&
\langle \varphi | \cos^2 Q |\varphi \rangle 
+\langle \varphi | \sin^2 Q |\varphi \rangle 
-\langle \varphi | \cos Q |\varphi \rangle^2
-\langle \varphi | \sin Q |\varphi \rangle^2 \\
=&
1-\langle \varphi | \cos Q |\varphi \rangle^2
-\langle \varphi | \sin Q |\varphi \rangle^2 .
\end{align*}
Then, as the uncertainty relation, 
we consider trade-off between
$\Delta_{\varphi}^2(\cos Q, \sin Q)$
and $\Delta_{\varphi}^2 P$,
which is formulated as the following theorem.

\begin{theorem}\Label{T5-12-16}
The minimum of the uncertainty $\Delta_{\varphi}^2(\cos Q, \sin Q)$
under the constraint for $\Delta_{\varphi}^2 P $
is calculated as
\begin{align}
\min_{\varphi \in L^2_{p,n}([-\pi,\pi])}
\{
\Delta_{\varphi}^2(\cos Q, \sin Q)|
\Delta_{\varphi}^2 P \le E \}
=
\max_{s>0} 1-(sE-\frac{s a_0(\frac{2}{s})}{4})^2 .\Label{5-12-16}
\end{align}
The minimum in \eqref{5-12-16}
is realized by $\varphi$
if and only if
$\varphi$ is given as a shift of the Mathieu function
$\ce_0(\frac{\theta}{2}, -\frac{2}{s_E})$.
Further, 
$\min_{s>0} sE-\frac{s a_0(\frac{2}{s})}{4})$ is a positive value,
\eqref{5-12-16} has the asymptotic expansion 
$\frac{1}{4E}-\frac{1}{32 E^2}$ as $E \to \infty$.
\end{theorem}

\begin{proof}
Due to the symmetry, without loss of generality,
we can assume that 
$\langle \varphi | \sin Q |\varphi \rangle=0$
and 
$\langle \varphi | P  |\varphi \rangle=0$.
Hence, we obtain
\begin{align*}
&\min_{\varphi \in L^2_{p,n}([-\pi,\pi])}
\{
\Delta_{\varphi}^2(\cos Q, \sin Q)|
\Delta_{\varphi}^2 P \le E \} \\
=&
\min_{\varphi \in L^2_{p,n}([-\pi,\pi])}
\{
1-\langle \varphi| \cos Q| \varphi\rangle^2 |
\langle \varphi| P^2 | \varphi\rangle \le E \} \\
=& 1-
(\max_{\varphi \in L^2_{p,n}([-\pi,\pi])}
\{
\langle \varphi| \cos Q| \varphi\rangle |
\langle \varphi| P^2 | \varphi\rangle \le E \})^2.
\end{align*}
Thanks to Theorem \ref{T3-13-3}, 
we have
\begin{align}
&1-
(\max_{\varphi \in L^2_{p,n}([-\pi,\pi])}
\{
\langle \varphi| \cos Q| \varphi\rangle |
\langle \varphi| P^2 | \varphi\rangle \le E \})^2 \nonumber \\
=&
1-(1-\max_{s>0} \frac{s a_0(\frac{2}{s})}{4} +1 -sE)^2 
=
1-(\max_{s>0} \frac{s a_0(\frac{2}{s})}{4} -sE)^2 .\nonumber
\end{align}
Since $\max_{s>0} \frac{s a_0(\frac{2}{s})}{4} +1 -sE\le 1$,
$sE-\frac{s a_0(\frac{2}{s})}{4} \ge 0$.
Hence,
\begin{align*}
&1-(\max_{s>0} \frac{s a_0(\frac{2}{s})}{4} -sE)^2
=1-(\min_{s>0} sE- \frac{s a_0(\frac{2}{s})}{4} )^2 \\
=&1-\min_{s>0} (sE- \frac{s a_0(\frac{2}{s})}{4} )^2
=\max_{s>0} 1-(sE- \frac{s a_0(\frac{2}{s})}{4} )^2,
\end{align*}
which implies \eqref{5-12-16}.
Then, the minimum uncertainty in \eqref{5-12-16}
is realized by $\varphi$
if and only if
$\varphi$ is given as a shift of the Mathieu function
$\ce_0(\frac{\theta}{2}, -\frac{2}{s_E})$.

Further,
when $E$ approaches infinity,
$\max_{s>0} 1-(sE- \frac{s a_0(\frac{2}{s})}{4} )^2
\cong 1- (1-(\frac{1}{8E}-\frac{1}{128 E^2}))^2
\cong \frac{1}{4E}-\frac{1}{32 E^2}$.
\end{proof}

Next, as another type of uncertainty relation,
we consider the trade-off between
$\Delta_{\varphi}^2(\cos Q, \sin Q)$
and $\Delta_{\varphi,\max} P$,
which is defined as
the maximum eigenvalue of 
$|P - \langle \varphi|P |\varphi\rangle|$
such that 
the corresponding projection $A$ satisfies $\langle \varphi| A| \varphi\rangle>0$.

\begin{theorem}
The minimum of the uncertainty $\Delta_{\varphi}^2(\cos Q, \sin Q)$
under the constraint for $\Delta_{\varphi,\max} P $
is calculated as
\begin{align}
\min_{\varphi \in L^2_{p,n}([-\pi,\pi])}
\{
\Delta_{\varphi}^2(\cos Q, \sin Q)|
\Delta_{\varphi,\max} P \le E \} =
\sin^2 \frac{\pi}{2\lfloor E \rfloor+2}
\Label{5-12-17}.
\end{align}
The minimum uncertainty is realized by $\varphi$
if and only if
${\cal F}_{\pi}[\varphi](\lambda)$ is given as a shift of the wave function
$C \sin\frac{\pi(\lambda+\lfloor E \rfloor+1)}{2\lfloor E \rfloor+2}$, 
where $C$ is the normalizing constant.
\end{theorem}
\begin{proof}
Due to the symmetry, without loss of generality,
we can assume that 
$\langle \varphi | \sin Q |\varphi \rangle=0$
and 
$\langle \varphi | P  |\varphi \rangle=0$.
Hence, we obtain
\begin{align*}
&\min_{\varphi \in L^2_{p,n}([-\pi,\pi])}
\{
\Delta_{\varphi}^2(\cos Q, \sin Q)|
\Delta_{\varphi,\max} P \le E \} \\
=&
\min_{\varphi \in L^2_{p,n}([-\pi,\pi])}
\{
1-\langle \varphi| \cos Q| \varphi\rangle^2 |
{\cal F}_{\pi}[\varphi]\in {\cal K}_{\Lambda_{\lfloor E \rfloor}}
\le E \} \\
=& 1-
(\max_{\varphi \in L^2_{p,n}([-\pi,\pi])}
\{
\langle \varphi| \cos Q| \varphi\rangle |
{\cal F}_{\pi}[\varphi]\in {\cal K}_{\Lambda_{\lfloor E \rfloor}}
\})^2.
\end{align*}
Thanks to Theorem \ref{T5-12-4}, 
we have
\begin{align}
&1-
(\max_{\varphi \in L^2_{p,n}([-\pi,\pi])}
\{
\langle \varphi| \cos Q| \varphi\rangle |
\langle \varphi| P^2 | \varphi\rangle \le E 
{\cal F}_{\pi}[\varphi]\in {\cal K}_{\Lambda_{\lfloor E \rfloor}}
\})^2 \nonumber \\
=&
1-(1-(1-\cos \frac{\pi}{2\lfloor E \rfloor+2} )
)^2 
=
1-\cos^2 \frac{\pi}{2\lfloor E \rfloor+2}
=
\sin^2 \frac{\pi}{2\lfloor E \rfloor+2}.
\end{align}
Then, the minimum uncertainty in \eqref{5-12-17}
is realized by $\varphi$
if and only if
${\cal F}_{\pi}[\varphi](\lambda)$ is given as a shift of the wave function
$C \sin\frac{\pi(\lambda+\lfloor E \rfloor+1)}{2\lfloor E \rfloor+2}$, 
where $C$ is the normalizing constant.
\end{proof}

\section{2-dimensional special unitary group $\SU(2)$}\Label{s11}
\subsection{General structure of $\SU(2)$ estimation}\Label{s11-1}
We consider the special linear group $\SU(2)$, which is simply connected.
In the case of $G=\SU(2)$, 
we identify the elements of the set $\hat{G}=\hat{\SU(2)}$ by the highest weight.
Now, we consider two kinds of parametrization of $\SU(2)$.
Using the matrices
\begin{align}
\sigma_1 &:=
\left(
\begin{array}{cc}
0 & 1 \\
1 & 0
\end{array}
\right), \quad
\sigma_2:=
\left(
\begin{array}{cc}
0 & -i \\
i & 0
\end{array}
\right) , \quad
\sigma_3:=
\left(
\begin{array}{cc}
1 & 0 \\
0 & -1
\end{array}
\right) ,
\end{align}
we have the first kind of parametrization 
\begin{align}
g_{\vec{\theta}}:= \exp(i\sum_{k=1}^3 \frac{\theta^k}{2} \sigma_k) \Label{5-19-5}
\end{align}
with the range $\{\vec{\theta} =(\theta^1,\theta^2,\theta^3) |
\|\vec{\theta} \|:= \sqrt{\sum_{k=1}^3 (\theta^k)^2} \le 2\pi\} $.
We also the second kind of parametrization
\begin{align}
\tilde{g}_{\theta,\eta_1,\eta_2}:
= 
\left(
\begin{array}{cc}
x_0+i x_1 & x_2+i x_3 \\
x_2-i x_3 & x_0-i x_1
\end{array}
\right) ,
\Label{5-19-5c}
\end{align}
where 
$x_0=\cos \frac{\theta}{2}$,
$x_1=\sin \frac{\theta}{2} \cos \eta_1$,
$x_2=\sin \frac{\theta}{2} \sin \eta_1 \cos \eta_2$,
$x_3=\sin \frac{\theta}{2} \sin \eta_1 \sin \eta_2$
with the range $\theta \in (-2\pi,2\pi]$,
$\eta_1 \in (-\frac{\pi}{2},\frac{\pi}{2}]$,
$\eta_2 \in (-\frac{\pi}{2},\frac{\pi}{2}]$.
Under these parametrization, 
the character $\chi_{\frac{k}{2}}$ can be written as
\begin{align}
\chi_{\frac{k}{2}}(\tilde{g}_{\theta,\eta_1,\eta_2})
&= \sum_{l=0 \hbox{ \rm or } \frac{1}{2}}^{\frac{k}{2}} 
\cos l \theta
=\frac{\sin \frac{k+1}{2}\theta}{\sin \frac{\theta}{2}}
\Label{5-28-11} \\
\chi_{\frac{k}{2}}({g}_{\vec{\theta}})
&= \sum_{l=0 \hbox{ \rm or } \frac{1}{2}}^{\frac{k}{2}} 
\cos l \|\vec{\theta}\|.
\Label{5-28-12} 
\end{align}
The second equation of \eqref{5-28-11} can be shown as follows.
\begin{proofof}{the second equation of \eqref{5-28-11}}
Then, 
for an even $k$,
we have 
$\chi_{\frac{k}{2}}(\tilde{g}_{\theta,\eta_1,\eta_2})
=
1+\sum_{l=1}^{\frac{k}{2}} 2 \cos l \theta$, which implies 
\begin{align}
&\chi_{\frac{k}{2}}(\tilde{g}_{\theta,\eta_1,\eta_2})
\sin \frac{\theta}{2} 
=
(1+\sum_{l=1}^{\frac{k}{2}} 2 \cos l \theta)\sin \frac{\theta}{2} 
\nonumber \\
=&
\sin \frac{\theta}{2} 
+\sum_{l=1}^{\frac{k}{2}}
(\sin \frac{2l+1}{2} \theta -\sin \frac{2l-1}{2} \theta) 
=\sin \frac{k+1}{2} \theta \Label{5-20-9}.
\end{align}
For an odd $k$,
we have
$\chi_{\frac{k}{2}}(\tilde{g}_{\theta,\eta_1,\eta_2})
=
\sum_{l=0}^{\frac{k-1}{2}} \cos (l+\frac{1}{2}) \theta$,
which implies 
\begin{align}
& \chi_{\frac{k}{2}}(\tilde{g}_{\theta,\eta_1,\eta_2})
\sin \frac{\theta}{2} 
=
(\sum_{l=0}^{\frac{k-1}{2}} \cos (l+\frac{1}{2}) \theta)
\sin \frac{\theta}{2} 
\nonumber \\
=&
\sum_{l=0}^{\frac{k-1}{2}}
(\sin (l+1) \theta -\sin l \theta) 
=\sin \frac{k+1}{2} \theta \Label{5-20-10}.
\end{align}
\end{proofof}

When the risk function $R$ satisfies the condition \eqref{6-26-14}, 
the risk function is written as 
\begin{align}
R(e,\hat{g})= \alpha_0- \sum_{k=1}^{\infty} \alpha_{\frac{k}{2}} \chi_{\frac{k}{2}}(\hat{g})
\Label{5-19-1}
\end{align}
with $\alpha_{\frac{k}{2}} \ge 0$.
Defining the even periodic function 
$w(\theta):=(\alpha_0- \sum_{k=1}^{\infty}\alpha_{\frac{k}{2}} 
(\sum_{l=0 \hbox{ \rm or } \frac{1}{2}}^{\frac{k}{2}} 
\cos l \theta))$ with the period $4\pi$,
we have $R(e,\tilde{g}_{\theta,\eta_1,\eta_2})=w(\theta)$.
As a typical risk function, we often adopt the risk function 
$R_{\SU(2)}(e,\hat{g})= 1- \frac{1}{2}\chi_{\frac{1}{2}}(\hat{g})$,
which is written as
$R_{\SU(2)}(e,g_{\vec{\theta}})= 1- \cos \frac{\|\vec{\theta} \|}{2}$
and
$R_{\SU(2)}(e,\tilde{g}_{\theta,\eta_1,\eta_2})= 1- \cos \frac{\theta}{2}$.

We often use the risk function 
$R(e,g_{\vec{\theta}})=
\frac{3-\chi_1(g_{\vec{\theta}})}{2}=1- \cos \|\vec{\theta} \|$.
However, we cannot distinguish matrices $e$ and $-e$
under this risk function because $\frac{3-\chi_1(-e)}{2}=0$.
That is, under the projection $\varpi: \SU(2) \to \SO(3)$,
the two elements in $\varpi^{-1}(g)$ cannot be distinguished.
So, it is better to use this function as a risk function of estimation of $\SO(3)$.
In this case, the representation with the highest weight $\frac{1}{2}$
can be treated as a projective representation of $\SO(3)$.

Further, we also assume that
the Hamiltonian $H$ is written by using a function $h$ as 
\begin{align}
H=\sum_{k=0}^{\infty} h((\frac{k+1}{2})^2) I_{\frac{k}{2}}.
\Label{5-19-2}
\end{align}
Then, we have the following theorem.
\begin{theorem}\Label{T5-20b}
Assume the assumptions \eqref{5-19-1} and \eqref{5-19-2}.
For an input state 
\begin{align}
|{\phi} \rangle:=
\bigoplus_{k=0}^{\infty} \frac{{\beta}_{\frac{k}{2}}}{\sqrt{k+1}}
|\Psi_{\frac{k}{2}}\rangle\rangle, 
\Label{5-20-8}
\end{align}
we have the relations
\begin{align}
\varphi(\theta)
&:=
{\cal F}^{-1}[\phi](\tilde{g}_{\theta,\eta_1,\eta_2}) {\sin \frac{\theta}{2}}
=
\sqrt{2}\sum_{k=0}^{\infty}\beta_{\frac{k}{2}} \sin (k+\frac{1}{2}) \theta 
\Label{5-28-10},
\\
{\cal D}_R(|\phi\rangle)
&=\langle \varphi| w(Q) | \varphi\rangle , \quad
\langle \phi |H |\phi\rangle =\langle \varphi|h(P^2)| \varphi\rangle 
\Label{5-18-2n}.
\end{align}
Here $\varphi(\theta)$
is an odd function and is included in $L_{p}^2((-2\pi,2\pi])$.

Then,
the relations
\begin{align}
&\min_{\rho \in {\cal S}({\cal K}_{\hat{\SU(2)}})} 
\min_{M \in {\cal M}_{\cov}(\SU(2))} 
\{{\cal D}_R(\rho,M)| \Tr \rho H \le E \}
\nonumber \\
=&
\min_{ \{p_i\} }
\min_{\rho_i \in {\cal S}({\cal K}_{\hat{\SU(2)}})}
\min_{M_i \in {\cal M}_{\cov}(\SU(2))} 
\{\sum_i p_i{\cal D}_R(\rho_i,M_i)| \sum_i p_i \Tr \rho_i H \le E \}
\nonumber \\
=&\min_{|\phi\rangle \in L^2_n(\hat{\SU(2)})} 
\{ {\cal D}_R(|\phi\rangle)|
\langle \phi |H |\phi\rangle \le E 
\}\nonumber \\
=&
\min_{\varphi \in L^2_{p,\odd,n}((-2\pi,2\pi])} 
\{\langle \varphi| w(Q) | \varphi\rangle|
\langle \varphi|h(P^2)| \varphi\rangle \le E \}
\Label{5-18-2}
\end{align}
hold.

Further, an input state $|{\phi} \rangle$ given in (\ref{5-20-8})
with $\beta_{\frac{k}{2}}\ge 0$
satisfies the relation
\begin{align}
\min_{M\in {\cal M}_{\cov}(\SU(2))} 
{\cal D}_R (|{\phi}\rangle \langle {\phi}|,M)
= {\cal D}_R (|{\phi}\rangle)
= \eqref{5-18-2}
\end{align}
if and only if
\begin{align}
{\cal F}^{-1}[\phi](\tilde{g}_{\theta,\eta_1,\eta_2})
\sin \frac{\theta}{2}
= \argmin_{\varphi \in L^2_{p,\odd,n}((-2\pi,2\pi])} 
\{\langle \varphi| w(Q) | \varphi\rangle|
\langle \varphi|h(P^2)| \varphi\rangle \le E \}.
\end{align}
Additionally,
when $H
=\sum_{k=0}^{\infty} \frac{k}{2}(\frac{k}{2}+1) I_{\frac{k}{2}}
=\sum_{k=0}^{\infty} ((\frac{k+1}{2})^2-\frac{1}{4}) I_{\frac{k}{2}}$,
i.e., $h(x)=x-\frac{1}{4}$,
we have
\begin{align}
\eqref{5-18-2}
=\min_{\varphi \in L^2_{p,\odd,n}((-2\pi,2\pi])} 
\{\langle \varphi| w(Q) | \varphi\rangle|
\langle \varphi|P^2| \varphi\rangle \le E+\frac{1}{4} \}  
\Label{5-18-3}.
\end{align}
\end{theorem}
\begin{proof}
The second equation in \eqref{5-28-11} and the equation \eqref{5-4-2} 
yields that
\begin{align*}
{\cal F}^{-1}[\phi](\tilde{g}_{\theta,\eta_1,\eta_2}) \sin \frac{\theta}{2} 
=
\sum_{k=0}^{\infty} 
\frac{\beta_{\frac{k}{2}}}{\sqrt{k+1}} \sqrt{k+1}
\chi_{\frac{k}{2}}(\tilde{g}_{\theta,\eta_1,\eta_2})
\sin \frac{\theta}{2} 
=
\sum_{k=0}^{\infty} \beta_{\frac{k}{2}} \sin \frac{k+1}{2} \theta ,
\end{align*}
which implies \eqref{5-28-10}.

Due to the form of Hamiltonian, 
Theorem \ref{t6-24-1} implies the first, the second, and the third equations in
\eqref{5-18-2}.
So, we need to show only the fourth equation in \eqref{5-18-2}.
Thanks to Lemma \ref{11-16-2},
the minimum value\par\noindent$\min_{|\phi\rangle \in L^2(\hat{\SU(2)})} 
\{ {\cal D}_R(|\phi\rangle)|
\langle \phi |H |\phi\rangle \le E \}$
can be attained by the input state $|\phi\rangle$ with the form \eqref{5-20-8}.
Hence, for the minimization of ${\cal D}_R(|\phi\rangle)$, 
it is enough to consider the inputs with the form \eqref{5-20-8}.
Since $R(e,\tilde{g}_{\theta,\eta_1,\eta_2})$ depends only on $\theta$,
we have
\begin{align}
&{\cal D}_R(|\phi\rangle)
=
\int_{\SU(2)}
(\alpha_0- \sum_{k=1}^{\infty}\alpha_{\frac{k}{2}} \chi_{\frac{k}{2}}(g))
|{\cal F}^{-1}[\phi](g) |^2
\mu_{\SU(2)}(d g) \nonumber \\
=&
\int_{-2\pi}^{2\pi}
(\alpha_0- 
\sum_{k=1}^{\infty}\alpha_{\frac{k}{2}} 
(\sum_{l=0 \hbox{ or } \frac{1}{2}}^{\frac{k}{2}} 
\cos l \theta)
|{\cal F}^{-1}[\phi](\tilde{g}_{\theta,\eta_1,\eta_2})|^2
\sin^2 \frac{\theta}{2} \frac{d \theta}{2\pi} \nonumber \\
=&
\int_{-2\pi}^{2\pi}
w(\theta)
|\sum_{k=0}^{\infty} \beta_{\frac{k}{2}} \sin \frac{k+1}{2} \theta |^2
 \frac{d \theta}{2\pi} 
=
\int_{-2\pi}^{2\pi}
w(\theta)
|\varphi (\theta) |^2
 \frac{d \theta}{2\pi} ,
\Label{5-20-2}
\end{align}
where $\varphi(\theta):= \sum_{k=0}^{\infty} \beta_{\frac{k}{2}} \sin \frac{k+1}{2} \theta$.
Then, we have
\begin{align}
\langle \phi |H |\phi\rangle
=\sum_{k} h((\frac{k+1}{2})^2) \beta_{\frac{k}{2}}^2 
=\langle \varphi|h(P^2)| \varphi\rangle .
\Label{5-20-1}
\end{align}
Hence, we obtain \eqref{5-18-2n}.
Since any odd function with the period $4 \pi$ can be written as 
$\sqrt{2} \sum_{k=0}^{\infty} \beta_{\frac{k}{2}} \sin \frac{k+1}{2} \theta$,
the relations \eqref{5-20-2} and \eqref{5-20-1} yield
\begin{align*}
&\min_{|\phi\rangle \in L^2_n(\hat{\SU(2)})} 
\{ {\cal D}_R(|\phi\rangle)|
\langle \phi |H |\phi\rangle \le E 
\}\\
=&
\min_{\varphi \in L^2_{p,\odd,n}((-2\pi,2\pi])} 
\{\langle \varphi| w(Q) | \varphi\rangle|
\langle \varphi|h(P^2)| \varphi\rangle \le E \}.
\end{align*}
\end{proof}

\subsection{Constraint for available irreducible representation}\Label{s11-2}
Next, we restrict available weights to the set $
\Lambda_{n}:= \{0,\frac{1}{2},1 , \frac{3}{2}, \ldots, \frac{n}{2} \}$.
Then, we consider the risk function $R_{\SU(2)}$ 
on the system ${\cal K}_{\Lambda_{n}}$.
When the input state $|\phi\rangle$ has the form (\ref{5-20-8}),
\eqref{5-18-2n} in Theorem \ref{T5-20b} implies that
\begin{align}
{\cal D}_{R_{\SU(2)}}(|\phi \rangle)
=
1-\frac{1}{2}\sum_{k=0}^{n-1} 
(\beta_{\frac{k}{2}}\overline{\beta_{\frac{k+1}{2}}}+\beta_{\frac{k+1}{2}}
\overline{\beta_{\frac{k}{2}}}).
\end{align}
Hence, applying Lemma \ref{t6-10-1}, we have
\begin{align}
\min_{|\phi \rangle \in {\cal K}_{\Lambda_n} }
{\cal D}_{R_{\SU(2)}}(|\phi \rangle)
=
1-\cos \frac{\pi}{n+2}.
\end{align}
This fact can be also shown by the relation
$C^{\frac{1}{2}}_{\frac{k}{2},\frac{k'}{2}}=
\delta_{k,k'-1}+\delta_{k,k'+1}$
and Lemma \ref{11-16-2}.
Hence, Theorem \ref{th3} implies
\begin{align}
&
\min_{\rho \in {\cal S}({\cal K}_{\Lambda_n})} 
\min_{M \in {\cal M}_{\cov}(G)} 
{\cal D}_{R_{\SU(2)}}(\rho,M)
\nonumber \\
=&
\min_{ \{p_i\} }
\min_{\rho_i \in {\cal S}({\cal K}_{\Lambda_n})} 
\min_{M_i \in {\cal M}_{\cov}(G)} 
\sum_i p_i {\cal D}_{R_{\SU(2)}}(\rho_i,M_i) 
\nonumber \\
=&
\min_{|\phi \rangle \in {\cal K}_{\Lambda_n},n} 
{\cal D}_{R_{\SU(2)}}(|\phi \rangle)
=1-\cos \frac{\pi}{n+2}.
\Label{3-12-1vc}
\end{align}
Due to Lemma \ref{t6-10-1},
the minimum is attained only by 
$|\phi\rangle= C \sum_{k=0}^{n}\frac{
\sin \frac{(k+1) \pi}{n+2}}{\sqrt{k+1}}|\Psi_{\frac{k}{2}} \rangle$,
i.e.,
${\cal F}^{-1}[\phi](\tilde{g}_{\theta,\eta_1,\eta_2})
= C \sqrt{2} \sum_{k=0}^n 
\sin \frac{(k+1) \pi}{n+2}
\frac{\sin \frac{k+1}{2}\theta}{\sin \frac{\theta}{2}}$.

\subsection{Typical energy constraint}\Label{s11-2-2}
Next, we consider the risk function 
$R_{\SU(2)}(e,\hat{g})= 1- \frac{1}{2}\chi_{\frac{1}{2}} (\hat{g})$
and the Hamiltonian 
$H= \sum_{k=0}^{\infty} 
\frac{k}{2}(\frac{k}{2}+1)I_{\frac{k}{2}}$.
In this case, the function $w(\theta)$ is given as $1-\cos \frac{\theta}{2}$.
Then, thanks to Theorem \ref{T5-20b}, 
the minimum error can be characterized by the following value.
\begin{align}
\kappa_{\SU(2)}(E)
:=
\min_{\varphi \in L^2_{p,\odd,n}((-2\pi,2\pi])} 
\{\langle \varphi| I-\cos (\frac{Q}{2})| \varphi\rangle|
\langle \varphi| P^2 | \varphi\rangle \le E+\frac{1}{4}\} .
\Label{3-13-11b}
\end{align}
For example, we can show that
\begin{align}
\kappa_{\SU(2)}(0)=1 \Label{5-10-10b}.
\end{align}
This fact can be also checked by the following way.
In fact, the condition $\langle \phi |H|\phi\rangle=0$
can be realized only when $\beta_0=1$ and $\beta_{\frac{k}{2}}=0$ with $k \neq 0$,
i.e., $\varphi(\theta)=\sqrt{2}\sin \frac{\theta}{2}$.
In this case, we have $\int_{-2\pi}^{2\pi}
(1-\cos (\hat{\theta}))
|\varphi(\hat{\theta})|^2 \frac{d\hat{\theta}}{{4\pi}} 
=1$.
Hence, we see (\ref{5-10-10b}).

Now, we consider the case with non-zero $E$.
Since the condition of Lemma \ref{L4-25-2} hold, $\kappa_{\SU(2)}(E)$ is convex.
Hence, we employ Lemma \ref{L5-10-1} to calculate $\kappa_{\SU(2)}(E)$,
and consider 
the minimum 
\begin{align*}
\gamma_{\SU(2)}(s):=&
\min_{\varphi\in L^2_{p,n}((-2\pi, 2\pi ])}
\langle \varphi| (I-\cos (\frac{Q}{2})) +  s P^2| \varphi \rangle \\
=& \min_{\varphi\in L^2_{p,n}((-\pi/2,\pi/2 ])}
\langle \varphi| (I-\cos (2Q))+ \frac{s P^2}{16}| \varphi \rangle .
\end{align*}
So, $\gamma_{\SU(2)}(s)$ can be characterized as 
the minimum $\gamma_{\SU(2)}$
having the solution in $L^2_{p,n}((-\pi/2,\pi/2 ])$ 
of the following differential equation.
\begin{align}
\frac{s}{16}\frac{d^2}{d\theta^2} \varphi(\theta) + 
( \gamma_2- 1 +\cos (2\theta))
\varphi(\theta) =0,
\end{align}
which is equivalent to
\begin{align}
\frac{d^2}{d\theta^2} \varphi(\theta)+ 
( \frac{16(\gamma_{\SU(2)}- 1)}{s} + \frac{16}{s}\cos (2\theta))
\varphi(\theta) =0,
\end{align}
In order to find the minimum $\gamma_{\SU(2)}$, we employ Mathieu equation (\ref{5-10-7}),
whose detail is summarized in Subsection \ref{asB}.
Hence, using the function $b_2$ given in Subsection \ref{asB},
we have 
$\gamma_{\SU(2)}(s)
=\frac{s b_2(-\frac{8}{s} )}{16} +1
=\frac{s b_2(\frac{8}{s})}{16} +1$.
So, applying (\ref{5-10-2}) to $\kappa_{\SU(2)}(E)$, and
combining 
the facts given in Subsection \ref{asB},
we obtain the following theorem.

\begin{theorem}\Label{T3-13-3b}
\begin{align}
\kappa_{\SU(2)}(E)
=\max_{s>0} 
\frac{s b_2(\frac{8}{s})}{16} +1 -s(E+\frac{1}{4})
\Label{5-10-13b}.
\end{align}
The minimum \eqref{3-13-11b} is attained by the input state $|\phi\rangle$
with the measurement ${\cal M}_{|I\rangle \langle I|}$
if and only if 
${\cal F}^{-1}[\phi](\tilde{\rho}_{\theta,\eta_1,\eta_2})=
\frac{\se_2(\frac{\theta}{4}, -\frac{8}{s_E})}{\sin \frac{\theta}{2}}$,
where $s_E$ is 
$\argmax_{s>0} 
\frac{s b_2(\frac{8}{s})}{16} +1 -s(E+\frac{1}{4})$
and Mathieu function $\se_2$ is given in Subsection \ref{asB}.
\end{theorem}

Using the formula \eqref{5-10-13b},
we can calculate $\kappa_{\SU(2)}(E)$ as Fig. \ref{SU(2)-1}.

\begin{figure}[htbp]
\begin{center}
\scalebox{0.6}{\includegraphics[scale=1.2]{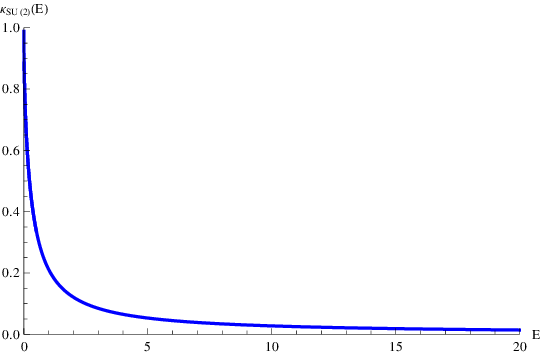}}
\end{center}
\caption{Graph of $\kappa_{\SU(2)}(E)$.}
\Label{SU(2)-1}
\end{figure}%

By using the expansion \eqref{5-20-16b} for $b_2$, 
as $s \to 0$, $\gamma_{\SU(2)}(s)$ can be expanded to
\begin{align*}
\gamma_{\SU(2)}(s)
\cong
\frac{s(-2\frac{8}{s}+6 \sqrt{\frac{8}{s}}-\frac{5}{4})}{16} +1
=\frac{3}{2}\sqrt{\frac{s}{2}}- \frac{5s}{64}.
\end{align*}
As is shown in Lemma \ref{L5-10-1},
$s_E$ is decreasing as a function of $E$.
Hence, when $E$ is large, 
solving the equation $\gamma_{\SU(2)}'(s_E)=E+\frac{1}{4}$, we approximately obtain
$s_E\cong 2\cdot(\frac{3}{8E+21/8})^2$.
Hence,
\begin{align}
&\kappa_{\SU(2)}(E)= \gamma_{\SU(2)}(s_E)-s_E 
(E+\frac{1}{4})
\cong 
\frac{3}{2\sqrt{2}}\sqrt{s_E}- \frac{5s_E}{64}
- s_E(E+\frac{1}{4}) \nonumber \\
=& \frac{3}{2 \sqrt{2}}\sqrt{s_E}- s_E
(\frac{21}{64}+E) 
\cong 
\frac{9}{2(8E+21/8)}- 
2 \cdot (\frac{3}{8E+21/8})^2
\frac{(\frac{21}{8}+8E)}{8} \nonumber \\
= &
\frac{9}{2(8E+21/8)}- 
\frac{9}{4(8E+21/8)}
= 
\frac{9}{4(8E+21/8)}
\cong
\frac{9}{32E}-\frac{7 \cdot 3^3 }{2^{11} E^2} .
\Label{5-10-4b}
\end{align}
As is shown in Fig. \ref{SU(2)-L},
while 
the first order approximation $\kappa_{1,\SU(2),\infty}(E):=\frac{9}{32E}$ gives a good approximation 
for $\kappa_{\SU(2)}(E)$ with a large $E$,
the second order approximation 
$\kappa_{2,\SU(2),\infty}(E):=\frac{9}{32E}-\frac{7 \cdot 3^3 }{2^{11} E^2} $
much improves the approximation 
for $\kappa_{\SU(2)}(E)$ with a large $E$.
Hence, we have the following asymptotic characterization.
\begin{align}
&
\lim_{E \to \infty}
E
\min_{|{\phi} \rangle \in L^2_n(\hat{\SU(2)}) }
\{{\cal D}_{R_{\SU(2)}}(|\phi \rangle)|
\langle \phi | H |\phi \rangle \le E \} 
=
\frac{9}{32} 
\Label{3-13-7x}.
\end{align}

\begin{figure}[htbp]
\begin{center}
\scalebox{0.6}{\includegraphics[scale=1.2]{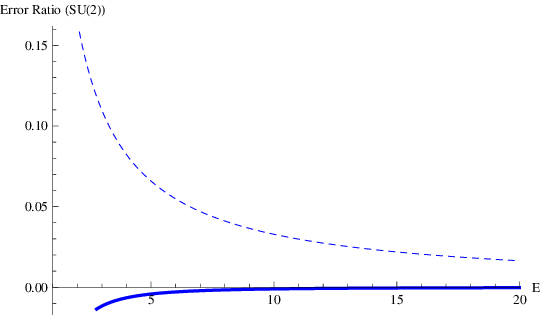}}
\end{center}
\caption{Comparison of two approximations $\kappa_{1,\SU(2),\infty}$
and $\kappa_{2,\SU(2),\infty}$
of $\kappa_{\SU(2)}$ with a large $E$.
Thick line expresses 
the error ratio$\frac{\kappa_{2,\SU(2),\infty}(E)-\kappa_{\SU(2)}(E)}{\kappa_{\SU(2)}(E)}$,
and
dashed line expresses 
the error ratio$\frac{\kappa_{1,\SU(2),\infty}(E)-\kappa_{\SU(2)}(E)}{\kappa_{\SU(2)}(E)}$.}
\Label{SU(2)-L}
\end{figure}%

Next, we consider the case when $E$ is small.
By using the expansion \eqref{exp3} for $b_2$,
when $s$ is large,
$\gamma_{\SU(2)}(s)$ can be expanded to
\begin{align*}
\gamma_{\SU(2)}(s)
\cong
\frac{s(4-\frac{1}{12}(\frac{8}{s})^2 +\frac{5}{13824}(\frac{8}{s})^4)}{16} +1
= 
\frac{s}{4}+1-\frac{1}{3s}  +\frac{5}{54 s^3}.
\end{align*}
When $E$ is small, 
since $\gamma_{\SU(2)}'(s)\cong \frac{1}{4}+\frac{1}{3 s^2}-\frac{5}{18s^4} $
solving the equation $\gamma_{\SU(2)}'(s_E)=E+\frac{1}{4}$, 
we approximately obtain
$s_E
\cong \sqrt{\frac{1}{3E}+\frac{5}{6}}
\cong \sqrt{\frac{1}{3E}}(1+\frac{5}{4}E)$.
Hence,
\begin{align}
&\kappa_{\SU(2)}(E)= 
\gamma_{\SU(2)}(s_E)-s_E( E+\frac{1}{4}) \nonumber \\
\cong &
\frac{s_E}{4}+1-\frac{1}{3s_E}  +\frac{5}{54 s_E^3}
-s_E( E+\frac{1}{4}) 
=
1+\frac{1}{3s_E}(\frac{5-18 s_E^2}{18 s_E^2}) -s_E  E \nonumber \\
=&
1+\frac{1}{3(\sqrt{\frac{1}{3E}}(1+\frac{5}{4}E))}
(\frac{5-18 (\frac{1}{3E}+\frac{5}{6})}{18 (\frac{1}{3E}+\frac{5}{6})}) 
-(\sqrt{\frac{1}{3E}}(1+\frac{5}{4}E))  E \nonumber \\
=&
1-\frac{\sqrt{E}}{\sqrt{3}(1+\frac{5}{4}E)}
(\frac{1+\frac{5E}{3}}{1+\frac{5E}{2}}) 
-(\sqrt{\frac{E}{3}}(1+\frac{5}{4}E))  \nonumber \\
\cong &
1-\frac{\sqrt{E}}{\sqrt{3}}
(1-\frac{5}{4}E+\frac{5E}{3}-\frac{5E}{2}) 
-(\sqrt{\frac{E}{3}}(1+\frac{5}{4}E))\nonumber   \\
= &
1-\frac{\sqrt{E}}{\sqrt{3}}
(1-\frac{25}{12}E) 
-\sqrt{\frac{E}{3}}(1+\frac{5}{4}E)  
= 
1-\frac{2}{\sqrt{3}}\sqrt{E}
+\frac{5}{6\sqrt{3}}E^{\frac{3}{2}}  
\Label{5-10-4e1}
\end{align}
This expansion with $E=0$ coincides with \eqref{5-10-10b}. 
As is shown in Fig. \ref{SU(2)-S},
while 
the first order approximation $\kappa_{1,\SU(2),+0}(E):=1-\frac{2}{\sqrt{3}}\sqrt{E}
$ gives a good approximation 
for $\kappa_{\SU(2)}(E)$ with a small $E$,
the second order approximation 
$\kappa_{2,\SU(2),+0}(E):=
1-\frac{2}{\sqrt{3}}\sqrt{E}
+\frac{5}{6\sqrt{3}}E^{\frac{3}{2}}$
much improves the approximation 
for $\kappa_{\SU(2)}(E)$ with a small $E$.

\begin{figure}[htbp]
\begin{center}
\scalebox{0.6}{\includegraphics[scale=1.2]{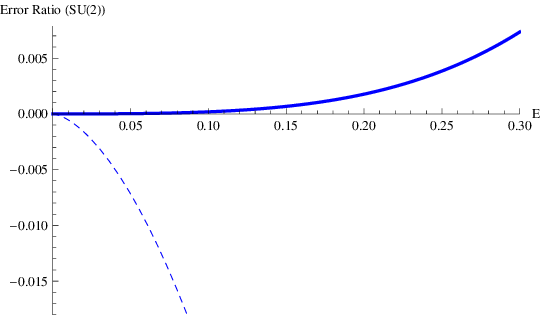}}
\end{center}
\caption{Comparison of two approximations 
$\kappa_{1,\SU(2),+0}$
and
$\kappa_{2,\SU(2),+0}$
of $\kappa_{\SU(2)}$ with a small $E$.
Thick line expresses 
the error ratio$\frac{\kappa_{2,\SU(2),+0}(E)-\kappa_{\SU(2)}(E)}{\kappa_{\SU(2)}(E)}$,
and
dashed line expresses 
the error ratio$\frac{\kappa_{1,\SU(2),+0}(E)-\kappa_{\SU(2)}(E)}{\kappa_{\SU(2)}(E)}$.}
\Label{SU(2)-S}
\end{figure}%

For the asymptotic optimality condition with respect to input states,
we obtain the following lemma.
\begin{lemma}\Label{T3-13-3y}
For a sequence $\{E_l\}$ satisfying $E_l \to \infty$ as $l \to \infty$,
we focus on a sequence of input states 
$\{\phi_{E_l}\}$ in $L^2_n(\hat{\SU(2)})$
with the form 
$|\phi_{E_l}\rangle=\oplus_{k=0}^{\infty} 
\frac{\beta_{\frac{k}{2},E_l}}{\sqrt{k+1} } |\Psi_{\frac{k}{2}}\rangle\rangle$
satisfying that
$\langle \phi_{E_l}| H |\phi_{E_l}\rangle \le E_l$. 
We also define the odd function
\begin{align*}
\tilde{\phi}_{E_l}(\lambda):= 
\left\{
\begin{array}{ll}
(2\pi E_l)^{1/4}
\beta_{\lfloor 2 \sqrt{E_l} \lambda  -\frac{1}{2}\rfloor/2, E_l} 
& \hbox {if } \lambda >0 \\
-(2\pi E_l)^{1/4}
\beta_{\lfloor 2 \sqrt{E_l} |\lambda|  -\frac{1}{2}\rfloor/2, E_l} 
& \hbox {if } \lambda <0 \\
0 
& \hbox {if } \lambda =0.
\end{array}
\right.
\end{align*}
Then,
$\min_{M\in {\cal M}_{\cov}(\SU(2))} 
{\cal D}_R (|\phi_{E_l}\rangle \langle \phi_{E_l}|,M)
= {\cal D}_R (|\phi_{E_l}\rangle)
\cong \frac{9}{32 E_l} $ as $l \to \infty$
if and only if the sequence of functions
$\tilde{\phi}_{E_l}(\lambda)$ goes to 
$3^{\frac{3}{4}}\lambda e^{-\frac{3\lambda^2}{4}}$
as $l \to \infty$ on $\bR_+$.
\end{lemma}

\begin{proof}
The relation 
$\min_{M\in {\cal M}_{\cov}(\SU(2))} 
{\cal D}_R (|\phi_{E_l}\rangle \langle \phi_{E_l}|,M)
=
{\cal D}_R (|\phi_{E_l}\rangle)$ holds by the same reason as Lemma \ref{T3-13-3c}.
Now, we choose 
the function $\varphi_{E_l}(\theta):= \sum_{k=0}^{\infty}
\sqrt{2}\beta_{\frac{k}{2},E_l} \sin \frac{k+1}{2} \theta
\in L^2_{p}((-2\pi,2\pi])$, and the parameters
$\lambda:= \pm \frac{k+1}{2\sqrt{E} }$ and $\hat{g}:= \sqrt{E} \hat{\theta}$.
Then, 
we have
\begin{align*}
&\frac{1}{E_l}
\sum_{k=0}^{\infty}
(\frac{k+1}{2})^2 |\beta_{\frac{k}{2},E_l}|^2
=
\sum_{k=0}^{\infty} 
(\frac{k+1}{2\sqrt{E_l}})^2 \frac{1}{\sqrt{2\pi E_l}} 
|\tilde{\phi}_{E_l}(\frac{k+1}{2\sqrt{E_l}})|^2 \\
\to &
2 \int_{0}^{\infty} 
\lambda^2 |\tilde{\phi}(\lambda )|^2 \frac{d \lambda}{\sqrt{2\pi}} 
= 
\int_{-\infty}^{\infty} 
\lambda^2 |\tilde{\phi}(\lambda )|^2 \frac{d \lambda}{\sqrt{2\pi}} 
\end{align*}
as $E_l \to \infty$.
Similarly,
since
\begin{align*}
& \frac{\varphi_{E_l}(\frac{\hat{g}}{\sqrt{E_l}})}{(2\pi E_l)^{\frac{1}{4}}} 
=
\sum_{k=0}^{\infty}
\frac{e^{-i \frac{k+1}{2} \frac{\hat{g}}{\sqrt{E_l}}}}{\sqrt{2}} 
\beta_{\frac{k}{2},E_l} 
\frac{1}{(2\pi E_l)^{\frac{1}{4}}} 
-\sum_{k=0}^{\infty}
\frac{e^{i \frac{k+1}{2} \frac{\hat{g}}{\sqrt{E_l}}}}{\sqrt{2}} 
\beta_{\frac{k}{2},E_l} 
\frac{1}{(2\pi E_l)^{\frac{1}{4}}} \\
=&
\sum_{k'=-\infty}^{\infty} 
e^{-i \frac{k'}{2\sqrt{E_l}} \hat{g}} \tilde{\phi}_{E_l}(\frac{k'}{2\sqrt{E_l}}) 
\frac{1}{2\sqrt{\pi E_l}} 
\to 
\sqrt{2} \int_{-\infty}^{\infty}
e^{-i \lambda \hat{g}} \tilde{\phi}(\lambda) 
\frac{d \lambda}{\sqrt{2\pi}} 
=
\sqrt{2}{\cal F}^{-1}[\tilde{\phi}](\hat{g}) ,
\end{align*}
where $k'=k+1,-(k+1)$,
we have
\begin{align*}
& E_l
\int_{-2\pi}^{2\pi}
(1-\cos \frac{\hat{\theta}}{2})
|\varphi_{E_l}(\hat{\theta})|^2 
\frac{d\hat{\theta}}{{4\pi}}  
\cong
 E_l
\int_{-2\pi}^{2\pi}
\frac{\hat{\theta}^2}{8}
|\varphi_{E_l}(\hat{\theta})|^2 
\frac{d\hat{\theta}}{{4\pi}}  \\
= &
\int_{-2\pi \sqrt{E_l}}^{2\pi \sqrt{E_l}}
\frac{\hat{g}^2}{8}
|\varphi_{E_l}(\frac{\hat{g}}{\sqrt{E_l}})|^2 
\frac{d\hat{g}}{{4\pi} \sqrt{E_l}} 
\to
\int_{-\infty}^{\infty}
\frac{\hat{g}^2}{8}
|{\cal F}^{-1}[\tilde{\phi}](\hat{g})|^2 
\frac{d\hat{g}}{\sqrt{2\pi} } .
\end{align*}
In Lemma \ref{L3-13-1}, the minimum \eqref{5-9-3} with $E=1$
is attained only by 
$\tilde{\phi}(\lambda)=3^{\frac{3}{4}}\lambda e^{-\frac{3\lambda^2}{4}}$.
Hence, 
${\cal D}_R (|\phi_{E_l}\rangle) \cong \frac{9}{32 E_l} 
=\frac{1}{8 E_l} \cdot \frac{9}{4}$ as $l \to \infty$
if and only if 
$\tilde{\phi}_l(\lambda)$ goes to 
$3^{\frac{3}{4}}\lambda e^{-\frac{3\lambda^2}{4}}$
as $l \to \infty$.
\end{proof}

\subsection{Practical construction of asymptotically optimal estimator with energy constraint}\Label{s11-3}
While Lemma \ref{T3-13-3y} characterizes the asymptotically optimal estimator with energy constraint,
no practical construction is provided.
In this subsection, we give its practical construction.
For this purpose, we introduce th operators
on ${\cal K}_{\hat{\SU(2)}}$ as
\begin{align}
J_{l}:= \bigoplus_{k=0}^{\infty} f_{\frac{k}{2}}(\sigma_l) .
\end{align}
Since
\begin{align}
f_{\frac{k}{2}}(\sigma_1) ^2+f_{\frac{k}{2}}(\sigma_2) ^2 +f_{\frac{k}{2}}(\sigma_3) ^2
= 4 \cdot \frac{k}{2}(\frac{k}{2}+1),
\end{align}
the Hamiltonian $H$ is characterized as 
\begin{align}
J_1^2+J_2^2+J_3^2 = 4 H,
\end{align}
Further, we have
\begin{align}
\Tr f_{\frac{k}{2}}(\sigma_l) ^2= 4 \cdot \frac{k+1}{3} \cdot \frac{k}{2}(\frac{k}{2}+1).
\Label{5-22-1}
\end{align}
Now, we give the tensor product system ${\cal K}_{\hat{\SU(2)}}^{\otimes n}$
and define the Hamiltonian $H^{(n)}$ as follows.
\begin{align}
H^{(n)}:= \frac{1}{4}\sum_{l=1}^3 (J_{l}^{(n)})^2, \quad
J_{l}^{(n)}:= \sum_{t=1}^n J_{l,t}^{(n)} , \quad
J_{l,t}^{(n)}:=I^{\otimes i-1} \otimes J_{l} \otimes I^{\otimes n-i}.
\end{align}
Now, we choose a state $|\phi\rangle =\oplus_{k=0}^{\infty}
\frac{\beta_{\frac{k}{2}}}{\sqrt{k+1}} |\Psi_{\frac{k}{2}} \rangle \rangle
\in {\cal K}_{\hat{\SU(2)}}$ with $\beta_{\frac{k}{2}}\ge 0$.
This state has the energy
$E_{\phi}:=\sum_{k=0}^\infty \frac{k}{2}(\frac{k}{2}+1)\beta_{\frac{k}{2}}^2$.

Now, we give a practical estimation protocol for the $n$-tensor-products system 
${\cal K}_{\hat{\SU(2)}}^{\otimes n}$
in the following way:
\begin{description}
\item[(2.1)]
We set the initial state $|\phi\rangle^{\otimes n}$
on  the tensor product system ${\cal K}_{\hat{\SU(2)}}^{\otimes n}$.
\item[(2.2)]
We apply the covariant measurement $M_{|{\cal I}\rangle \langle {\cal I}|}$
on each system 
${\cal K}_{\hat{\SU(2)}}$.
Then, we obtain $n$ outcomes ${g}_1, \ldots, {g}_n$.
Each outcome ${g}_i$ obeys the distribution
$p_{g}(g_i) \mu_{\SU(2)}(d g_i)$ where
$p_{g}(g_i):=
|\sum_{k=0}^{\infty}\frac{\beta_{k/2}}{\sqrt{k+1}} \Tr f_{k/2}(g_i^{-1}g)|^2
=
|\sum_{k=0}^{\infty}\frac{\beta_{k/2}}{\sqrt{k+1}} \chi_{k/2}(g_i^{-1}g)|^2$.

\item[(2.3)]
We apply the maximum likelihood estimator to the obtained outcomes
${g}_1, \ldots, {g}_n$.
Then, we obtain the final estimate $\hat{g}_n$.
That is, we decide $\hat{g}_n$ as
\begin{align}
\hat{g}_n:= \argmax_{g \in \SU(2)} \sum_{i=1}^n \log p_{g}(g_i).
\end{align}
\end{description}
We denote the above measurement with the output $\hat{g}_n$ by $M_n$.
Then, due to the following theorem,
the above protocol asymptotically realizes 
the minimum error under the energy constraint. 
\begin{theorem}\Label{T5-5-1}
Assume that
there exist at lest one even number $k_e \ge 0$ and one odd number $k_o >0$ 
such that $\beta_{\frac{k_e}{2}}>0$ and $\beta_{\frac{k_o}{2}}>0$.
Then, the relations
\begin{align}
\lim_{n \to \infty} n 
{\cal D}_R(|\phi^{\otimes n}\rangle, M_n)
&= \frac{9}{32} E_\phi \Label{5-1-5}\\
\langle\phi^{\otimes n}| H^{(n)} |\phi^{\otimes n}\rangle
&= E_{\phi} n \Label{5-1-6}
\end{align}
hold. That is,
\begin{align}
\lim_{n \to \infty} 
\langle\phi^{\otimes n}| H^{(n)} |\phi^{\otimes n}\rangle
{\cal D}_R(|\phi\rangle^{\otimes n}, M_n)
= \frac{9}{32}.
\end{align}
\end{theorem}

\begin{proof}
First, we show (\ref{5-1-6}).
Since $\langle\phi^{\otimes n}| J_{l,t}^{(n)} |\phi^{\otimes n}\rangle =
\langle\phi| J_{l} |\phi\rangle =0$,
the interaction terms vanish
so that
\begin{align*}
\langle\phi^{\otimes n}|
(J_{l}^{(n)})^2
|\phi^{\otimes n}\rangle
=&
\sum_{t=1}^n
\langle\phi^{\otimes n}|
(J_{l,t}^{(n)})^2
|\phi^{\otimes n}\rangle .
\end{align*}
Thus,
\begin{align*}
&\langle\phi^{\otimes n}|
H^{(n)}
|\phi^{\otimes n}\rangle
=
\frac{1}{4}\sum_{t=1}^n
\langle\phi^{\otimes n}|
(J_{1,t}^{(n)})^2+(J_{2,t}^{(n)})^2+(J_{3,t}^{(n)})^2
|\phi^{\otimes n}\rangle \\
=& \frac{n}{4} \langle\phi| J_{1}^2+J_{2}^2+J_{3}^2 |\phi\rangle 
=n E_{\phi},
\end{align*}
which implies (\ref{5-1-6}).

Next, we show (\ref{5-1-5}).
Due to the assumption, the map $g \mapsto p_{g}(g_i)$ is one-to-one.
Then, the Fisher information gives the asymptotic error.
Note that if the map $g \mapsto p_{g}(g_i)$ is one-to-one, e.g., the case when 
the distribution $p_{g}(g_i)$ is decided by the element $\varpi (g) \in \SO(3)$ 
with the projection $\varpi:\SU(2) \to \SO(3)$,
the Fisher information does not give the asymptotic error.
So, this assumption is crucial.

Due to the covariance of the estimator, it is enough to show (\ref{5-1-5})
in the case of $g= I$.
We choose the parameter $\hat{\vec{\theta}}_n$ as 
$\hat{g}_n= g_{\hat{\vec{\theta}}_n}$ by using the parametrization \eqref{5-19-5}.
For this purpose, we calculate the Fisher information
of the distribution family $\{p_{\vec{\theta}} \}$ with $p_{\vec{\theta}}(g):= p_{g_{\vec{\theta}}}(g)$.
Then, we can define the square root $\sqrt{p_{\vec{\theta}}(g)}$ as
$\sqrt{p_{\vec{\theta}}(g)}:=
\sum_{k=0}^{\infty}\frac{\beta_{k/2}}{\sqrt{k+1}} \chi_{k/2}(g^{-1}g_{\vec{\theta}}) \in \bR$
because all characters $\chi_{k/2}$ are real.

Since $\frac{d}{d \theta_s} g_{\vec{\theta}}|_{\theta=0}= i \sigma_s$
and 
$\frac{d}{d \theta_s} \sqrt{p_{\vec{\theta}}(g)}|_{\vec{\theta}=0}=
\sum_{k=0}^{\infty}\frac{\beta_{k/2}}{\sqrt{k+1}} \Tr f_{k/2}(g^{-1}) i \frac{f_{k/2}(\sigma_s)}{2}$,
the logarithmic derivative is given as
\begin{align}
l^s_{0}(g):=
2 \frac{d}{d \theta_s} \log \sqrt{p_{\vec{\theta}}(g)}|_{\vec{\theta}=0}
= 2 \frac{\frac{d}{d \theta_s} \sqrt{p_{\vec{\theta}}(g)}|_{\vec{\theta}=0}}{\sqrt{p_{\vec{\theta}}(g)}} 
\Label{5-22-2}
\end{align}
Since
$\Tr f_{k'/2}(g^{-1}) i f_{k'/2}(\sigma_t)$ is real,
we have
\begin{align}
\Tr f_{k'/2}(g^{-1}) i f_{k'/2}(\sigma_t)=
-i \Tr f_{k'/2}(\sigma_t)^{\dagger} f_{k'/2}(g^{-1})^{\dagger}
=-i \Tr f_{k'/2}(\sigma_t) f_{k'/2}(g).
\Label{5-22-3}
\end{align}
Further, we have
\begin{align}
&\int_{\SU(2)} 
\langle\langle f_{k/2}(\sigma_t) | f_{k/2}(g)    \rangle\rangle 
\langle\langle f_{k/2}(g^{-1})| f_{k/2}(\sigma_s)\rangle\rangle
\mu_{\SU(2)}(dg) \nonumber \\
=&
\left\{
\begin{array}{ll}
(k+1) \langle\langle  f_{k/2}(\sigma_t) | f_{k/2}(\sigma_s)\rangle\rangle & \hbox{if } k'=k \\
0 & \hbox{if } k' \neq k 
\end{array}
\right.
\Label{5-22-4}.
\end{align}
By combining \eqref{5-22-1}, \eqref{5-22-2}, \eqref{5-22-3}, and \eqref{5-22-4},
the Fisher information matrix is calculated to
\begin{align}
& J^{s,t}_{0}=
\int_{\SU(2)} 
l^s_{0}(g)l^t_{0}(g)
p_{\vec{\theta}}(g) \mu_{\SU(2)}(dg) \nonumber \\
=&
4 \int_{\SU(2)} 
\frac{d}{d \theta_s} \sqrt{p_{\vec{\theta}}(g)}|_{\vec{\theta}=0}
\frac{d}{d \theta_t} \sqrt{p_{\vec{\theta}}(g)}|_{\vec{\theta}=0}
\mu_{\SU(2)}(dg) \nonumber\\
=&
4 \int_{\SU(2)} 
\sum_{k=0}^{\infty}\frac{\beta_{k/2}}{\sqrt{k+1}} \Tr f_{k/2}(g^{-1}) i 
\frac{f_{k/2}(\sigma_s)}{2}
\sum_{k'=0}^{\infty}\frac{\beta_{k'/2}}{\sqrt{k'+1}} \Tr f_{k'/2}(g^{-1}) i 
\frac{f_{k'/2}(\sigma_t)}{2}
\mu_{\SU(2)}(dg) \nonumber\\
=&
\int_{\SU(2)} 
\sum_{k=0}^{\infty}\frac{\beta_{k/2}}{\sqrt{k+1}} \Tr f_{k/2}(g^{-1}) f_{k/2}(\sigma_s)
\sum_{k'=0}^{\infty}\frac{\beta_{k'/2}}{\sqrt{k'+1}} \Tr f_{k'/2}(\sigma_t) f_{k'/2}(g)
\mu_{\SU(2)}(dg) \nonumber\\
=&
\sum_{k=0}^{\infty}\frac{\beta_{k/2}^2}{k+1} 
\int_{\SU(2)} 
\langle\langle  f_{k/2}(\sigma_t) | f_{k/2}(g)\rangle\rangle 
\langle \langle f_{k/2}(g^{-1}) | f_{k/2}(\sigma_s)\rangle\rangle
\mu_{\SU(2)}(dg) \nonumber\\
=&
\sum_{k=0}^{\infty}\beta_{k/2}^2 
\langle\langle  f_{k/2}(\sigma_t) | f_{k/2}(\sigma_s)\rangle\rangle \nonumber\\
=&
\sum_{k=0}^{\infty}\beta_{k/2}^2 
\Tr f_{k/2}(\sigma_t) f_{k/2}(\sigma_s) 
=
\sum_{k=0}^{\infty}\beta_{k/2}^2 
\frac{\delta_{s,t}}{3} \cdot 4 \cdot \frac{k}{2}(\frac{k}{2}+1) 
=
\frac{4}{3}E_{\phi} \delta_{s,t}.
\Label{5-2-1}
\end{align}
Hence, $(J_0^{-1})_{i,j}= \frac{3 \delta_{i,j}}{4 E_{\phi}}$.
We have $ R_{\SU(2)}(e,g_\theta)=
1-\cos \frac{\|\vec{\theta} \|}{2}
\cong \frac{\|\vec{\theta} \|^2}{8}
= \frac{1}{8}((\theta^1)^2+(\theta^2)^2+(\theta^3)^2)$ when 
$\|\vec{\theta}\|$ is small.
Remember that
the asymptotic mean square error of the maximum likelihood estimator
can be characterized by 
the inverse of Fisher information matrix \cite{van}.
Hence, we obtain  
$n \rE_0 \hat{\theta}_n^i \hat{\theta}_n^j \to  (J_0^{-1})_{i,j}$,
where $\rE_{\vec{\theta}}$ expresses the expectation under the distribution $p_{\vec{\theta}} $.
Hence, 
we have
\begin{align*}
& n {\cal D}_R(|\phi^{\otimes n}\rangle, M_n)
= 
n \rE_0 [
\frac{1}{8}
((\hat{\theta}_n^1)^2+(\hat{\theta}_n^2)^2+(\hat{\theta}_n^3)^2)
] \\
\to &
\frac{1}{8} \sum_{i=1}^3 (J_0^{-1})_{i,i}
=\frac{1}{8}
 \sum_{i=1}^3 
\frac{3 \delta_{i,i}}{4 E_{\phi}}
=
\frac{9}{32 E_{\phi}}.
\end{align*}
Since the error of the maximum likelihood estimator does not depend on 
the true parameter, we obtain \eqref{5-1-5}.
\end{proof}

\begin{remark}
In fact, (\ref{5-2-1}) coincides with the symmetric logarithmic derivative (SLD) Fisher information matrix
by the following reason.
The SLD $L^i_0$ 
is given as 
\begin{align*}
\left.
\frac{d}{d \theta_i} 
| g_\theta/ \sqrt{2} \rangle\rangle 
\langle \langle 
g_\theta/ \sqrt{2}|
\right|_{\theta=0}
=
\frac{1}{2}
(L^i_0 |
I / \sqrt{2}\rangle\rangle \langle \langle 
I/ \sqrt{2}|
+
|
I / \sqrt{2}\rangle\rangle \langle \langle 
I/ \sqrt{2}|L^i_0),
\end{align*}
Then, $L^i_0=
2 (|
\sigma_i/ \sqrt{2} \rangle\rangle \langle \langle 
I/ \sqrt{2}|
+
|
I / \sqrt{2}\rangle\rangle \langle \langle 
\sigma_i/ \sqrt{2}|)$.
Thus, the SLD Fisher information matrix
is calculated to
\begin{align*}
J^{i,j}_{SLD,0}=
\Tr L^j_0 
\frac{1}{2}
(L^i_0 |
I/ \sqrt{2} \rangle\rangle \langle \langle 
I/ \sqrt{2}|
+
|
I/ \sqrt{2} \rangle\rangle \langle \langle 
I/ \sqrt{2}|L^i_0) 
= 4
\langle \langle 
\sigma_j/ \sqrt{2}| 
\sigma_i / \sqrt{2}\rangle\rangle .
\end{align*}
Hence, the Fisher information matrix given in (\ref{5-2-1}) coincides with the SLD Fisher information matrix.
This coincidence holds for a general pure state family $\{|\phi_\theta\rangle\}_\theta$ and a continuous POVM $|\omega \rangle \langle \omega |d \omega$
when the following conditions hold
because the above discussion uses only the following properties.
\begin{description}
\item[(1)] $\langle \phi_\theta| \omega \rangle $ is a real number.
\item[(2)] $\langle \phi_\theta| \frac{d}{d \theta_i}\phi_\theta|_{\theta=0} \rangle =0$.
\item[(3)] The probability of the set $\{\omega |\langle \phi_0| \omega \rangle=0\} $ is zero
when the true parameter $\theta$ is zero. 
\end{description}
\end{remark}

\subsection{Uncertainty relation}\Label{s11-4}
We consider the relation between the above results and the 
uncertainty relation on the space $L^2(\SU(2))$.
In fact, the group $\SU(2)$ is isomorphic to 
the 3-dimensional sphere $S^3$ by the correspondence
$g \mapsto (x^0(g),x^1(g),x^2(g),x^3(g))$,
where 
$x^j(g)$ ($j=0,1,2,3$) is given as
$g= x^0(g)I+\sum_{j=1}^3 x^j(g) \sigma_i$ for $g \in \SU(2)$.
This, we have $L^2(\SU(2))=L^2(S^3)$.
Then, we define the operator $Q_j$ as 
the multiplication of $x^j(g)$.
For the set of operators $\vec{Q}:=(Q_0,Q_1,Q_2,Q_3)$, 
we focus on the uncertainty 
\begin{align*}
&\Delta_{\varphi^*}^2\vec{Q}:=
\sum_{j=0}^3 \Delta_{\varphi^*}^2 Q_j 
=
\sum_{j=0}^3 \langle \varphi^* | Q_j |\varphi^* \rangle 
-\sum_{j=0}^3 \langle \varphi^* | Q_j |\varphi^* \rangle^2 
=
1 -\sum_{i=0}^3 \langle \varphi^* | Q_j |\varphi^* \rangle^2 
\end{align*}
for $\varphi^*\in L^2(\SU(2))$.
Next, we define the momentum operator $P_j$ on $L^2(\SU(2))$ as
\begin{align}
P_j \varphi^*(g):= 
\frac{d \varphi^*(e^{-it\frac{\sigma_j}{2}} g )}{dt}|_{t=0}.
\end{align}
For the set of momentum operators $\vec{P}:=(P_1,P_2,P_3)$,
we also consider the uncertainty 
\begin{align*}
&\Delta_{\varphi^*}^2\vec{P}:=\sum_{j=1}^3 \Delta_{\varphi^*}^2 P_j .
\end{align*}
Then, as the uncertainty relation, 
we consider trade-off between
$\Delta_{\varphi^*}^2\vec{Q}$
and $\Delta_{\varphi^*}^2 \vec{P}$,
which is formulated as the following theorem.

\begin{theorem}\Label{t5-29-1}
The minimum of the uncertainty $\Delta_{\varphi^*}^2\vec{Q}$
under the constraint for $\Delta_{\varphi^*}^2 \vec{P}$
is calculated as
\begin{align}
\min_{\varphi^* \in L^2_{n}(\SU(2))}
\{
\Delta_{\varphi^*}^2 \vec{Q}|
\Delta_{\varphi^*}^2 \vec{P} \le E \}
=1-(\min_{s>0}
s(E+\frac{1}{4})- \frac{sb_2(\frac{8}{s})}{16})^2.
\Label{5-12-20}
\end{align}
The minimum in \eqref{5-12-16}
is realized by $\varphi^*$
if and only if
there exists $g \in \SU(2)$ such that
$\varphi^*(g \tilde{g}_{\theta,\eta_1,\eta_2})
=
\frac{\se_2(\frac{\theta}{4}, -\frac{8}{s_E})}{\sin \frac{\theta}{2}}$,
where $s_E$ is given in Theorem \ref{T3-13-3b}.

Further, $\min_{s>0}
s(E+\frac{1}{4})- \frac{sb_2(\frac{8}{s})}{16}$ is a positive value,
and \eqref{5-12-20} has the asymptotic expansion 
$\frac{9}{16E} -\frac{5 \cdot 3^3}{2^9 E^2}$
as $E \to \infty$.
\end{theorem}

\begin{proof}
Due to the symmetry, without loss of generality, we can assume that 
$\langle \varphi^* |  Q_j |\varphi^* \rangle=0$ and 
$\langle \varphi^* | P_j  |\varphi^* \rangle=0$ for $j=1,2,3$.
Hence, using the Hamiltonian $H$ given in Subsection \ref{s11-2-2},
we obtain
\begin{align*}
&
\min_{\varphi^* \in L^2_{n}(\SU(2))}
\{
\Delta_{\varphi^*}^2 \vec{Q}|
\Delta_{\varphi^*}^2 \vec{P} \le E \} \\
=&
\min_{\varphi^* \in L^2_{n}(\SU(2))}
\{
1-\langle \varphi^*| Q_0 | \varphi^*\rangle^2 |
\langle \varphi^*| H | \varphi^*\rangle \le E ,
\langle \varphi^*| Q_j | \varphi^*\rangle =0, j=1,2,3\} \\
=&
1-
(
\max_{\varphi^* \in L^2_{n}(\SU(2))}
\{
\langle \varphi^*| Q_0 | \varphi^*\rangle |
\langle \varphi^*| H | \varphi^*\rangle \le E ,
\langle \varphi^*| Q_j | \varphi^*\rangle =0, j=1,2,3\})^2 .
\end{align*}
Using Theorems \ref{T5-20b} and \ref{T3-13-3b}, 
we have
\begin{align}
& \max_{\varphi^* \in L^2_{n}(\SU(2))}
\{
\langle \varphi^*| Q_0 | \varphi^*\rangle |
\langle \varphi^*| H | \varphi^*\rangle \le E ,
\langle \varphi^*| Q_j | \varphi^*\rangle =0, j=1,2,3\}\nonumber \\
\le & 
\max_{\varphi^* \in L^2_{n}(\SU(2))}
\{
\langle \varphi^*| Q_0 | \varphi^*\rangle |
\langle \varphi^*| H | \varphi^*\rangle \le E \} \nonumber\\
= &
\max_{\varphi \in L^2_{p,\odd,n}([-\pi,\pi])}
\{
\langle \varphi| \cos Q| \varphi\rangle |
\langle \varphi| P^2 | \varphi\rangle \le E+\frac{1}{4} \} \nonumber\\
=&
\min_{s>0}
s(E+\frac{1}{4})- \frac{sb_2(\frac{8}{s})}{16}
\ge 0,
\Label{5-29-1}
\end{align}
where $\varphi$ is chosen as 
$\varphi(\theta)= \varphi^*(\tilde{g}_{\theta,\eta_1,\eta_2})\sin \frac{\theta}{2}$.
The minimum in right hand side of \eqref{5-29-1}
can be realized by
$\varphi^*(\tilde{g}_{\theta,\eta_1,\eta_2})
=
\frac{\se_2(\frac{\theta}{4}, -\frac{8}{s_E})}{\sin \frac{\theta}{2}}$.
This function satisfies the condition
$\langle \varphi^*| Q_j | \varphi^*\rangle =0$ for $j=1,2,3$.
Hence, we obtain the equality in \eqref{5-29-1}.

For a general function $\varphi^*$, there exists an element $g\in \SU(2)$
such that $\varphi^*(g \tilde{g}_{\theta,\eta_1,\eta_2})$ satisfies the condition.
$\langle \varphi^*| Q_j | \varphi^*\rangle =0$ for $j=1,2,3$.
So, $\varphi^*$ attains 
the minimum in \eqref{5-12-16}
if and only if
$\varphi^*(g \tilde{g}_{\theta,\eta_1,\eta_2})
=
\frac{\se_2(\frac{\theta}{4}, -\frac{8}{s_E})}{\sin \frac{\theta}{2}}$,
where $s_E$ is given in Theorem \ref{T3-13-3b}.

Further, we have
\begin{align}
1-(\min_{s>0}
s(E+\frac{1}{4})- \frac{sb_2(\frac{8}{s})}{16})^2
\cong 1-(1-(
\frac{9}{32E}-\frac{7 \cdot 3^3 }{2^{11} E^2} ))^2
\cong \frac{9}{16E} -\frac{5 \cdot 3^3}{2^9 E^2}.
\end{align}
\end{proof}

Next, as another type of uncertainty relation,
we consider the trade-off between
$\Delta_{\varphi^*}^2 \vec{Q}$
and $\Delta_{\varphi^*,\max} \vec{P}$, which is defined as
the square root of the maximum eigenvalue of 
$\sum_{j=1}^3 (P_j- \langle \varphi^*|P_j |\varphi^*\rangle)^2$
such that 
the corresponding projection $A$ satisfies $\langle \varphi^*| A |\varphi^*\rangle>0$.

Then, we obtain the following theorem.
\begin{theorem}
The minimum of the uncertainty $\Delta_{\varphi^*}^2 \vec{Q}$
under the constraint for $\Delta_{\varphi^*,\max} \vec{P} $
is calculated as
\begin{align}
\min_{\varphi^* \in L^2_{n}(\SU(2))}
\{
\Delta_{\varphi^*}^2 \vec{Q}|
\Delta_{\varphi^*,\max} P \le E \} =
\sin^2 \frac{\pi}{\lfloor 2(\sqrt{E^2+\frac{1}{4}}-\frac{1}{2}) \rfloor+2}
\Label{5-12-17b}.
\end{align}
The minimum uncertainty is realized by $\varphi^*\in L^2_{n}(\SU(2))$
if and only if
there exists $g \in \SU(2)$ such that
$\varphi^*(g \tilde{g}_{\theta,\eta_1,\eta_2})
=
C \sqrt{2} \sum_{k=0}^{n_0} 
\sin \frac{(k+1) \pi}{n_0+2}
\frac{\sin \frac{k+1}{2}\theta}{\sin \frac{\theta}{2}}$,
where $C$ is the normalizing constant,
and $n_0:=\lfloor 2(\sqrt{E^2+\frac{1}{4}}-\frac{1}{2}) \rfloor$.
\end{theorem}
\begin{proof}
Due to the symmetry, without loss of generality, we can assume that 
$\langle \varphi^* |  Q_j |\varphi^* \rangle=0$ and 
$\langle \varphi^* | P_j  |\varphi^* \rangle=0$ for $j=1,2,3$.
Hence, using the Hamiltonian $H$ given in Subsection \ref{s11-2-2},
we obtain
\begin{align*}
&
\min_{\varphi^* \in L^2_{n}(\SU(2))}
\{
\Delta_{\varphi^*}^2 \vec{Q}|
\Delta_{\varphi^*,\max}^2 \vec{P} \le E \} \\
=&
\min_{\varphi^* \in L^2_{n}(\SU(2))}
\{
1-\langle \varphi^*| Q_0 | \varphi^*\rangle^2 |
{\cal F}[\varphi^* ] \in {\cal K}_{\Lambda_{n_0}},
\langle \varphi^*| Q_j | \varphi^*\rangle =0, j=1,2,3\} \\
=&
1-
(
\max_{\varphi^* \in L^2_{n}(\SU(2))}
\{
\langle \varphi^*| Q_0 | \varphi^*\rangle |
{\cal F}[\varphi^* ] \in {\cal K}_{\Lambda_{n_0}},
\langle \varphi^*| Q_j | \varphi^*\rangle =0, j=1,2,3\})^2 .
\end{align*}
Similar to \eqref{5-29-1},
using Theorem \ref{T5-20b} and \eqref{3-12-1vc}, 
we have
\begin{align}
& \min_{\varphi^* \in L^2_{n}(\SU(2))}
\{
1- \langle \varphi^*| Q_0 | \varphi^*\rangle |
{\cal F}[\varphi^* ] \in {\cal K}_{\Lambda_{n_0}},
\langle \varphi^*| Q_j | \varphi^*\rangle =0, j=1,2,3\}\nonumber \\
=& 
\min_{|\phi\rangle \in {\cal K}_{\Lambda_{n_0},n}}
\{
{\cal R}_{\SU(2)}(|\phi\rangle )
|
\langle \varphi^*| Q_j | \varphi^*\rangle =0, j=1,2,3\}\nonumber \\
\ge & 
\min_{|\phi\rangle \in {\cal K}_{\Lambda_{n_0},n}}
{\cal R}_{\SU(2)}(|\phi\rangle )
=
1- \cos \frac{\pi}{\lfloor 2(\sqrt{E^2+\frac{1}{4}}-\frac{1}{2}) \rfloor+2}.
\Label{5-29-1b}
\end{align}
Since $1- \cos^2 \frac{\pi}{\lfloor 2(\sqrt{E^2+\frac{1}{4}}-\frac{1}{2}) \rfloor+2}
=\sin^2 \frac{\pi}{\lfloor 2(\sqrt{E^2+\frac{1}{4}}-\frac{1}{2}) \rfloor+2}$,
we obtain \eqref{5-12-17b}.
Similar to Theorem \ref{t5-29-1}, the condition for realizing the minimum in \eqref{3-12-1vc}
yields the condition for realizing the minimum in \eqref{5-12-17b}.
\end{proof}

\section{3-dimensional Special Orthogonal Group $\SO(3)$}\Label{s11b}
\subsection{General structure of $\SO(3)$ estimation}\Label{s11b-1}
Next, we consider the group $\SO(3)$, whose universal covering group is $\SU(2)$.
That is, there is the projection $\varpi:\SU(2) \to \SO(3)$.
When $\lambda$ is even, the representation $f_{\lambda}$ of $\SU(2)$ gives
the representation of $\SO(3)$.
When $\lambda$ is odd, the representation $f_{\lambda}$ of $\SU(2)$ gives
the projective representation of $\SO(3)$ with the same factor system,
which will be denoted by $-1$.
More precisely, firstly, we define the projective representation $f_{\frac{1}{2}}$
of $\SO(3)$ by $f_{\frac{1}{2}}(g):=f_{\frac{1}{2}}(g')$ with a choice of $g'\in 
\varpi^{-1}(g)$.
Next, we define the projective representation $f_{k+\frac{1}{2}}$ 
of $\SO(3)$ whose factor system is the same as $f_{\frac{1}{2}}$ for $k>0$.
Then, we have $\hat{\SU(2)}= \hat{\SO(3)}\cup \hat{\SO(3)}[-1]$.
That is, we describe the elements of $\hat{\SO(3)}$ and $\hat{\SO(3)}[-1]$
by using the maximal weight of the representation of $\SU(2)$.
Using two kinds of parameterizations of $\SU(2)$,
we introduce two kinds of parameterizations of $\SO(3)$ as
$\varpi_{\vec{\theta}}:=\varpi(g_{\vec{\theta}})$ with 
the range $\{\vec{\theta}|\|\vec{\theta}\| \le \pi \}$
and
$\tilde{\varpi}_{\theta,\eta_1,\eta_2}:=\varpi(\tilde{g}_{\theta,\eta_1,\eta_2})$ with the range
$\theta \in (-\pi,\pi]$, 
$\eta_1\in (-\frac{\pi}{2},\frac{\pi}{2}]$,
$\eta_2\in (-\frac{\pi}{2},\frac{\pi}{2}]$.

When the risk function $R$ satisfies the condition \eqref{6-26-14}, 
the risk function is written as 
\begin{align}
R(e,\hat{g})= \alpha_0- \sum_{k=1}^{\infty} \alpha_{k} \chi_{k}.
\Label{5-19-1b}
\end{align}
Then, we obtain
$R(e,\tilde{\varpi}_{\theta,\eta_1,\eta_2})=w(\theta)$, 
where we define the even function 
$w(\theta):=(\alpha_0- \sum_{k=1}^{\infty}\alpha_{k} 
(\sum_{l=0}^{k} \cos l \theta))$ with the period $2\pi$.
As a typical risk function, we often adopt the risk function 
$R_{\SO(3)}(e,\hat{g})= \frac{1}{2}(3- \chi_{1}(\hat{g}))
=\frac{1}{4}(4- |\Tr \hat{g}|^2)$
by using the gate fidelity,
and is written as
$R_{\SO(3)}(e,\varpi_{\vec{\theta}})= 1- \cos \|\vec{\theta} \|$
by using \eqref{5-28-12}.

Further, we also assume that
the Hamiltonians $H_1$ and $H_{-1}$ on ${\cal K}_{\hat{\SO(3)}}$ and ${\cal K}_{\hat{\SO(3)}[-1]}$ are written by using a function $h$ as 
\begin{align}
H_1=\sum_{k=0}^{\infty} h((k+\frac{1}{2})^2) I_{k} , \quad
H_{-1}=\sum_{k=0}^{\infty} h((k+1)^2) I_{k+\frac{1}{2}} 
\Label{5-19-2b}
\end{align}
For the description of the following theorem, we prepare the function space:
\begin{align*}
L^2_{a,\odd}((-\pi,\pi]):=
\{ f \in L^2_{p,\odd}((-2\pi,2\pi])|
f(\theta+ 2\pi)=-f(\theta) \}.
\end{align*}

The following theorem holds for the representation of $\SO(3)$.
\begin{theorem}\Label{t5-26-1}
Assume the assumptions \eqref{5-19-1b} and \eqref{5-19-2b}.
For an input state 
\begin{align}
|{\phi} \rangle:=
\bigoplus_{k=0}^{\infty} \frac{{\beta}_{k}}{\sqrt{2k+1}}
|\Psi_{k}\rangle\rangle \Label{5-20-8b}
\end{align}
on ${\cal K}_{\hat{\SO(3)}}$,
we have the relations
\begin{align}
\varphi(\theta)
&:=
{\cal F}^{-1}[\phi](\tilde{\varpi}_{\theta,\eta_1,\eta_2}) \sin \frac{\theta}{2}
=\sqrt{2}
\sum_{k=0}^{\infty}\beta_{k} \sin (k+\frac{1}{2}) \theta 
\Label{5-28-10b} \\
{\cal D}_R(|\phi\rangle)& = \langle \varphi| w(Q) | \varphi\rangle, 
\quad
\langle \phi |H_1 |\phi\rangle 
=\langle \varphi|h(P^2)| \varphi\rangle .
\Label{5-26-2}
\end{align}
Here, $\varphi(\theta)$ is an odd function belonging to $L^2_{a,\odd}((-\pi,\pi])$.
Then,
the relations
\begin{align}
&\min_{\rho \in {\cal S}({\cal K}_{\hat{\SO(3)}})} 
\min_{M \in {\cal M}_{\cov}(\SO(3))} 
\{{\cal D}_R(\rho,M)| \Tr \rho H_1 \le E \}
\nonumber \\
=&
\min_{ \{p_i\} }
\min_{\rho_i \in {\cal S}({\cal K}_{\hat{\SO(3)}})}
\min_{M_i \in {\cal M}_{\cov}(\SO(3))} 
\{\sum_i p_i{\cal D}_R(\rho_i,M_i)| \sum_i p_i \Tr \rho_i H_1 \le E \}
\nonumber \\
=&\min_{|\phi\rangle \in L^2_n(\hat{\SO(3)})} 
\{ {\cal D}_R(|\phi\rangle)|
\langle \phi |H_1 |\phi\rangle \le E 
\}\nonumber \\
=&
\min_{\varphi \in L^2_{a,\odd,n}((-\pi,\pi])} 
\{\langle \varphi| w(Q) | \varphi\rangle|
\langle \varphi|h(P^2)| \varphi\rangle \le E \}
\Label{5-18-2b} 
\end{align}
hold.
Further, an input state $\phi$ given in \eqref{5-20-8b}
with $\beta_{k}\ge 0$
satisfies the relation
\begin{align}
\min_{M\in {\cal M}_{\cov}(\SO(3))} 
{\cal D}_R (|{\phi}\rangle \langle {\phi}|,M)
= {\cal D}_R (|{\phi}\rangle)
= \eqref{5-18-2b}
\Label{5-26-4}
\end{align}
if and only if
the odd function 
$
{\cal F}^{-1}[\phi](\tilde{\varpi}_{\theta,\eta_1,\eta_2}) \sin \frac{\theta}{2}
=\sum_{k=0}^{\infty}\beta_{k} \sin (k+\frac{1}{2}) \theta$
realizes the minimum \eqref{5-18-2b}.
Additionally,
when 
$H_1=\sum_{k=0}^{\infty} k(k+1) I_{k}$
i.e., $h(x)=x-\frac{1}{4}$,
we have
\begin{align}
\eqref{5-18-2b}
=&\min_{\varphi \in L^2_{a,\odd,n}((-\pi,\pi])} 
\{\langle \varphi| w(Q) | \varphi\rangle|
\langle \varphi|P^2| \varphi\rangle \le E+\frac{1}{4} \} ,
\Label{5-18-3b} 
\end{align}
\end{theorem}

The following theorem holds for the projective representation of $\SO(3)$ with the factor system $-1$.
\begin{theorem}\Label{t5-26-2}
Assume the assumptions \eqref{5-19-1b} and \eqref{5-19-2b}.
For an input state 
\begin{align}
|{\phi} \rangle:=
\bigoplus_{k=0}^{\infty} \frac{{\beta}_{k+\frac{1}{2}}}{\sqrt{2k+2}}
|\Psi_{k+\frac{1}{2}}\rangle\rangle \Label{5-20-8c}
\end{align}
on ${\cal K}_{\hat{\SO(3)}[-1]}$,
we have the relations
\begin{align}
\varphi(\theta) &:=
{\cal F}^{-1}[\phi](\tilde{\varpi}_{\theta,\eta_1,\eta_2}) \sin \frac{\theta}{2}
=\sqrt{2}
\sum_{k=0}^{\infty}\beta_{k+\frac{1}{2}} \sin (k+1) \theta
\Label{5-28-10c} \\
{\cal D}_R(|\phi\rangle) &= \langle \varphi| w(Q) | \varphi\rangle, 
\quad
\langle \phi |H_{-1} |\phi\rangle 
=\langle \varphi|h(P^2)| \varphi\rangle .
\Label{5-26-3}
\end{align}
Here, $\varphi(\theta)$ is an odd function belonging to $L^2_{p,\odd}((-\pi,\pi])$.
Then, the relations
\begin{align}
&\min_{\rho \in {\cal S}({\cal K}_{\hat{\SO(3)}[-1]})} 
\min_{M \in {\cal M}_{\cov}(\SO(3))} 
\{{\cal D}_R(\rho,M)| \Tr \rho H_{-1} \le E \}
\nonumber \\
=&
\min_{ \{p_i\} }
\min_{\rho_i \in {\cal S}({\cal K}_{\hat{\SO(3)}[-1]})}
\min_{M_i \in {\cal M}_{\cov}(\SO(3))} 
\{\sum_i p_i{\cal D}_R(\rho_i,M_i)| \sum_i p_i \Tr \rho_i H_{-1} \le E \}
\nonumber \\
=&\min_{|\phi\rangle \in L^2_n(\hat{\SO(3)}[-1])} 
\{ {\cal D}_R(|\phi\rangle)|
\langle \phi |H_{-1} |\phi\rangle \le E 
\}\nonumber \\
=&
\min_{\varphi \in L^2_{p,\odd,n}((-\pi,\pi])} 
\{\langle \varphi| w(Q) | \varphi\rangle|
\langle \varphi|h(P^2)| \varphi\rangle \le E \}
\Label{5-18-2c}
\end{align}
hold.
Further, an input state $|\phi \rangle$
given in \eqref{5-20-8c} with $\beta_{k+\frac{1}{2}}$
satisfies the relation
\begin{align}
\min_{M\in {\cal M}_{\cov}(\SO(3))} 
{\cal D}_R (|{\phi}\rangle \langle {\phi}|,M)
= {\cal D}_R (|{\phi}\rangle)
= \eqref{5-18-2c}
\Label{5-26-5}
\end{align}
if and only if
the odd function 
${\cal F}^{-1}[\phi](\tilde{\varpi}_{\theta,\eta_1,\eta_2}) \sin \frac{\theta}{2}
=\sum_{k=0}^{\infty}\beta_{k+\frac{1}{2}} \sin (k+1) \theta$
realizes the minimum \eqref{5-18-2c}.
Additionally,
when 
$H_{-1}=\sum_{k=0}^{\infty} (k+\frac{1}{2})(k+\frac{3}{2}) I_{k+\frac{1}{2}}$,
i.e., $h(x)=x-\frac{1}{4}$,
we have
\begin{align}
\eqref{5-18-2c}
=&\min_{\varphi \in L^2_{p,\odd,n}((-\pi,\pi])} 
\{\langle \varphi| w(Q) | \varphi\rangle|
\langle \varphi|P^2| \varphi\rangle \le E+\frac{1}{4} \} .
\Label{5-18-3c}
\end{align}
\end{theorem}

\begin{proofof}{Theorem \ref{t5-26-1}}
Similar to \eqref{5-28-10}, the second equation in \eqref{5-28-11} and the equation \eqref{5-4-2} yield \eqref{5-28-10b}.
Due to the form of Hamiltonian, 
Theorem \ref{t6-24-1} implies 
the first, the second, and the third equations in \eqref{5-18-2b}.


Now, we show \eqref{5-26-2}.
Thanks to Lemma \ref{11-16-2},
the maximum value $\min_{|\phi\rangle \in L^2_n(\hat{\SO(3)})} 
\{ {\cal D}_R(|\phi\rangle)|
\langle \phi |H_1 |\phi\rangle \le E \}$
can be attained by the input state $|\phi\rangle$ with the form 
\eqref{5-20-8b}.
Hence, for the minimization of ${\cal D}_R(|\phi\rangle)$, 
it is enough to consider the inputs with the form \eqref{5-20-8b}.
We use the parametrization 
$\tilde{\varpi}_{\theta,\eta_1,\eta_2}$ for $\SO(3)$
with $\theta \in (-\pi,\pi ]$, $\phi_1\in [0,\pi)$, $\phi_2\in [ 0,2\pi)$.
Thanks to \eqref{5-20-9},
the equation \eqref{5-4-2} implies that
\begin{align*}
&{\cal F}^{-1}[\phi](\tilde{\varpi}_{\theta,\eta_1,\eta_2}) 
\sin \frac{\theta}{2} 
=
\sum_{k=0}^{\infty} 
\frac{\beta_{k}}{\sqrt{2k+1}} \sqrt{2k+1}
\chi_{k}(\tilde{\varpi}_{\theta,\eta_1,\eta_2})
\sin \frac{\theta}{2} \\
=&
\sum_{k=0}^{\infty} \beta_{k} \sin (k + \frac{1}{2}) \theta .
\end{align*}
In this case,
since $R(e,\tilde{\varpi}_{\theta,\eta_1,\eta_2})$ depends only on $\theta$,
we have
\begin{align}
&{\cal D}_R(|\phi\rangle)
=
\int_{\SO(3)}
(\alpha_0- \sum_{k=1}^{\infty}\alpha_{k} \chi_{k}(g))
|{\cal F}^{-1}[\phi](g) |^2
\mu_{\SO(3)}(d g) \nonumber \\
=&
\int_{-\pi}^{\pi}
(\alpha_0- 
\sum_{k=1}^{\infty}\alpha_{k} 
(\sum_{l=0}^{k} 
\cos l \theta)
|{\cal F}^{-1}[\phi](\tilde{\varpi}_{\theta,\eta_1,\eta_2})|^2
\sin^2 \frac{\theta}{2} \frac{d \theta}{\pi} \nonumber \\
=&
\int_{-\pi}^{\pi}
w(\theta)
|\sum_{k=0}^{\infty} \beta_{k} \sin (k + \frac{1}{2}) \theta |^2
 \frac{d \theta}{\pi} 
=
\int_{-\pi}^{\pi}
w(\theta)
|\varphi (\theta) |^2
 \frac{d \theta}{\pi} ,
\Label{5-20-2b}
\end{align}
where $\varphi(\theta):= \sum_{k=0}^{\infty} \beta_{k} \sin (k + \frac{1}{2}) \theta$.
Then, we have
\begin{align}
\langle \phi |H |\phi\rangle
=\sum_{k} h(((k + \frac{1}{2}))^2) |\beta_{k}|^2 
=\langle \varphi|h(P^2)| \varphi\rangle .
\Label{5-20-1b}
\end{align}
Hence, we obtain \eqref{5-26-2}.
Since any odd function with the period $2 \pi$ can be written as 
$\sum_{k=0}^{\infty} \beta_{k} \sin (k + \frac{1}{2}) \theta$,
the relations \eqref{5-20-2b} and \eqref{5-20-1b} yield
\begin{align*}
&\min_{|\phi\rangle \in L^2_n(\hat{\SO(3)})} 
\{ {\cal D}_R(|\phi\rangle)|
\langle \phi |H |\phi\rangle \le E 
\}\\
=&
\min_{\varphi \in L^2_{p,\odd,n}((-\pi,\pi])} 
\{\langle \varphi| w(Q) | \varphi\rangle|
\langle \varphi|h(P^2)| \varphi\rangle \le E \}.
\end{align*}
Hence, we obtain the fourth equation in \eqref{5-18-2b}.

Further, \eqref{5-4-1} of Lemma \ref{11-16-2}
the first equation in \eqref{5-26-4}.
Summarizing the above discussion, we can conclude that
\eqref{5-26-4} if and only if the odd function 
$\sum_{k=0}^{\infty}\beta_{k} \sin (k+\frac{1}{2}) \theta$ realizes the minimum \eqref{5-18-2b}.
\end{proofof}

\begin{proofof}{Theorem \ref{t5-26-2}}
Similar to \eqref{5-28-10}, the second equation in \eqref{5-28-11} and the equation \eqref{5-4-2} yield \eqref{5-28-10c}.
Similarly, we can show the first, the second, and the third equations in \eqref{5-18-2c}.
Next, we show the \eqref{5-26-3}.
Thanks to Lemma \ref{11-16-2},
the maximum value $\min_{|\phi\rangle \in L^2(\hat{\SO(3)}[-1])} 
\{ {\cal D}_R(|\phi\rangle)|
\langle \phi |H_1 |\phi\rangle \le E \}$
can be attained by the input state $|\phi\rangle$ with the form \eqref{5-20-8c}.

Hence, for the minimization of ${\cal D}_R(|\phi\rangle)$, 
it is enough to consider the inputs with the form \eqref{5-20-8c}.
Thanks to \eqref{5-20-10},
the equation \eqref{5-4-2} implies that
\begin{align*}
&{\cal F}^{-1}[\phi](\tilde{\varpi}_{\theta,\eta_1,\eta_2}) 
\sin \frac{\theta}{2} 
=
\sum_{k=0}^{\infty} 
\frac{\beta_{k+\frac{1}{2}}}{\sqrt{2k+2}} \sqrt{2k+2}
\chi_{k+\frac{1}{2}}(\tilde{\varpi}_{\theta,\eta_1,\eta_2})
\sin \frac{\theta}{2} \\
=&
\sum_{k=0}^{\infty} \beta_{k+\frac{1}{2}} 
\sin (k+1) \theta .
\end{align*}
In this case,
since $R(e,\tilde{\varpi}_{\theta,\eta_1,\eta_2})$ depends only on $\theta$,
we have
\begin{align}
&{\cal D}_R(|\phi\rangle)
=
\int_{\SO(3)}
(\alpha_0- \sum_{k=1}^{\infty}\alpha_{k} \chi_{k}(g))
|{\cal F}^{-1}[\phi](g) |^2
\mu_{\SO(3)}(d g) \nonumber \\
=&
\int_{-\pi}^{\pi}
(\alpha_0- 
\sum_{k=1}^{\infty}\alpha_{k} 
(\sum_{l=0}^{k} 
\cos l \theta)
|{\cal F}^{-1}[\phi](\tilde{\varpi}_{\theta,\eta_1,\eta_2})|^2
\sin^2 \frac{\theta}{2} \frac{d \theta}{\pi} \nonumber \\
=&
\int_{-\pi}^{\pi}
w(\theta)
|\sum_{k=0}^{\infty} \beta_{k+\frac{1}{2}} \sin (k+1) \theta |^2
 \frac{d \theta}{\pi} 
=
\int_{-\pi}^{\pi}
w(\theta)
|\varphi (\theta) |^2
 \frac{d \theta}{\pi} ,
\Label{5-20-2c}
\end{align}
where $\varphi(\theta):= \sum_{k=0}^{\infty} \beta_{k+\frac{1}{2}} \sin (k+1) \theta$.
Then, we have
\begin{align}
\langle \phi |H |\phi\rangle
=\sum_{k=0}^{\infty} h((k+1)^2) \beta_{k+\frac{1}{2}}^2 
=\langle \varphi|h(P^2)| \varphi\rangle .
\Label{5-20-1c}
\end{align}
We obtain \eqref{5-26-3}.
Since any odd function in $L^2_{a,\odd}((-\pi,\pi])$ can be written as 
$\sum_{k=0}^{\infty} \beta_{k+\frac{1}{2}} \sin (k+1) \theta$,
the relations \eqref{5-20-2c} and \eqref{5-20-1c} yield
\begin{align*}
&\min_{|\phi\rangle \in L^2_n(\hat{\SO(3)})} 
\{ {\cal D}_R(|\phi\rangle)|
\langle \phi |H |\phi\rangle \le E 
\}\\
=&
\min_{\varphi \in L^2_{a,\odd,n}((-\pi,\pi])} 
\{\langle \varphi| w(Q) | \varphi\rangle|
\langle \varphi|h(P^2)| \varphi\rangle \le E \}.
\end{align*}
Hence, we obtain \eqref{5-18-2c}.
We can show the equivalence condition by the same way as Theorem \ref{t5-26-1}.
\end{proofof}

\subsection{Constraint for available irreducible representations}\Label{s11b-2}
We restrict available weights to the set 
$\tilde{\Lambda}_{n}:= \{0,1, 2, \ldots, n \}$
or 
$\tilde{\Lambda}_{n+\frac{1}{2}}:= \{\frac{1}{2},\frac{3}{2}, \ldots, n+\frac{1}{2} \}$.
First, we consider the risk function $R_{\SO(3)}$ on the system 
${\cal K}_{\tilde{\Lambda}_{n}}$. 
When the input state $|\phi\rangle$ has the form (\ref{5-20-8b}),
Theorem \ref{T5-20b} implies that
\begin{align}
{\cal D}_{R_{\SO(3)}}(|\phi \rangle)
=1+\frac{1}{2}|\beta_{0}|^2
-\frac{1}{2}\sum_{k=0}^{n-1} 
(\beta_{k}\overline{\beta_{k+1}}+\beta_{k+1} \overline{\beta_{k}}).
\end{align}
This fact can be also shown by Lemma \ref{11-16-2} and
the relation 
\begin{align*}
C^1_{k,k'}=
\left\{
\begin{array}{ll}
\delta_{k,k'-1} +\delta_{k,k'}+\delta_{k,k'+1} & \hbox{ if } k>0 \\
\delta_{0,k'-1} & \hbox{ if } k=0.
\end{array}
\right.
\end{align*}
In order to find the minimum eigenvalue and the eigenvector,
we focus on the operator \eqref{5-26-1}.
Then, the discussion in Appendix \ref{a5-18} with $l=n+1$ 
implies that
\begin{align}
\min_{|\phi \rangle \in {\cal K}_{\tilde{\Lambda}_n,n} }
{\cal D}_{R_{\SO(3)}}(|\phi \rangle)
=
1-\cos \frac{2\pi}{2n+3}.
\end{align}
Hence, Theorem \ref{th3} implies
\begin{align}
&
\min_{\rho \in {\cal S}({\cal K}_{\tilde{\Lambda}_n})} 
\min_{M \in {\cal M}_{\cov}(\SO(3))} 
{\cal D}_{R_{\SO(3)}}(\rho,M)
\nonumber \\
=&
\min_{ \{p_i\} }
\min_{\rho_i \in {\cal S}({\cal K}_{\tilde{\Lambda}_n})} 
\min_{M_i \in {\cal M}_{\cov}(\SO(3))} 
\sum_i p_i {\cal D}_{R_{\SO(3)}}(\rho_i,M_i) 
\nonumber \\
=& 1-\cos \frac{2\pi}{2n+3}
\cong \frac{1}{2}(\frac{2\pi}{2n+3})^2
\cong \frac{\pi^2}{2n^2}.
\Label{3-12-1v}
\end{align}

Next, we consider the risk function $R_{\SO(3)}$ on the system 
${\cal K}_{\tilde{\Lambda}_{n+\frac{1}{2}}}$.
When the input state $|\phi\rangle$ has the form (\ref{5-20-8c}),
Theorem \ref{T5-20b} implies that
\begin{align}
{\cal D}_{R_{\SO(3)}}(|\phi \rangle)
=1
-\frac{1}{2}\sum_{k=0}^{n-1} 
(\beta_{k+\frac{1}{2}}\overline{\beta_{k+\frac{3}{2}}}+\beta_{k+\frac{3}{2}} \overline{\beta_{k+\frac{1}{2}}}).
\end{align}
This fact can be also shown by Lemma \ref{11-16-2} and
the relation 
$C^1_{k+\frac{1}{2},k'+\frac{1}{2}}=
\delta_{k,k'-1} +\delta_{k,k'}+\delta_{k,k'+1}$.
Hence, applying Lemma \ref{t6-10-1} with $m=n+1$,
we obtain
\begin{align}
\min_{|\phi \rangle \in {\cal K}_{\tilde{\Lambda}_{n+\frac{1}{2}},n} }
{\cal D}_{R_{\SO(3)}}(|\phi \rangle)
=
1-\cos \frac{\pi}{n+2}.
\end{align}
Hence, Theorem \ref{th3} implies
\begin{align}
&
\min_{\rho \in {\cal S}({\cal K}_{\tilde{\Lambda}_{n+\frac{1}{2}}})} 
\min_{M \in {\cal M}_{\cov}(\SO(3))} 
{\cal D}_{R_{\SO(3)}}(\rho,M)
\nonumber \\
=&
\min_{ \{p_i\} }
\min_{\rho_i \in {\cal S}({\cal K}_{\tilde{\Lambda}_{n+\frac{1}{2}}})} 
\min_{M_i \in {\cal M}_{\cov}(\SO(3))} 
\sum_i p_i {\cal D}_{R_{\SO(3)}}(\rho_i,M_i) \nonumber \\
=&1-\cos \frac{\pi}{n+2}
\cong \frac{\pi^2}{2n^2}.
\Label{3-12-1w}
\end{align}

Indeed, the asymptotic expansion in \eqref{3-12-1v} and \eqref{3-12-1w}
are given in \cite{Ba,Chi,Ha1}.
However, 
the exact calculations in \eqref{3-12-1v} and \eqref{3-12-1w}
are not given in these references.

Now, we consider the $n$-tensor product representation on $(\complex^2)^{\otimes n}$.
The relation ${\cal K}_{\tilde{\Lambda}_{m-1}} \subset 
{\cal K}_{(\complex^2)^{\otimes n}} \subset {\cal K}_{\tilde{\Lambda}_{m}}$
holds for the even case $n=2m$,
and 
the relation ${\cal K}_{\tilde{\Lambda}_{m-\frac{1}{2}}} 
\subset 
{\cal K}_{(\complex^2)^{\otimes n}} \subset {\cal K}_{\tilde{\Lambda}_{m+\frac{1}{2}}}$
holds for the odd case $n=2m+1$.
For the definition of ${\cal K}_{(\complex^2)^{\otimes n}}$, see \eqref{6-1-1}.
Hence, 
using \eqref{3-12-1v} and \eqref{3-12-1w},
we can recover the following proposition known in 
\cite{Ba,Chi,Ha1}.
\begin{proposition}\Label{p4-26-1x}
The relations 
\begin{align}
&
\lim_{n \to \infty}
n^2 
\min_{\rho \in {\cal S}((\complex^2)^{\otimes n})} 
\min_{M \in {\cal M}_{\cov}(\SO(3))} 
{\cal D}_{R_{\SO(3)}}(\rho,M)
\nonumber \\
=&
\lim_{n \to \infty}n^2
\min_{ \{p_i\} }
\min_{\rho_i \in {\cal S}((\complex^2)^{\otimes n})} 
\min_{M_i \in {\cal M}_{\cov}(\SO(3))} 
\sum_i p_i {\cal D}_{R_{\SO(3)}}(\rho_i,M_i) 
=
2\pi^2
\Label{3-12-1bx}
\end{align}
hold.
\end{proposition}

\subsection{Typical energy constraint}\Label{s11b-3}
Next, we consider the risk function 
$R_{\SO(3)}(e,\hat{g})= \frac{1}{2}(3- \chi_{1} (\hat{g}))$
and the Hamiltonian 
$H= \sum_{k=-\infty}^{\infty} 
\frac{k}{2}(\frac{k}{2}+1)I_{\frac{k}{2}}$.
In this case, the function $w(\theta)$ is given as $1-\cos \theta$.
Then, thanks to Theorem \ref{t5-26-1}, 
the minimum error with respect to the representation
can be characterized by the following value.
\begin{align}
\kappa_{\SO(3)}(E)
:=
\min_{\varphi \in L^2_{a,\odd,n}((-\pi,\pi])} 
\{\langle \varphi| I-\cos (Q)| \varphi\rangle|
\langle \varphi| P^2 | \varphi\rangle \le E+\frac{1}{4}\} .
\Label{3-13-11c}
\end{align}
Similarly,
thanks to Theorem \ref{t5-26-2}, 
the minimum error with respect to the representation 
with the factor system $-1$
can be characterized by the following value.
\begin{align}
\kappa_{\SO(3),[-1]}(E)
:=
\min_{\varphi \in L^2_{p,\odd,n}((-\pi,\pi])} 
\{\langle \varphi| I-\cos (Q)| \varphi\rangle|
\langle \varphi| P^2 | \varphi\rangle \le E+\frac{1}{4}\} .
\Label{3-13-11d}
\end{align}

For example, we can show that
\begin{align}
\kappa_{\SO(3)}(0)&=\frac{3}{2} \Label{5-10-10c}\\
\kappa_{\SO(3),[-1]}(\frac{3}{4})&=1 \Label{5-10-10d}.
\end{align}
These facts can be also checked by the following way.
In \eqref{5-10-10c}, the condition $\langle \phi |H|\phi\rangle=0$
can be realized only when $\beta_{0}=1$ and $\beta_{k}=0$ with $k \neq 0$,
i.e., $\varphi(\theta)=\sqrt{2}\sin \frac{\theta}{2}$.
In this case, we have $\int_{-\pi}^{\pi}
(1-\cos (\hat{\theta}))
|\varphi(\hat{\theta})|^2 \frac{d\hat{\theta}}{{2\pi}} 
=\frac{3}{2}$.
Hence, we see (\ref{5-10-10c}).
In \eqref{5-10-10d}, the condition $\langle \phi |H|\phi\rangle=\frac{3}{4}$
can be realized only when $\beta_{\frac{1}{2}}=1$ and $\beta_{k+\frac{1}{2}}=0$ with $k \neq 0$,
i.e., $\varphi(\theta)=\sqrt{2}\sin \theta$.
In this case, we have $\int_{-\pi}^{\pi}
(1-\cos (\hat{\theta}))
|\varphi(\hat{\theta})|^2 \frac{d\hat{\theta}}{{2\pi}} 
=1$.
Hence, we see (\ref{5-10-10d}).

Now, we consider the general case.
Since the condition of Lemma \ref{L4-25-2} hold, $\kappa_{\SO(3)}(E)$ and $\kappa_{\SO(3),[-1]}(E)$ are convex.
Hence, we employ Lemma \ref{L5-10-1} to calculate $\kappa_{\SO(3)}(E)$ and $\kappa_{\SO(3),[-1]}(E)$,
and consider 
the minimums 
\begin{align*}
\gamma_{\SO(3)}(s):=&
\min_{\varphi\in L^2_{a,n}((-\pi, \pi ])}
\langle \varphi| (I-\cos (Q)) +  s P^2| \varphi \rangle \\
=& \min_{\varphi\in L^2_{a,n}((-\pi/2,\pi/2 ])}
\langle \varphi| (I-\cos (2Q))+ \frac{s P^2}{4}| \varphi \rangle \\
\gamma_{\SO(3),[-1]}(s):=&
\min_{\varphi\in L^2_{p,n}((-\pi, \pi ])}
\langle \varphi| (I-\cos (Q)) +  s P^2| \varphi \rangle \\
=& \min_{\varphi\in L^2_{p,n}((-\pi/2,\pi/2 ])}
\langle \varphi| (I-\cos (2Q))+ \frac{s P^2}{4}| \varphi \rangle .
\end{align*}
So, 
$\gamma_{\SO(3)}(s)$ and $\gamma_{\SO(3),[-1]}(s)$ 
can be characterized as 
the minimums of
$\gamma$ having the solution in 
$L^2_{a,n}((-\pi/2,\pi/2 ])$ 
and
$L^2_{p,n}((-\pi/2,\pi/2 ])$ 
of the following differential equation, respectively.
\begin{align}
\frac{s}{4}\frac{d^2}{d\theta^2} \varphi(\theta) + 
( \gamma - 1 +\cos (2\theta))
\varphi(\theta) =0,
\end{align}
which is equivalent to
\begin{align}
\frac{d^2}{d\theta^2} \varphi(\theta)+ 
( \frac{4(\gamma - 1)}{s} + \frac{4}{s}\cos (2\theta))
\varphi(\theta) =0.
\end{align}
In order to find the minimums 
$\gamma_{\SO(3)}(s)$ and $\gamma_{\SO(3),[-1]}(s)$,
we employ Mathieu equation (\ref{5-10-7}),
whose detail is summarized in Subsection \ref{asB}.
Hence, using the functions $a_1$, $b_1$ and $b_2$ 
given in Subsection \ref{asB},
we have 
$\gamma_{\SO(3)}(s)
=\frac{s b_1(-\frac{2}{s} )}{4} +1
=\frac{s a_1(\frac{2}{s})}{4} +1$,
and
$\gamma_{\SO(3),[-1]}(s)
=\frac{s b_2(-\frac{2}{s} )}{4} +1
=\frac{s b_2(\frac{2}{s})}{4} +1$,
where we employ the relation \eqref{6-3-1}.
So, applying (\ref{5-10-2}) to 
$\kappa_{\SO(3)}(E)$
and
$\kappa_{\SO(3),[-1]}(E)$, and
combining 
the facts given in Subsection \ref{asB},
we obtain the following theorem.

\begin{theorem}\Label{T3-13-3cx}
The relations
\begin{align}
\kappa_{\SO(3)}(E)
=\max_{s>0} 
\frac{s a_1(\frac{2}{s})}{4} +1 -s(E+\frac{1}{4})
\Label{5-10-13c} \\
\kappa_{\SO(3),[-1]}(E)
=\max_{s>0} 
\frac{s b_2(\frac{2}{s})}{4}  +1 -s(E+\frac{1}{4})
\Label{5-10-13d} 
\end{align}
hold. The minimum \eqref{3-13-11c} 
is attained by the input state $|\phi\rangle$
with the measurement ${\cal M}_{|I\rangle \langle I|}$
if and only if 
${\cal F}^{-1}[\phi](\tilde{\varpi}_{\theta,\eta_1,\eta_2})=
\frac{\se_1(\frac{\theta}{2}, -\frac{2}{s_E})}{\sin \frac{\theta}{2}}$,
where $s_{E}$ is
$\argmax_{s>0} 
\frac{s a_1(\frac{2}{s})}{4} +1 -s(E+\frac{1}{4})$
and Mathieu function $\se_1$ is given in Subsection \ref{asB}.

Similarly, the minimum \eqref{3-13-11d} 
is attained by the input state $|\phi\rangle$
with the measurement ${\cal M}_{|I\rangle \langle I|}$
if and only if 
${\cal F}^{-1}[\phi](\tilde{\varpi}_{\theta,\eta_1,\eta_2})=
\frac{\se_2(\frac{\theta}{2}, -\frac{4}{s_E})}{\sin \frac{\theta}{2}}$,
where $s_{E}$ is
$\argmax_{s>0} 
\frac{s b_2(\frac{2}{s})}{4}  +1 -s(E+\frac{1}{4})$
and Mathieu function $\se_2$ is given in Subsection \ref{asB}.
\end{theorem}

Using the formula \eqref{5-10-13b},
we can calculate $\kappa_{\SO(3)}(E)$ and $\kappa_{\SO(3),[-1]}(E)$ 
as Fig. \ref{SO(3)-1}.

\begin{figure}[htbp]
\begin{center}
\scalebox{0.6}{\includegraphics[scale=1.2]{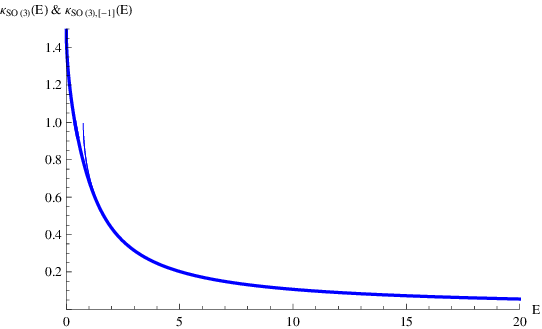}}
\end{center}
\caption{Comparison of $\kappa_{\SO(3)}(E)$ and $\kappa_{\SO(3),[-1]}(E)$.
Thick line expresses $\kappa_{\SO(3)}(E)$ and normal line expresses 
$\kappa_{\SO(3),[-1]}(E)$.
$\kappa_{\SO(3),[-1]}(E)$ is not defined when $E < \frac{3}{4}$.}
\Label{SO(3)-1}
\end{figure}%

Since, as is shown in \eqref{5-20-16b}, $a_1$ and $b_2$ have the same asymptotic expansion up to higher orders,
$\gamma_{\SO(3)}(s)$ and $\gamma_{\SO(3),[-1]}(s)$
have the same asymptotic expansion up to higher order as $s$ goes to zero.
Hence, 
$\kappa_{\SO(3)}(E)$ and $\kappa_{\SO(3),[-1]}(E)$
have the same asymptotic expansion up to higher order as $E$ goes to infinity.
So, As is shown in Fig. \ref{SO(3)-L2}, the difference rate
$\frac{\kappa_{\SO(3),[-1]}(E)-\kappa_{\SO(3)}(E)}{\kappa_{\SO(3)}(E)}$
goes to zero very quickly.

\begin{figure}[htbp]
\begin{center}
\scalebox{0.6}{\includegraphics[scale=1.2]{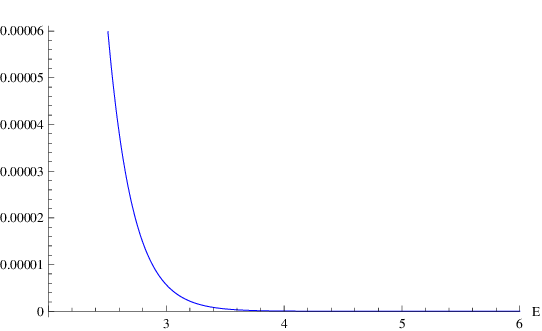}}
\end{center}
\caption{Graph of 
$\frac{\kappa_{\SO(3),[-1]}(E)-\kappa_{\SO(3)}(E)}{\kappa_{\SO(3)}(E)}$.}
\Label{SO(3)-L2}
\end{figure}%

Thanks to the expansion \eqref{5-20-16b},
As $s \to 0$, $\gamma_{\SO(3)}(s)$ can be expanded to
\begin{align*}
\gamma_{\SO(3)}(s)
&\cong
\gamma_{\SO(3),[-1]}(s)
\cong
\frac{s(-2\frac{2}{s}+6 \sqrt{\frac{2}{s}}-\frac{5}{4})}{4} +1
=\frac{3}{2}\sqrt{2s}- \frac{5s}{16}.
\end{align*}
As is shown in Lemma \ref{L5-10-1},
$s_E$ is decreasing as a function of $E$.
Hence, when $E$ is large, 
solving the equation 
$\gamma_{\SO(3)}'(s_E)
\cong
\gamma_{\SO(3),[-1]}'(s_E)
=E+\frac{1}{4}$, we approximately obtain
$s_E\cong \frac{1}{8}\cdot(\frac{3}{E+9/16})^2$.
Hence,
\begin{align*}
\kappa_{\SO(3)}(E)&= \gamma_{\SO(3)}(s_E)-s_E 
(E+\frac{1}{4})
\cong 
\frac{3\sqrt{2}}{2}\sqrt{s_E}- \frac{5s_E}{16}
- s_E(E+\frac{1}{4}) \nonumber \\
\kappa_{\SO(3),[-1]}(E)&= \gamma_{\SO(3),[-1]}(s_E)-s_E 
(E+\frac{1}{4})
\cong 
\frac{3\sqrt{2}}{2}\sqrt{s_E}- \frac{5s_E}{16}
- s_E(E+\frac{1}{4}) .
\end{align*}
and
\begin{align}
&\frac{3\sqrt{2}}{2}\sqrt{s_E}- \frac{5s_E}{16}
- s_E(E+\frac{1}{4}) 
= \frac{3\sqrt{2}}{2 }\sqrt{s_E}- s_E
(\frac{9}{16}+E) \nonumber \\
\cong &
\frac{3}{4}
\frac{3}{E+9/16}
-
\frac{1}{8}
\frac{9}{E+9/16}
= 
\frac{9}{8(E+9/16)}
\cong
\frac{9}{8E}
-
\frac{81}{128 E^2}
\Label{5-10-4e2}.
\end{align}
As is shown in Fig. \ref{SO(3)-L},
while 
the first order approximation $\kappa_{1,\SO(3),\infty}(E):=\frac{9}{8E}$ gives a good approximation 
for 
$\kappa_{\SO(3)}(E)$ 
and
$\kappa_{\SO(3),[-1]}(E)$ 
with a large $E$,
the second order approximation 
$\kappa_{2,\SO(3),\infty}(E):=\frac{9}{8E}
-\frac{81}{128 E^2}$
much improves the approximation 
for 
$\kappa_{\SO(3)}(E)$ 
and
$\kappa_{\SO(3),[-1]}(E)$ 
with a large $E$.
Hence, we have the following asymptotic characterization.
\begin{align}
&
\lim_{E \to \infty}
E
\min_{|{\phi} \rangle \in L^2_n(\hat{\SO(3)}) }
\{{\cal D}_{R_{\SO(3)}}(|\phi \rangle)|
\langle \phi | H |\phi \rangle \le E \} \nonumber \\
=&
\lim_{E \to \infty}
E
\min_{|{\phi} \rangle \in L^2_n(\hat{\SO(3)}[-1]) }
\{{\cal D}_{R_{\SO(3)}}(|\phi \rangle)|
\langle \phi | H |\phi \rangle \le E \} 
=
\frac{9}{8} 
\Label{3-13-7c}.
\end{align}

\begin{figure}[htbp]
\begin{center}
\scalebox{0.6}{\includegraphics[scale=1.2]{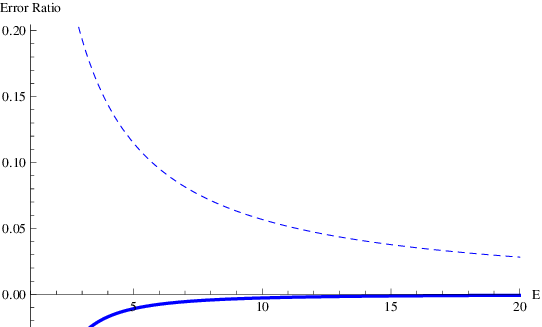}}
\end{center}
\caption{Comparison of two approximations 
$\kappa_{1,\SO(3),\infty}$
and
$\kappa_{2,\SO(3),\infty}$
of $\kappa_{\SO(3)}$ with a large $E$.
Thick line expresses 
the error ratio$\frac{\kappa_{2,\SO(3),\infty}(E)-\kappa_{\SO(3)}(E)}{\kappa_{\SO(3)}(E)}$,
and
dashed line expresses 
the error ratio$\frac{\kappa_{1,\SO(3),\infty}(E)-\kappa_{\SO(3)}(E)}{\kappa_{\SO(3)}(E)}$.}
\Label{SO(3)-L}
\end{figure}%

Next, we consider the case when $E$ is small.
Since $s$ is large,
by using the expansions \eqref{5-26-11} and \eqref{exp3},
$\gamma_{\SO(3)}(s)$ 
and
$\gamma_{\SO(3),[-1]}(s)$ 
can be expanded to
\begin{align*}
\gamma_{\SO(3)}(s)
&\cong 
\frac{s(1+\frac{2}{s}
-\frac{1}{8}(\frac{2}{s})^2 -\frac{1}{64}(\frac{2}{s})^3)}{4} +1
= 
\frac{s}{4}+\frac{3}{2}-\frac{1}{8 s}  -\frac{1}{32 s^3}.
\\
\gamma_{\SO(3),[-1]}(s)
&\cong
\frac{s(4-\frac{1}{12}(\frac{2}{s})^2 +\frac{5}{13824}(\frac{2}{s})^4)}{4} +1
= 
s+1-\frac{1}{12s}  +\frac{5}{3456 s^3}.
\end{align*}
Since $E>0$ is small, 
since $\gamma_{\SO(3)}'(s)\cong \frac{1}{4} +\frac{1}{8 s^2}
+\frac{1}{16 s^3}$,
solving the equation $\gamma_{\SO(3)}'(s_E)=E+\frac{1}{4}$, 
we approximately obtain
$s_E
\cong \sqrt{\frac{1}{8 E}}(1+\frac{E}{\sqrt{2}})$.
Hence,
\begin{align}
&\kappa_{\SO(3)}(E)= 
\gamma_{\SO(3)}(s_E)-s_E( E+\frac{1}{4}) \nonumber \\
\cong &
\frac{s_E}{4}+\frac{3}{2}-\frac{1}{8 s_E}  -\frac{1}{32 s_E^2}
-s_E( E+\frac{1}{4}) 
=
\frac{3}{2}-\frac{1}{8 s_E}  -\frac{1}{32 s_E^2}
-s_E  E \nonumber \\
=&
\frac{3}{2}-\frac{1}{8 s_E}(1+\frac{1}{4 s_E}) -s_E E
\cong
\frac{3}{2}-\frac{1}{8 s_E}(1+\sqrt{\frac{E}{2}}) -s_E E
\nonumber\\
\cong &
\frac{3}{2}-\sqrt{\frac{E}{8}}
(1-\frac{E}{\sqrt{2}})(1+\sqrt{\frac{E}{2}}) -
\sqrt{\frac{E}{8}} (1+\frac{E}{\sqrt{2}}) \nonumber\\
\cong &
\frac{3}{2}-\sqrt{\frac{E}{8}}
(1-\frac{E}{\sqrt{2}}+\sqrt{\frac{E}{2}}) -
\sqrt{\frac{E}{8}} (1+\frac{E}{\sqrt{2}})
=
\frac{3}{2}-\sqrt{\frac{E}{8}}
(2+\sqrt{\frac{E}{2}}) \nonumber \\
= &
\frac{3}{2}
-\frac{\sqrt{E}}{\sqrt{2}}
-\frac{E}{4}.
\Label{5-10-4g}
\end{align}
This expansion with $E \to 0$ coincides with \eqref{5-10-10c}. 
As is shown in Fig. \ref{SO(3)-S+},
while 
the first order approximation $\kappa_{1,\SO(3),+0}(E):=
\frac{3}{2}-\frac{\sqrt{E}}{\sqrt{2}}$ 
gives a good approximation 
for $\kappa_{\SO(3)}(E)$ with a small $E$,
the second order approximation 
$\kappa_{2,\SO(3),+0}(E):=
\frac{3}{2}-\frac{\sqrt{E}}{\sqrt{2}}
-\frac{E}{4}$
much improves the approximation 
for $\kappa_{\SO(3)}(E)$ with a small $E$.

\begin{figure}[htbp]
\begin{center}
\scalebox{0.6}{\includegraphics[scale=1.2]{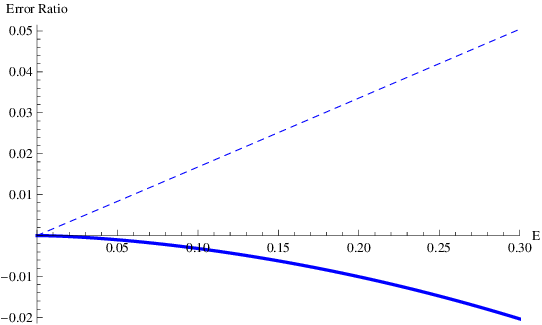}}
\end{center}
\caption{Comparison of two approximations 
$\kappa_{1,\SO(3),+0}$ and $\kappa_{2,\SO(3),+0}$
of $\kappa_{\SO(3)}$ with a small $E$.
Thick line expresses 
the error ratio$\frac{\kappa_{2,\SO(3),+0}(E)-\kappa_{\SO(3)}(E)}{\kappa_{\SO(3)}(E)}$,
and
dashed line expresses 
the error ratio$\frac{\kappa_{1,\SO(3),+0}(E)-\kappa_{\SO(3)}(E)}{\kappa_{\SO(3)}(E)}$.}
\Label{SO(3)-S+}
\end{figure}%

Next, we consider $\kappa_{\SO(3),[-1]}(E)$
in the case when $\eta:= E -\frac{3}{4}>0$ is small, 
since $\gamma_{\SO(3),[-1]}'(s)\cong 1+\frac{1}{12 s^2}-\frac{5}{1152 s^4}$,
solving the equation $\gamma_{\SO(3),[-1]}'(s_E)=E+\frac{1}{4}$, 
we approximately obtain
$s_E
\cong \sqrt{\frac{1}{12 \eta}-\frac{5}{96}}
\cong \sqrt{\frac{1}{12 \eta}}(1-\frac{5}{4}\eta)$.
Hence,
\begin{align}
&\kappa_{\SO(3),[-1]}(E)= 
\gamma_{\SO(3),[-1]}(s_E)-s_E( E+\frac{1}{4}) \nonumber \\
\cong &
s_E+1-\frac{1}{12 s_E}  +\frac{5}{3456 s_E^3}
-s_E( E+\frac{1}{4}) 
=
1-\frac{1}{12 s_E}  +\frac{5}{3456 s_E^3} -s_E  \eta \nonumber \\
=&
1-\frac{1}{12 s_E}(1-\frac{5}{288 s_E^2}) -s_E \eta
\cong
1-\frac{1}{12 s_E}(1-\frac{5}{288 (\frac{1}{12 \eta}-\frac{5}{96})}) -s_E \eta \nonumber\\
=&
1-\frac{1}{12 s_E}(1-\frac{5}{\frac{24}{\eta}-15}) -s_E \eta 
\cong 
1-\frac{1}{12 s_E}(1-\frac{5}{24} \eta) -s_E \eta \nonumber\\
\cong 
&
1-\sqrt{\frac{\eta}{12}}
(1
+\frac{5}{4}\eta
-\frac{5}{24} \eta) -\sqrt{\frac{\eta}{12}}(1-\frac{5}{4}\eta) 
=
1-\sqrt{\frac{\eta}{12}}
(2-\frac{5}{24} \eta) \nonumber\\
=&
1-
\frac{1}{\sqrt{3}} \eta^{\frac{1}{2}}
+
\frac{5}{48\sqrt{3}}\eta^{\frac{3}{2}}.
\Label{5-10-4f}
\end{align}
This expansion with $E \to \frac{3}{4}$ coincides with \eqref{5-10-10d}. 

As is shown in Fig. \ref{SO(3)-S-},
while 
the first order approximation $\kappa_{1,\SO(3),[-1],+0}(E):=
1-\frac{1}{\sqrt{3}} (E-\frac{3}{4})^{\frac{1}{2}}$ 
gives a good approximation 
for $\kappa_{\SO(3),[-1]}(E)$ with a small $E-\frac{3}{4}$,
the second order approximation 
$\kappa_{2,\SO(3),[-1],+0}(E):=
1-
\frac{1}{\sqrt{3}} (E-\frac{3}{4})^{\frac{1}{2}}
+\frac{5}{48\sqrt{3}}(E-\frac{3}{4})^{\frac{3}{2}}$
much improves the approximation 
for $\kappa_{\SO(3),[-1]}(E)$ with a small $E-\frac{3}{4}$.

\begin{figure}[htbp]
\begin{center}
\scalebox{0.6}{\includegraphics[scale=1.2]{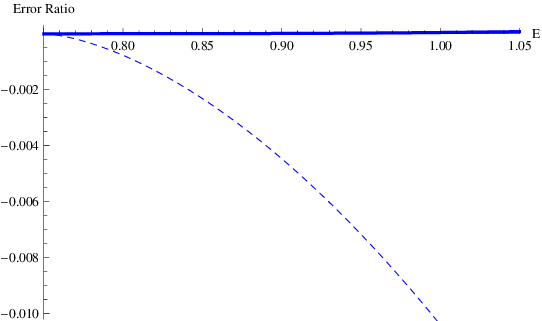}}
\end{center}
\caption{Comparison of two approximations 
$\kappa_{1,\SO(3),[-1],+0}$
and $\kappa_{2,\SO(3),[-1],+0}$
of $\kappa_{\SO(3)}$ with a small $E$.
Thick line expresses 
the error ratio$\frac{\kappa_{2,\SO(3),[-1],+0}(E)-\kappa_{\SO(3),[-1]}(E)}{\kappa_{\SO(3),[-1]}(E)}$, and dashed line expresses 
the error ratio$\frac{\kappa_{1,\SO(3),[-1],+0}(E)-\kappa_{\SO(3),[-1]}(E)}{\kappa_{\SO(3),[-1]}(E)}$.}
\Label{SO(3)-S-}
\end{figure}%

For the asymptotic optimality condition with respect to input states,
we obtain the following lemma.
\begin{lemma}\Label{T3-13-3y2}
[Case 1]
For a sequence $\{E_l\}$ satisfying $E_l \to \infty$ as $l \to \infty$,
we focus on a sequence of input states 
$\{\phi_{E_l}\}$ in $L^2_n(\hat{\SO(3)})$
with the form 
$|\phi_{E_l}\rangle=\oplus_{k=0}^{\infty} 
\frac{\beta_{k,E_l}}{\sqrt{2k+1} } 
|\Psi_{k}\rangle\rangle$
satisfying that
$\langle \phi_{E_l}| H |\phi_{E_l}\rangle \le E_l$. 
We also define the odd function
\begin{align*}
\tilde{\phi}_{E_l}(\lambda):= 
\left\{
\begin{array}{ll}
(\frac{\pi E_l}{2})^{1/4}
\beta_{\lfloor \sqrt{E_l} \lambda \rfloor, E_l} 
& \hbox {if } \lambda >0 \\
-
(\frac{\pi E_l}{2})^{1/4}
\beta_{\lfloor \sqrt{E_l} |\lambda| \rfloor, E_l} 
& \hbox {if } \lambda <0 \\
0 
& \hbox {if } \lambda =0.
\end{array}
\right.
\end{align*}
Then,
$\min_{M\in {\cal M}_{\cov}(\SO(3))} 
{\cal D}_R (|\phi_{E_l}\rangle \langle \phi_{E_l}|,M)
= {\cal D}_R (|\phi_{E_l}\rangle)
\cong \frac{9}{8 E_l} $ as $l \to \infty$
if and only if the sequence of functions
$\tilde{\phi}_l(\lambda)$ goes to 
$3^{\frac{3}{4}}\lambda e^{-\frac{3\lambda^2}{4}}$
as $l \to \infty$ on $\bR_+$.

[Case 2]
For a sequence $\{E_l\}$ satisfying $E_l \to \infty$ as $l \to \infty$,
we focus on a sequence of input states 
$\{\phi_{E_l}\}$ in $L^2_n(\hat{\SO(3)}[-1])$
with the form 
$|\phi_{E_l}\rangle=\oplus_{k=0}^{\infty} 
\frac{\beta_{k+\frac{1}{2},E_l}}{\sqrt{2k+2} } 
|\Psi_{k+\frac{1}{2}}\rangle\rangle$
satisfying that
$\langle \phi_{E_l}| H |\phi_{E_l}\rangle \le E_l$. 
We also define the odd function
\begin{align*}
\tilde{\phi}_{E_l}(\lambda):= 
\left\{
\begin{array}{ll}
(\frac{\pi E_l}{2})^{1/4}
\beta_{\lfloor \sqrt{E_l} \lambda -\frac{1}{2} \rfloor+\frac{1}{2}, E_l} 
& \hbox {if } \lambda >0 \\
-(\frac{\pi E_l}{2})^{1/4}
\beta_{\lfloor \sqrt{E_l} |\lambda| -\frac{1}{2} \rfloor+\frac{1}{2}, E_l} 
& \hbox {if } \lambda <0 \\
0 
& \hbox {if } \lambda =0.
\end{array}
\right.
\end{align*}
Then,
$\min_{M\in {\cal M}_{\cov}(\SO(3))} 
{\cal D}_R (|\phi_{E_l}\rangle \langle \phi_{E_l}|,M)
= {\cal D}_R (|\phi_{E_l}\rangle)
\cong \frac{9}{8 E_l} $ as $l \to \infty$
if and only if the sequence of functions
$\tilde{\phi}_l(\lambda)$ goes to 
$3^{\frac{3}{4}}\lambda e^{-\frac{3\lambda^2}{4}}$
as $l \to \infty$ on $\bR_+$.
\end{lemma}

\begin{proof}
The relation 
$\min_{M\in {\cal M}_{\cov}(\SO(3))} 
{\cal D}_R (|\phi_{E_l}\rangle \langle \phi_{E_l}|,M)
=
{\cal D}_R (|\phi_{E_l}\rangle)$ holds by the same reason as Lemma \ref{T3-13-3c}.
Now, we show the first part, i.e., we treat the case of the representation.
For this purpose, we choose 
the function $\varphi_{E_l}(\theta):= \sum_{k=0}^{\infty}
\sqrt{2} \beta_{k,E_l} \sin (k+\frac{1}{2}) \theta
\in L^2_{a}((-\pi,\pi])$, and the parameters
$\lambda:= \pm \frac{2k+1}{2\sqrt{E_l} }$ and $\hat{g}:= \sqrt{E_l} \hat{\theta}$.
Then, 
we have
\begin{align*}
&\frac{1}{E_l}
\sum_{k=0}^{\infty}
(k+\frac{1}{2})^2 |\beta_{k,E_l}|^2
=
\sum_{k=0}^{\infty} 
(\frac{2k+1}{2\sqrt{E_l}})^2 \frac{\sqrt{2}}{\sqrt{\pi E_l}} 
|\tilde{\phi}_{E_l}(\frac{2k+1}{2\sqrt{E_l}})|^2 \\
\to &
2 \int_{0}^{\infty} 
\lambda^2 |\tilde{\phi}(\lambda )|^2 \frac{d \lambda}{\sqrt{2\pi}} 
= 
\int_{-\infty}^{\infty} 
\lambda^2 |\tilde{\phi}(\lambda )|^2 \frac{d \lambda}{\sqrt{2\pi}} 
\end{align*}
as $E_l \to \infty$.
Similarly,
since
\begin{align*}
& \frac{\varphi_{E_l}(\frac{\hat{g}}{\sqrt{E_l}})}{(2\pi E_l)^{\frac{1}{4}}} 
=
\sum_{k=0}^{\infty}
\frac{e^{-i \frac{2k+1}{2} \frac{\hat{g}}{\sqrt{E_l}}}}{\sqrt{2}} 
\beta_{k,E_l} 
\frac{1}{(2\pi E_l)^{\frac{1}{4}}} 
-\sum_{k=0}^{\infty}
\frac{e^{i \frac{2k+1}{2} \frac{\hat{g}}{\sqrt{E_l}}}}{\sqrt{2}} 
\beta_{k,E_l} 
\frac{1}{(2\pi E_l)^{\frac{1}{4}}} \\
=&
\sum_{k'=-\infty}^{\infty} 
e^{-i \frac{k'}{2\sqrt{E_l}} \hat{g}} \tilde{\phi}_{E_l}(\frac{k'}{2\sqrt{E_l}}) 
\frac{1}{\sqrt{2\pi E_l}} 
\to 
\int_{-\infty}^{\infty}
e^{-i \lambda \hat{g}} \tilde{\phi}(\lambda) 
\frac{d \lambda}{\sqrt{2\pi}} 
=
{\cal F}^{-1}[\tilde{\phi}](\hat{g}) ,
\end{align*}
where $k'=2k+1,-(2k+1)$,
we have
\begin{align*}
& E_l
\int_{-\pi}^{\pi}
(1-\cos \hat{\theta})
|\varphi_{E_l}(\hat{\theta})|^2 
\frac{d\hat{\theta}}{{2\pi}}  
\cong
 E_l
\int_{-\pi}^{\pi}
\frac{\hat{\theta}^2}{2}
|\varphi_{E_l}(\hat{\theta})|^2 
\frac{d\hat{\theta}}{{2\pi}}  \\
= &
\int_{-\pi \sqrt{E_l}}^{\pi \sqrt{E_l}}
\frac{\hat{g}^2}{2}
|\varphi_{E_l}(\frac{\hat{g}}{\sqrt{E_l}})|^2 
\frac{d\hat{g}}{{2\pi} \sqrt{E_l}} 
\to
\int_{-\infty}^{\infty}
\frac{\hat{g}^2}{2}
|{\cal F}^{-1}[\tilde{\phi}](\hat{g})|^2 
\frac{d\hat{g}}{\sqrt{2\pi} } .
\end{align*}
In Lemma \ref{L3-13-1}, the minimum \eqref{5-9-3} with $E=1$
is attained only by 
$\tilde{\phi}(\lambda)=3^{\frac{3}{4}}\lambda e^{-\frac{3\lambda^2}{4}}$.
Hence, 
${\cal D}_R (|\phi_{E_l}\rangle) \cong \frac{9}{8 E_l} 
=\frac{1}{2 E_l} \cdot \frac{9}{4}$ as $l \to \infty$
if and only if 
$\tilde{\phi}_l(\lambda)$ goes to 
$3^{\frac{3}{4}}\lambda e^{-\frac{3\lambda^2}{4}}$
as $l \to \infty$.

Next, show the second part, i.e., 
we treat the case of the projective representation
with the factor system $-1$.
For this purpose, we choose 
the function $\varphi_{E_l}(\theta):= \sum_{k=0}^{\infty}
\beta_{k+\frac{1}{2},E_l} \sin (k+1) \theta
\in L^2_{p}((-\pi,\pi])$, and the parameters
$\lambda:= \pm \frac{k+1}{\sqrt{E_l} }$ and $\hat{g}:= \sqrt{E} \hat{\theta}$.
Then, we can show the desired argument by the similar way.
\end{proof}

\subsection{Practical construction of asymptotically optimal estimator with energy constraint}\Label{s11-3b}
While Lemma \ref{T3-13-3y2} characterizes the asymptotically optimal estimator with energy constraint,
no practical construction is provided.
In this subsection, we give its practical construction
under the same Hamiltonian as in Subsection \ref{s11-3}
on the subspaces 
${\cal K}_{\hat{\SO(3)}}^{\otimes n}$ and ${\cal K}_{\hat{\SO(3)}[-1]}^{\otimes n}$
of ${\cal K}_{\hat{\SU(2)}}^{\otimes n}$.

Now, we choose a state 
$|\phi \rangle =
\oplus_{k=0}^{\infty}
\frac{\beta_{k}}{\sqrt{2k+1}} |\Psi_{k} \rangle \rangle
\in {\cal K}_{\hat{\SO(3)}}$
or
$ = \oplus_{k=0}^{\infty}
\frac{\beta_{k+\frac{1}{2}}}{\sqrt{2k+2}} 
|\Psi_{k+\frac{1}{2}} \rangle \rangle
\in {\cal K}_{\hat{\SO(3)}[-1]}$
with $\beta_{k}, \beta_{k+\frac{1}{2}} \ge 0$.
The state $|\phi \rangle$ has the energy
$E_{\phi}:=\sum_{k=0}^\infty k(k+1)|\beta_{k}|^2$
or $\sum_{k=0}^\infty (k+\frac{1}{2})(k+\frac{3}{2})|\beta_{k+\frac{1}{2}}|^2$.
Then, we give a practical estimation protocol for the $n$-tensor-products system 
${\cal K}_{\hat{\SO(3)}}^{\otimes n}$ or ${\cal K}_{\hat{\SO(3)}[-1]}^{\otimes n}$
in the following way:
\begin{description}
\item[(3.1)]
We set the initial state 
$|\phi\rangle^{\otimes n}$ 
on  the tensor product system 
${\cal K}_{\hat{\SO(3)}}^{\otimes n}$ or ${\cal K}_{\hat{\SO(3)}[-1]}^{\otimes n}$.
\item[(3.2)]
We apply the covariant measurement $M_{|{\cal I}\rangle \langle {\cal I}|}$
on each system 
${\cal K}_{\hat{\SO(3)}}$ or ${\cal K}_{\hat{\SO(3)}[-1]}$.
Then, we obtain $n$ outcomes ${g}_1, \ldots, {g}_n$.
Each outcome ${g}_i$ obeys the distribution
$p_{g}(g_i) \mu_{\SO(3)}(d g_i)$, where
$p_{g}(g_i):=
|\sum_{k=0}^{\infty}\frac{\beta_{k}}{\sqrt{2k+1}} \Tr f_{k}(g_i^{-1}g)|^2$
or
$|\sum_{k=0}^{\infty}\frac{\beta_{k+\frac{1}{2}}}{\sqrt{2k+2}} 
\Tr f_{k+\frac{1}{2}}(g_i^{-1}g)|^2$.
\item[(3.3)]
We apply the maximum likelihood estimator to the obtained outcomes
${g}_1, \ldots, {g}_n$.
Then, we obtain the final estimate $\hat{g}_n$.
That is, we decide $\hat{g}_n$ as
\begin{align}
\hat{g}_n:= \argmax_{g \in \SU(2)} \sum_{i=1}^n \log p_{g}(g_i).
\end{align}
\end{description}
We denote the above measurement with the output $\hat{g}_n$ by $M_n$.
Then, due to the following theorem,
the above protocol asymptotically realizes 
the minimum error under the energy constraint. 

\begin{theorem}\Label{T5-5-1b}
Assume that $E_{\phi}>0$.
Then, the relations
\begin{align}
\lim_{n \to \infty} n 
{\cal D}_R(|\phi^{\otimes n}\rangle, M_n)
&= \frac{9}{8} E_\phi \Label{5-1-5b}\\
\langle\phi^{\otimes n}| H^{(n)} |\phi^{\otimes n}\rangle
&= E_{\phi} n \Label{5-1-6bx}
\end{align}
hold. That is,
\begin{align}
\lim_{n \to \infty} 
\langle\phi^{\otimes n}| H^{(n)} |\phi^{\otimes n}\rangle
{\cal D}_R(|\phi\rangle^{\otimes n}, M_n)
= \frac{9}{8}.
\end{align}
\end{theorem}

Since the Bell state $|\Psi_{\frac{1}{2}}\rangle$
satisfies the condition in Theorem \ref{T5-5-1b},
the optimal performance with energy constraint 
can be attained by using the Bell state $|\Psi_{\frac{1}{2}}\rangle$.
The above protocol with the Bell state $|\Psi_{\frac{1}{2}}\rangle$
does not requires use of entanglement among more than two qubits,
and requires only the entanglement between two qubits.

\begin{proof}
\eqref{5-1-6bx} can be shown by the same as (\ref{5-1-6}) in Theorem \ref{T5-5-1}.
Similar to the proof of Theorem \ref{T5-5-1}, 
under the parametrization $\varpi_{\vec{\theta}}$,
we can show that
the Fisher information matrix $J_{0}$ at $\vec{\theta}=0$ is
calculated as
\begin{align}
J^{s,t}_{0}=  \frac{4}{3}E_{\phi} \delta_{s,t}.
\Label{5-2-1b}
\end{align}
Hence, $(J_{\vec{\theta}}^{-1})_{i,j}= \frac{3 \delta_{i,j}}{4 E_{\phi}}$.
Remember that
the inverse of the Fisher information matrix gives the asymptotic mean square error.
Since $ R_{\SO(3)}(e,\varpi_{\vec{\theta}})=
1-\cos \|\vec{\theta} \|
\cong \frac{\|\vec{\theta} \|^2}{2}
=\frac{1}{2}((\theta^1)^2+(\theta^2)^2+(\theta^3)^2)$,
we have
\begin{align*}
n {\cal D}_R(|\phi^{\otimes n}\rangle, M_n)
\cong 
\frac{n}{2}
\rE_0 [ \sum_{i=1}^3 \hat{\theta}_{i,n}^2]
\to
\frac{1}{2}
 \sum_{i=1}^3
(J_{0}^{-1})_{i,i}
=
\frac{1}{2}
 \sum_{i=1}^3 
\frac{3 \delta_{i,i}}{4 E_{\phi}}
=
\frac{9}{8 E_{\phi}},
\end{align*}
where $\rE_{\vec{\theta}}$ expresses the expectation under the distribution 
$p_{\varpi_{\vec{\theta}}} $.
\end{proof}

\subsection{Application to eigenvalue estimation in qubit system}\Label{s11b-4}
In order to estimate of the eigenvalue of the given density $\rho$ on
the qubit system $\complex^2$,
we often consider the following irreducible decomposition as
\begin{align}
(\complex^2)^{\otimes 2m} 
&= \bigoplus_{l=0}^m
{\cal U}_{l} \otimes \complex^{d(2m,l)}, \\
(\complex^2)^{\otimes 2m+1} 
&= \bigoplus_{l=0}^m
{\cal U}_{l+\frac{1}{2}} \otimes \complex^{d(2m+1,l+\frac{1}{2})},
\end{align}
where 
\begin{align}
d(2m,l)
&:=\left\{
\begin{array}{ll}
{2m \choose (m-l)} -{2m \choose (m-l-1)} & \hbox{ if } 0 \le l \le m-1 \\ 
{2m \choose (m-l)} & \hbox{ if } l = m \\ 
0 & \hbox{otherwise}
\end{array}
\right. \\
d(2m+1,l+\frac{1}{2})
&:=\left\{
\begin{array}{ll}
{2m+1 \choose (m-l)} -{2m+1 \choose (m-l-1)} & \hbox{ if } 0 \le l \le m-1 \\ 
{2m+1 \choose (m-l)} & \hbox{ if } l = m \\ 
0 & \hbox{otherwise.}
\end{array}
\right.
\end{align}
Then, we define the projections $P_{l}^{(2m)}$ and $P_{l+\frac{1}{2}}^{(2m+1)}$ 
as the projections to
${\cal U}_{l} \otimes \complex^{d(2m,l)}$
and 
${\cal U}_{l+\frac{1}{2}} \otimes \complex^{d(2m+1,l+\frac{1}{2})}$. 
These projections form projection-valued measures.
When the initial state is given as $\rho^{\otimes n}$
and we apply the measurement corresponding to the projection-valued measure
$\{ P_{l}^{(2m)}\}$ or $\{ P_{l+\frac{1}{2}}^{(2m+1)}\}$,
the value $\frac{l}{2m}$ or $\frac{l+\frac{1}{2}}{2m+1}$ 
gives the estimate of the smaller eigenvalue of $\rho$ \cite[Appendix A]{qubit1}\cite{qubit2,qubit3}.
When the smaller eigenvalue $p$ of $\rho$ is less than $\frac{1}{2}$,
the error 
$\sqrt{2m}(\frac{l}{2m}- p)$ or $\sqrt{2m+1}(\frac{l+\frac{1}{2}}{2m+1}-p)$ 
asymptotically obeys
the Gaussian distribution with the average $0$ and the variance $p(1-p)$.
This fact can be shown by combining the local asymptotic normality \cite{qubit4,qubit5} and 
the achievement of the asymptotic minimum error bound by this measurement \cite{qubit2,qubit3}.
However, 
the asymptotic behavior of 
$\sqrt{2m}(\frac{l}{2m}- p)$ or $\sqrt{2m+1}(\frac{l+\frac{1}{2}}{2m+1}-p)$ is not known 
when $p$ is $\frac{1}{2}$, i.e.,
$\rho$ is the completely mixed state.
The stochastic asymptotic behaviors of 
$\frac{l}{\sqrt{2m}}$
and
$\frac{l+\frac{1}{2}}{\sqrt{2m+1}}$
can be given as follows.
\begin{align}
\Pr^{(n)}\{\frac{l}{\sqrt{2m}} \le y \}
&\to
\int_{0}^y
\frac{8}{\sqrt{2\pi}} 
\lambda^2 e^{-2{\lambda^2}}
d \lambda \hbox{ as }m \to \infty \Label{5-27-1} \\
\Pr^{(n)}\{\frac{l+\frac{1}{2}}{\sqrt{2m+1}} \le y \}
&\to
\int_{0}^y
\frac{8}{\sqrt{2\pi}} 
\lambda^2 e^{-2{\lambda^2}}
d \lambda \hbox{ as }m \to \infty .\Label{5-27-2}
\end{align}
That is,
the wave function corresponding to the single photon
gives the limiting behavior of the outcome of the measurement corresponding 
to the irreducible decomposition
when the true state is the completely mixed state.
This argument can be shown by the following generalized argument.

Consider the mixed state 
$\rho:= \sum_{k=0}^{\infty} p_{\frac{k}{2}} \rho_{\mix,\frac{k}{2}}$
on the system ${\cal H}:=\sum_{k=0}^{\infty} {\cal U}_{\frac{k}{2}}$,
where $\rho_{\mix,\frac{k}{2}}$ is the completely mixed state on
the system ${\cal U}_{\frac{k}{2}}$.
Then, we consider the tensor product state
$\rho^{\otimes n}$ on ${\cal H}^{\otimes n}$.
Considering the tensor product representation of $\SU(2)$,
we decompose the tensor product space 
${\cal H}^{\otimes n}$ to $\sum_{k=0}^{\infty} {\cal U}_{\frac{k}{2}} \otimes 
{\cal V}_{\frac{k}{2}}$,
where the group $\SU(2)$ acts only on ${\cal U}_{\frac{k}{2}}$.
Then, we can define the projection $P_{\frac{k}{2}}^{(n)} $ 
to ${\cal U}_{\frac{k}{2}} \otimes {\cal V}_{\frac{k}{2}}$.
That is, when the state is $\rho^{\otimes n}$
and we apply measurement $\{P_{\frac{k}{2}}^{(n)}\}_k$,
we obtain the outcome $k$ with the probability 
$p^{(n)}_{\frac{k}{2}}:=\Tr \rho^{\otimes n} P_{\frac{k}{2}}^{(n)}$.
Defining 
\begin{align}
E:= \sum_{k=0}^{\infty} \frac{k}{2}(\frac{k}{2}+1) p_{\frac{k}{2}},
\end{align}

The following theorem holds.
\begin{theorem}\Label{T6-4}
When $E>0$,
we have
\begin{align}
\Pr^{(n)} \{ \frac{k}{2\sqrt{n}}\le x\}
:=
\sum_{k=0}^{2x\sqrt{n}}p^{(n)}_{\frac{k}{2}}
\to& 
\int_0^x 
\frac{\sqrt{2\cdot 3^3}}{\sqrt{\pi E^3}} t^2 e^{-\frac{3t^2}{2E}} dt 
\nonumber \\
=&
\int_0^{3x^2/E} 
\frac{\sqrt{z} }{\sqrt{2\pi}} e^{-\frac{z}{2}} dz.
\Label{5-28-2}
\end{align}
The right hand side of \eqref{5-28-2} is called 
$\chi^2$-distribution with 3 degrees of freedom
or Gamma distribution.
\end{theorem} 

This theorem with $l= \frac{k}{2}$ and $l+\frac{1}{2}= \frac{k}{2}$
implies \eqref{5-27-1} and \eqref{5-27-2}, respectively.
In fact, this theorem can be regarded as 
an $\SU(2)$-version of the central limit theorem. 
When we consider the similar problem in the case of $\U(1)$,
we recover the conventional central limit theorem
because the tensor product gives the sum of weight in the case of $\U(1)$.
Further, this theorem has been shown in a more general framework
by using the concept ``hypergroup''\cite{Heyer}.
In particular, Theorem \ref{T6-4} corresponds to the case of 
Chebychev hypergroup of the second kind \cite[p.166]{Heyer}.
Here, we give another proof by using our result.

\begin{proof}
Define the pure state 
$|\phi \rangle:=
\oplus_{k=0}^{\infty} \frac{\beta_{\frac{k}{2}}}{\sqrt{k+1}}
|\Psi_{\frac{k}{2}}\rangle\rangle$
with $\beta_{\frac{k}{2}}= \sqrt{p_{\frac{k}{2}}} $.
Then, we choose the coefficients $\beta^{(n)}_{\frac{k}{2}}$
such that
$|\phi \rangle^{\otimes n}:=
\oplus_{k=0}^{\infty} \frac{\beta^{(n)}_{\frac{k}{2}}}{\sqrt{k+1}}
|\Psi_{\frac{k}{2}}\rangle\rangle$.
Hence, we obtain
$\beta^{(n)}_{\frac{k}{2}}= \sqrt{p^{(n)}_{\frac{k}{2}}}$.

In the following, we show the theorem with three separated cases.
First, we 
assume that
there exist at lest one even number $k_e \ge 0$ and one odd number $k_o >0$ 
such that $p_{\frac{k_e}{2}}>0$ and $p_{\frac{k_o}{2}}>0$.
Theorem \ref{T5-5-1} implies that
the state $|\phi \rangle^{\otimes n}$ satisfies the condition of Lemma \ref{T3-13-3y}.
Hence, 
\begin{align}
(2\pi E n)^{\frac{1}{4}}\sqrt{p^{(n)}_{\frac{\lceil 2 \sqrt{En} \lambda-\frac{1}{2}\rceil}{2}}}
\to 3^{\frac{3}{4}}\lambda e^{-\frac{3 \lambda^2}{4}}.
\end{align}
Choosing $t=\sqrt{E} \lambda$, 
we have
\begin{align}
\sqrt{n}
p^{(n)}_{\frac{\lceil 2 \sqrt{n} t-\frac{1}{2}\rceil}{2}}
\to 
\sqrt{\frac{3^{3}}{2\pi E^3}}t^2 e^{-\frac{3 t^2}{2 E}}.
\end{align}
Hence, considering $t=\frac{k}{2\sqrt{n}}$, we have
\begin{align}
\sum_{k=0}^{2x\sqrt{n}}p^{(n)}_{\frac{k}{2}}
=
2 \sum_{k=0}^{2x\sqrt{n}}
\frac{1}{2\sqrt{n}}
\sqrt{n}
p^{(n)}_{(\frac{k}{2\sqrt{n}})\sqrt{n}}
\to 
\int_0^x 
\frac{\sqrt{2\cdot 3^{3}}}{\sqrt{\pi E^3}} t^2 e^{-\frac{3t^2}{2E}} dt.
\end{align}

Next, we assume that 
$p_{\frac{k}{2}}=0$ for all odd numbers $k$.
Theorem \ref{T5-5-1b} implies that
the state $|\phi \rangle^{\otimes n}$ satisfies the condition of [Case 1] of Lemma 
\ref{T3-13-3y2}.
Hence, 
\begin{align}
(\frac{\pi E n}{2})^{\frac{1}{4}}
\sqrt{p^{(n)}_{\lceil \sqrt{En} \lambda-\frac{1}{2}\rceil}}
\to 3^{\frac{3}{4}}\lambda e^{-\frac{3 \lambda^2}{4}} ,\quad 
p^{(n)}_{l+\frac{1}{2}}=0
\end{align}
Choosing $t=\sqrt{E} \lambda$, 
we have
\begin{align}
\sqrt{n}
p^{(n)}_{\lceil \sqrt{n} t-\frac{1}{2}\rceil}
\to 
\frac{2 \cdot 3^{\frac{3}{2}}}{\pi E^3}t^2 e^{-\frac{3 t^2}{2 E}}.
\end{align}
Hence, considering $t=\frac{k}{\sqrt{n}}$, we have
\begin{align}
\sum_{k=0}^{2x\sqrt{n}}p^{(n)}_{\frac{k}{2}}
=
2 \sum_{k=0}^{2x\sqrt{n}}
\frac{1}{2\sqrt{n}}
\sqrt{n}
p^{(n)}_{(\frac{k}{2\sqrt{n}})\sqrt{n}}
\to 
\int_0^x 
\frac{\sqrt{2}}{\sqrt{\pi E^3}} t^2 e^{-\frac{3t^2}{2E}} dt. \Label{5-28-1}
\end{align}

Finally, we consider the case when
$p_{\frac{k}{2}}=0$ for all even numbers $k$.
Theorems \ref{T5-5-1b} implies that
the state $|\phi \rangle^{\otimes n}$ satisfies the condition of [Case 1] of Lemma 
\ref{T3-13-3y2} for an even $n$,
and
the state $|\phi \rangle^{\otimes n}$ satisfies the condition of [Case 2] of Lemma 
\ref{T3-13-3y2} for an odd $n$.
\begin{align}
(\frac{\pi E n}{2})^{\frac{1}{4}}
\sqrt{p^{(n)}_{\lceil \sqrt{En} \lambda-\frac{1}{2}\rceil+\frac{1}{2}}}
\to 3^{\frac{3}{4}}\lambda e^{-\frac{3 \lambda^2}{4}} ,\quad 
p^{(n)}_{l}=0.
\end{align}
Hence, similar to \eqref{5-28-1},
we can show \eqref{5-28-2}. 
\end{proof}

\section{Heisenberg representation of $\real^2$}\Label{s12}
As a typical example of 
non-commutative representation of a non-compact group,
we treat the Heisenberg representation of $\real^2$.
Then, we fix the factor system ${\cal L}$
defined by the Heisenberg representation.
In this case, the representation space is $L^2(\real)$
and we allow to use the multiplicity space $L^2(\real)^*$.
Then, the inverse Fourier transform ${\cal F}_{{\cal L}}^{-1}$
with the equivalent relation from 
the input system $L^2(\real) \otimes L^2(\real)^*$ to 
$L^2(\real)^{\otimes 2}$.
We employ the operators 
$Q_1=Q \otimes I$, $Q_2=I \otimes Q$, $P_1=P \otimes I$, and $P_2=I \otimes P$
in the latter system $L^2(\real)^{\otimes 2}$.
Now, we focus on the average of the square error 
\begin{align}
\int_{\real^2} (\hat{x}_1 - x_1)^2+(\hat{x}_2 - x_2)^2
\Tr f(\zeta) \rho f(\zeta)^{\dagger} M(d \hat{\zeta}) ,
\end{align}
where
$\zeta=\frac{x_1+ix_2}{\sqrt{2}}$
when the input state is $\rho$ and the estimator is $M$.
When the input state $\rho$ is a pure state $\phi$
and the estimator $M$ is $M_{|{\cal I}\rangle \langle {\cal I}|}$,
the average of the square error 
is simplified to
\begin{align}
\int_{\real^2}
(x_1^2+x_2^2)
|{\cal F}^{-1}_{{\cal L}}[\phi](-\zeta)|^2
dx_1 dx_2 
= 
\langle \varphi | Q_1^2+Q_2^2 |\varphi\rangle 
\Label{6-24-8},
\end{align}
where 
$\varphi:={\cal F}^{-1}_{{\cal L}}[\phi]$.
Now, we consider the energy constraint as follows.
\begin{align}
\langle \phi| (Q^2+P^2)\otimes I |\phi\rangle 
\le E,\Label{6-24-10}
\end{align}
which can be rewritten as
\begin{align}
\langle \varphi | (P_2-\frac{1}{2}Q_1)^2+(-P_1-\frac{1}{2}Q_2)^2 |\varphi\rangle 
\le E.
\end{align}
Now, we apply the unitary transformation $U$
corresponding to the the following element of $\Sp(4,\real)$:
\begin{align*}
\left(
\begin{array}{cccc}
1 & 0 & 0 & 0\\
0 & 1 & 0 & 0\\
0 & -\frac{1}{2} & 1 & 0 \\
\frac{1}{2} & 0  & 0 & 1
\end{array}
\right)
\end{align*}
Then,
we can convert the above problem to the following:
We minimize
\begin{align}
\langle \varphi | U (Q_1^2+Q_2^2)U^\dagger |\varphi\rangle 
\Label{9-01}
\end{align}
under the condition
\begin{align}
\langle \varphi | U(P_1^2+P_2^2)U^\dagger |\varphi\rangle 
\le E.
\end{align}
This minimization problem can be solved by the combination of
the minimization problems
$\min \{
\langle \varphi | U Q_j^2 U^\dagger |\varphi\rangle 
|
\langle \varphi | U P_j^2 U^\dagger |\varphi\rangle 
\le E/2 \}= \frac{1}{2E}$ with $j=1,2$.
Then, 
the minimum value of \eqref{9-01} is $\frac{1}{E}$, which 
can be attained when 
$U^\dagger |\varphi\rangle$ is 
$\sqrt{E} e^{\frac{E}{2}(x_1^2+x_2^2)}$.
Thus,
\begin{align}
\min_{|\phi \rangle \in L^2_n(\real)} 
\{
{\cal D}_{R}(|\phi \rangle)
| \langle \phi | P^2+Q^2 |\phi \rangle \le E \} 
=\frac{1}{E}
\Label{e4-23b}.
\end{align}

Applying Theorem \ref{t6-24-1} to the above discussion, we obtain the following theorem.
\begin{theorem}\Label{Te4-23}
The relations 
\begin{align}
&
\min_{\rho \in {\cal S}(L^2(\real))} 
\min_{M \in {\cal M}_{\cov}(G)} 
\{{\cal D}_R(\rho,M)| \Tr \rho (P^2+Q^2) \le E \}
\nonumber \\
=&
\min_{ \{p_i\} }
\min_{\rho_i \in {\cal S}(L^2(\real))} 
\min_{M_i \in {\cal M}_{\cov}(G)} 
\{\sum_i p_i {\cal D}_{R}(\rho_i,M_i) | \sum_i p_i \Tr \rho_i (P^2+Q^2) \le E\} 
\nonumber \\
=&
\frac{1}{E}
\Label{e4-23}
\end{align}
hold.
\end{theorem}

Due to the construction, 
the outcome of the optimal estimator obeys 
the Gaussian distribution with the variance $\frac{1}{2E}$ and the average $(\theta_1,\theta_2)$
when the true parameter is $(\theta_1,\theta_2)$.

Now, we consider 
two systems ${\cal H}_i$ ($i=1,2$) equivalent with $L^2(\real)$ with the Hamiltonian $Q^2+P^2$.
We focus on the composite system ${\cal H}_1 \otimes {\cal H}_2$
with the Hamiltonian
$(Q \otimes I+ I \otimes Q)^2 +(P \otimes I+ I \otimes P)^2 
=
(Q^2+P^2) \otimes I+ 2 (Q \otimes Q+ P \otimes P) +I \otimes (Q^2+P^2)$,
which has a strong interaction term $2 (Q \otimes Q+ P \otimes P)$.
In this case, 
the optimal estimation 
in the composite system with the energy $E_1+E_2$
can be realized by the following way.
Let the input state $|\phi_i\rangle$ be the optimal input state with the energy $E_i$.
Then, due to the construction of $|\phi_i\rangle$ given above,
$\langle \phi_i|Q |\phi_i\rangle= \langle \phi_i|P |\phi_i\rangle=0$.
The input state $
|\phi_1\otimes \phi_2\rangle=
|\phi_1\rangle \otimes |\phi_2\rangle$ has the energy $E_1+E_2$
because
\begin{align*}
&\langle \phi_1\otimes \phi_2|
(P \otimes I+ I \otimes P)^2
+(Q \otimes I+ I \otimes Q)^2
 |\phi_1\otimes \phi_2\rangle \\
=&
\langle \phi_1\otimes \phi_2| (Q^2+P^2) \otimes I |\phi_1\otimes \phi_2\rangle
+\langle \phi_1\otimes \phi_2|I \otimes (Q^2+P^2) |\phi_1\otimes \phi_2\rangle \\
&+2 \langle \phi_1\otimes \phi_2|Q \otimes Q +P \otimes P|\phi_1\otimes \phi_2\rangle \\
=&
\langle \phi_1|Q^2+P^2 |\phi_1\rangle
+\langle \phi_2|Q^2+P^2 |\phi_2\rangle \\
&+2 \langle \phi_1|Q |\phi_1\rangle\langle \phi_2|Q |\phi_2\rangle
+2 \langle \phi_1|P |\phi_1\rangle\langle \phi_2|P |\phi_2\rangle \\
=& E_1+E_2.
\end{align*}
Since the outcomes of each optimal estimation in the subsystems ${\cal H}_i$
obey the Gaussian distribution,
the state $|\phi_1\otimes \phi_2\rangle$
realizes the optimal estimator in the composite system ${\cal H}_1 \otimes {\cal H}_2$ the energy $E_1+E_2$
by constructing the measurement in the same way as the end of Subsection \ref{s8-1}.
That is, 
we can realize the optimal estimator by the combination of the optimal estimators of the individual systems.

\section{Conclusion}\Label{s13}
We have shown two general formulas for the minimum error in the estimation of group action
based on the inverse Fourier transform of the input state.
One gives the minimum error without energy constraint,
and the other gives the minimum error with energy constraint.
Using the obtained former formula,
we have derived several known formulas, i.e., the maximum discrimination formula in the finite group case
and the minimum error formula for the compact group.
In fact, the obtained latter formula is essential for 
the estimation of action of the non-compact group
because many of their non-commutative projective representations 
are infinite-dimensional.
Then, we have succeeded in the calculations of the minimum error in the case of
$\bR$ with two types of energy constraints.
Applying the result with the energy constraints,
we have succeeded in the asymptotic calculations of the minimum error in the case of $\U(1)$ with two types of energy constraints.
Further, 
applying the result of $\U(1)$ with the energy constraint,
we have succeeded in the asymptotic calculations of the minimum error in the case of $\SU(2)$ with the energy constraints for total angular momentum.
Finally, we apply our formula with energy constraint 
to the Heisenberg representation.


Next, we discuss the reasonability of the square speed up in the estimation of unitary.
In all of the above examples,
when we consider the energy constraint $\Tr \rho H \le E$,
the minimum error asymptotically behaves as $\frac{c}{E}$
not $\frac{c}{E^2}$.
This fact implies that
there is no square speed up under the energy constraint.
However, we have square speed up 
under the interval constraint for $\bR$ and $\U(1)$ and 
the constraint for the number of tensors for $\SO(3)$.
In these cases,
the average energy of the input states increases with the order of square of the size of the constraint.
In the realistic setting, we have to consider the average energy as the cost
even though we are interested in the length of interval of the weight range or the number of tensor products.
In such a case, 
the energy constraint gives a more restrictive constraint than 
the constraint of the width of the weight range or the number of tensor products
when larger sizes in both constraints are available.
That is, the energy constraint is dominant. 
So, we essentially have no square speed up.

This observation may be extended to any other compact groups 
while it is known that the square speed up phenomena happens 
with respect to the number of tensor products in $\SU(d)$ \cite{Kahn}.
This is because the minimum error behaves as $\frac{c}{E}$ not $\frac{c}{E^2}$ in the estimation of $\SU(d)$
when we consider an energy constraint $\Tr \rho H \le E$ 
and the Hamiltonian $H$ is given by the Casimir element
because $\SU(d)$ contains the $\U(1)$ as a subgroup.

We have also given a practical construction of the asymptotically optimal estimator for 
$\U(1)$, $\SO(3)$, and $\SU(2)$ as follows.
In the estimation of $\U(1)$, 
in Subsection \ref{s10-3},
we have shown that the asymptotically optimal estimation
with the energy constraint
can be realized by the repetition of 
the estimation of $\U(1)$ by using the single qubit system.
That is, the optimal performance can asymptotically be attained by 
the maximum likelihood estimator based of the outcomes subject to 
the independent and identical distribution given by the single qubit system.
Hence, such an optimal performance can be easily realized.
The similar fact also holds in the estimation of $\SO(3)$ and $\SU(2)$.
In the case of $\SO(3)$, 
as has been shown in Subsection \ref{s11-3},
the asymptotically optimal estimator can be realized as follows.
First, we input the Bell state, in which the group $\SO(3)$ acts only on the first qubit and
the second qubit works as the reference system.
Then, we apply the covariant measurement on the total system.
We repeat this process and apply the maximum likelihood estimator to the obtained data.
A similar fact has been shown for $\SU(2)$.
However, in the case of $\SU(2)$, we need to prepare a superposition input state
of maximally entangled states on irreducible representations
with an integer weight and a half integer weight.
This is because 
estimation of $\SU(2)$ requires to distinguish the two elements of $\SU(2)$
corresponding to the same element of $\SO(3)$.

We have also shown a similar fact for $\bR$ and $\bR^2$ with the Heisenberg representation.
It was been shown that 
we can realize the optimal estimator by a linear combination of the optimal estimators of the individual systems
${\cal H}_1$ and ${\cal H}_2$
under the energy constraint for the estimation of $\bR$ and $\bR^2$ with Heisenberg representation.
In these cases,
any input state entangled between subsystems ${\cal H}_1$ and ${\cal H}_2$
is not required for the optimal estimation.

\section*{Acknowledgments}
The author is grateful for Professor Hideyuki Ishi to explaining the role of Type I group in the Plancherel Theorem and informing the references \cite{Fuhr,Folland}.
He also grateful for Professor Akihito Hora for informing 
the concept ``hypergroup'' and
the reference \cite{Heyer}.
The author 
is partially supported by a MEXT 
Grant-in-Aid for Scientific Research (A) No. 23246071.
The Centre for Quantum Technologies is funded by the
Singapore Ministry of Education and the National Research Foundation
as part of the Research Centres of Excellence programme.

\appendix 

\section{Periodic function space and Mathieu equation}\Label{asB}
In order to treat the space of periodic function.
for a positive real number $L$,
we introduce the notations as follows.
\begin{align}
L^2_{p}((- L,  L]):= 
\{f| f(x+2L)=f(x) , \int_{-L}^L |f(x)|^2 \frac{dx}{2L}<\infty\} .
\end{align}
As a generalization, 
we define the space of anti-periodic functions 
\begin{align}
L^2_{a}((- L,  L]):= 
\{f|  f(x+2L)=- f(x) , 
\int_{-L}^L |f(x)|^2 \frac{dx}{2L}<\infty\} ,
\end{align}
which is a subspace of $L^2_{p}((- 2L,  2L])$.
Further, we denote the spaces of even functions and odd functions
in $L^2_{p}((- L,  L])$ and $L^2_{a}((- L,  L])$ 
by $L^2_{p,\even}((- L,  L])$, $L^2_{p,\odd}((- L,  L])$, 
$L^2_{a,\even}((- L,  L])$, $L^2_{a,\odd}((- L,  L])$, respectively. 
For any $f,g \in L^2_{p}((- L,  L])$, we define the 
inner product as
\begin{align}
\langle f|g\rangle:=
\int_{-L}^L \overline{f(x)} g(x)\frac{dx}{2L}.
\end{align}
The subspaces $L^2_{p,\even}((- L,  L])$ and $L^2_{p,\odd}((- L,  L])$
(the subspaces $L^2_{a,\even}((- L,  L])$ and $L^2_{a,\odd}((- L,  L])$) 
are orthogonal to each other.
Also, the two subspaces $L^2_{p}((- L,  L])$ and $L^2_{a}((- L,  L])$ are orthogonal to each other.
Therefore, the space $L^2_{p}((- 2L,  2L])$ can be written as
$L^2_{p,\even}((- L,  L]) \oplus L^2_{p,\odd}((- L,  L])
\oplus L^2_{a,\even}((- L,  L]) \oplus L^2_{a,\odd}((- L,  L])$.

Now, we consider Mathieu equation:
\begin{align}
\frac{d^2}{d\theta^2} \varphi(\theta)+ ( a -2q \cos (2\theta))
\varphi(\theta)=0.
\Label{5-10-7}
\end{align}
A function $\varphi$ satisfies the above equation if and only if
the function $\varphi$ is the eigenfunction of 
the differential operator $P^2 +2 q \cos (2Q)$.
The operator $X(q):=P^2 +2 q \cos (2Q)$
preserves the subspaces
$L^2_{p,\even}((- \frac{\pi}{2}, \frac{\pi}{2}])$,
$L^2_{p,\odd}((- \frac{\pi}{2}, \frac{\pi}{2}])$,
$L^2_{a,\even}((- \frac{\pi}{2}, \frac{\pi}{2}])$,
and
$L^2_{a,\odd}((- \frac{\pi}{2}, \frac{\pi}{2}])$.
Then, we denote the minimum eigenvalues 
in 
$L^2_{p,\even}((- \frac{\pi}{2}, \frac{\pi}{2}])$,
$L^2_{p,\odd}((- \frac{\pi}{2}, \frac{\pi}{2}])$,
$L^2_{a,\even}((- \frac{\pi}{2}, \frac{\pi}{2}])$,
and
$L^2_{a,\odd}((- \frac{\pi}{2}, \frac{\pi}{2}])$
by $a_0(q)$, $b_2(q)$, $a_1(q)$, and $b_1(q)$,
respectively \cite[Section 28.2]{Mathieu}.
We call their eigenfunctions
Mathieu functions 
$\ce_0(\theta,q)$,
$\se_2(\theta,q)$,
$\ce_1(\theta,q)$, and
$\se_1(\theta,q)$.
The eigenvalues $a_0(q)$, $b_2(q)$, $a_1(q)$, and $b_1(q)$
satisfy the conditions 
$a_0(q)=a_0(-q)$, $a_1(-q)=b_1(q)$, and $b_2(-q)=b_2(q)$.
When $q <0$, the ordering relation
$a_0(q) < a_1(q) < b_1(q)< b_2(q)$ holds.


According to the reference \cite[Section 28.2(v)]{Mathieu}, let 
$a_0(q)$ be the minimum $a$ having the solution in $L^2((-\pi/2,\pi/2 ])$ of the above differential equation,
and $b_2(q)$ be the minimum $a$ having the odd solution in $L^2((-\pi/2,\pi/2 ])$ of the above differential equation.
The solution with $a_0(q)$ is 
Mathieu function $\ce_0(\theta,q)$
and
the solution with $b_2(q)$ is Mathieu function $\se_2(\theta,q)$ \cite[Section 28.2(vi)]{Mathieu}.
These values satisfies that \cite[Section 28.2(v)]{Mathieu}
\begin{align}
a_0(q)=a_0(-q), \quad
b_2(q)=b_2(-q), \quad
a_1(q)=b_1(-q). \Label{6-3-1}
\end{align}
Further, for a large $q$, the functions $a_0$, $a_1$, and $b_2$
have the following asymptotic expansions for a large $h$ as 
\begin{align}
a_0(h^2) &\cong
- 2h^2+2h -\frac{1}{4}-\frac{1}{2^5 h}
-\frac{3}{2^8 h^2} \Label{exp1} \\
a_1(h^2) &\cong b_2(h^2)\cong
- 2h^2+6h -\frac{5}{4}-\frac{9}{2^5 h}
-\frac{45}{2^8 h^2}. \Label{5-20-16b}
\end{align}
Further expansion is available in \cite[Section 28.8]{Mathieu}.

For a small $q$, the functions $a_0(q)$, $a_1(q)$, and $b_2(q)$ 
have the following asymptotic expansions as 
\begin{align}
a_0(q)& \cong -\frac{1}{2}q^2+\frac{7}{128}q^4 \Label{exp2} \\
a_1(q)& \cong 1+ q -\frac{1}{8}q^2 -\frac{1}{64}q^3-\frac{1}{1536}q^4 \Label{5-26-11} \\
b_2(q)& \cong 4-\frac{1}{12}q^2+\frac{5}{13824}q^4 \Label{exp3}.
\end{align}
Further expansion is available in \cite[Section 28.6]{Mathieu}.

\section{Technical lemma for operators}\Label{a1}
We show an important technical lemma.
For a given Hilbert space ${\cal H}$,
we consider two self-adjoint operators $Y$ and $Z$
and a two-dimensional subspace ${\cal V}$ of ${\cal H}$.
Then, we have the following lemma.
\begin{lemma}\Label{L4-25}
\begin{align*}
\min_{\rho \in {\cal S}({\cal V})}
\{\Tr \rho Y|\Tr \rho Z \le E\}
=
\min_{\phi \in {\cal V}}
\{\langle \phi| Y|\phi\rangle  |\langle \phi| Z|\phi\rangle \le E ,
\|\phi\|=1\}.
\end{align*}
If there is no element satisfying the condition,
we consider that the above minimums are infinity. 
\end{lemma}
\begin{proof}
It is enough to show 
\begin{align}
\min_{\rho \in {\cal S}({\cal V})}
\{\Tr \rho Y|\Tr \rho Z = E\}
=
\min_{\phi \in {\cal V}}
\{\langle \phi| Y|\phi\rangle  |\langle \phi| Z|\phi\rangle = E ,
\|\phi\|=1\}.\Label{4-25-4}
\end{align}
In this case, we can consider $Y$ and $Z$ as two-dimensional Hermitian matrixes.
Then, $Z$ can be diagonalized to $z_0|u_0\rangle \langle u_0|+z_1|u_1\rangle \langle u_1|$.
When $z_0=z_1$, the above equation is trivial.
So, we assume that $z_0< z_1$ and there exists a density operator $\rho$
satisfying the condition.
Then, there exists $p\in [0,1]$ such that
$p y_0+(1-p)y_1=E$.
Then, when a density operator $\rho$ satisfies $\Tr \rho Z=E$,
$\rho$ can be written as $\sum_l q_l |v_l\rangle \langle v_l|$,
where $\{q_i\}$ is a distribution and
$v_l= \sqrt{p}u_0 + e^{i\theta_l}\sqrt{1-p}u_1$.
Hence, we obtain
\begin{align*}
& \Tr Y =
\sum_l q_l \langle v_l |Z|v_l \rangle 
\ge
\min_l \langle v_l |Z|v_l \rangle \\
\ge &
\min_{\phi \in {\cal V}}
\{\langle \phi| Y|\phi\rangle  |\langle \phi| Z|\phi\rangle = E ,
\|\phi\|=1\},
\end{align*}
which implies (\ref{4-25-4}).
\end{proof}

\section{Diagonalization of matrix}\Label{a5-18}
\begin{lemma}\Label{t6-10-1}\cite{CS,OH}
The operator
$P_m:=
\sum_{k=1}^{m-1} |k \rangle \langle k+1|+|k+1 \rangle \langle k|$
has eigenvalues $2\cos \frac{j\pi}{m+1}$ $(j=1, \ldots, m)$
with the eigenvectors
$x^j:=\sum_{k=0}^{m}  \sin \frac{j k \pi}{m+1} |k\rangle$.
\end{lemma}

Now, we consider the case when $m$ is an even number $2l$.
Then, we change the basis with the correspondence 
$|k \rangle \to |k-l-\frac{1}{2}\rangle$.
The matrix 
$P_m $ is rewritten as
$\sum_{k=-l+\frac{1}{2}}^{l-\frac{3}{2}} |k \rangle \langle k+1|+|k+1 \rangle \langle k|$.
Now, we decompose 
the space $V_l$ spanned by the basis $\{|k\rangle\}$
as follows.
\begin{align}
V_l&=V_{l,\even} \oplus V_{l,\odd} \\
V_{l,\even}&:=\{ \sum_{k}a_k |k\rangle| a_{-k}=a_k \} \\
V_{l,\odd}&:=\{ \sum_{k}a_k |k\rangle| a_{-k}=-a_k \} 
\end{align}
The operator $P_{2l}$ preserves $V_{l,\even}$ and $V_{l,\odd}$.
The space $V_{l,\odd}$ is spanned by 
$|u_k\rangle:=
\frac{1}{\sqrt{2}}(|k -\frac{1}{2}\rangle -|-k +\frac{1}{2}\rangle )$
with $k=1, \ldots, l$.
On the space $V_{l,\odd}$, the operator $P_{2l}$ is written as
\begin{align}
P_{2l}:= -|u_1 \rangle \langle u_1|+\sum_{k=1}^{l-1}(|u_k \rangle \langle u_{k+1}|+|u_{k+1} \rangle \langle u_k|).
\Label{5-26-1}
\end{align}

Due to Lemma \ref{t6-10-1}, 
on the space $V_{l,\odd}$, 
the operator $P_{2l}$ has the eigenvalues $2 \cos \frac{2t \pi}{2l+1}$ with $t=1, \ldots, l$.
The eigenvector associated with the eigenvalues $2 \cos \frac{2t \pi}{2l+1}$
is $|v_{2t}\rangle=
\sum_{k=-l+\frac{1}{2}}^{l-\frac{1}{2}} \sin \frac{2tk \pi}{2l+1}|k \rangle
=\sqrt{2} \sum_{k=1}^l \sin \frac{2t(k-\frac{1}{2})}{2l+1} |u_k\rangle$.

\end{document}